\documentclass[a4paper,11pt]{amsart}
\usepackage{amsmath}
\usepackage{amssymb}
\usepackage{amsthm}
\usepackage{amsfonts}
\usepackage{color}
\usepackage{graphicx}
\usepackage{graphicx,epstopdf}
\usepackage{mathtools}
\usepackage{mathrsfs}
\usepackage{multirow}
\usepackage{rotating}
\usepackage{enumerate}
\usepackage{epsfig}
\usepackage{setspace}
\usepackage{ragged2e}
\usepackage{blindtext}
\usepackage{lineno}
\usepackage{url}
\usepackage{caption}
\usepackage{subcaption}
\usepackage{physics}
\usepackage{lscape}
\usepackage{float}
\usepackage[colorlinks=true,linkcolor=red,citecolor=blue]{hyperref}%
\usepackage{biblatex}
\bibliography{refs}

\newtheorem{lemma}{Lemma}
\newtheorem{thm}{Theorem}

\newtheorem{cor}{Corollary}

\theoremstyle{definition}

\newtheorem{rem}{Remark}

\usepackage{geometry}
\numberwithin{equation}{section}
\begin{document}

\title[ additional food ] 
    {T(w)o patch or not t(w)o patch: A novel biocontrol model}

 \keywords{infinite time blow up, patch model, biological control, additional food, chaotic dynamics}



\maketitle

\centerline{\scshape  Urvashi Verma$^{1}$, Aniket Banerjee$^{1}$ and  Rana D. Parshad$^{1}$}
\medskip
{\footnotesize

   \medskip
   
    \centerline{ 1) Department of Mathematics,}
 \centerline{Iowa State University,}
   \centerline{Ames, IA 50011, USA.}

 }

\begin{abstract}


A number of top down bio-control models have been proposed where the introduced predators' efficacy is enhanced via the provision of additional food (AF). 
However, if the predator has a pest dependent monotone functional response, pest extinction is unattainable. In the current manuscript, we propose a model where a predator with pest dependent monotone functional response is introduced into a ``patch" such as a prairie strip with AF, and then disperses or drifts into a neighboring ``patch" such as a crop field, to target a pest. We show the pest extinction state is attainable in the crop field and can be globally attracting. 
  The AF model with patch structure can eliminate predator explosion present therein and can keep pest densities lower than the classical top-down bio-control model. We provide the first proof of the global stability of the interior equilibrium for the classical AF model.
  We also observe ``patch-specific chaos" - the pest occupying the crop field can oscillate chaotically, while the pest in the prairie strip oscillates periodically. We discuss these results in light of bio-control strategies that utilize state-of-the-art farming practices such as prairie strips and drift and dispersal pressures driven by climate change.

\end{abstract}

\section{Introduction}

\noindent Invasive/pest species continue to be a significant global problem \cite{PZM05, PSCEWT16, S17}, of immense scale. Pests such as the European corn borer, Western corn rootworm, and soybean aphid cause damages to crop yield in the United States in excess of $\$$ 3 billion annually \cite{TS14, CO18, BOP22}. Populations of the invasive cane toad in Western Australia are reaching explosive levels \cite{U08}, despite several intervention strategies, while the Burmese python invasion in the Everglades region of south Florida has caused an excessive drop in local prey populations \cite{D12}.
Thus, the control of invasive species is crucial but also challenging \cite{S17}. Typical control methods use chemical pesticides despite negative environmental impacts, \cite{PB14}. Note the Kunming-Montreal Global Biodiversity Framework (GBF), which has been adopted by the UN Convention on Biological Diversity, mandates that by 2030, we must aim to reduce the risk from pesticides by at least 50 $\%$ \cite{UNC}.
To this end, a self-sustaining and environmentally friendly alternative is classical or top down bio-control. This is the introduction of natural enemies of the pest species \cite{V96, S99}, which will, in turn, depredate the pest. Bio-control comes with associated risks \cite{SV06}. These include non-target effects, \cite{PQB16}, which could lead to uncontrolled growth of the predator population. This has been seen with the cane toad in Western Australia and the Burmese python in South Florida \cite{U08, D12}. In order to help support bio-control programs, several innovative invasion control strategies using niche theory and competition theory have been proposed \cite{kettenring2011lessons, britton2018trophic}. For example, lowering nitrogen levels in soil via microbe use has been used as a management strategy to enhance intraspecies competition among native plants and preempt invasions \cite{vasquez2008creating}. Among these is also the method of Additional/Alternative food (AF). Herein, the efficacy of the predator can be ``boosted" by supplementing it with an AF source \cite{SV06, T15}. A number of mathematical models that describe predator pest dynamics with AF have been developed, \cite{SP07, SP10, S02, SP11, SPV18, SPD17, VA22}. These works posit that if high-quality AF is provided to the introduced predator in sufficient quantity, then pest extinction is achieved. 

The classical single species population model takes the form \cite{M03},

\begin{equation}
u_{t} = d u_{xx} + c u_{x} + f(u)
\end{equation}

here a population density of a species $u$ disperses via ``$u_{xx}$" over a one-dimensional landscape, is drifted via ``$u_{x}$," with $f(u)$ modeling demographic factors such as birth, death, and intraspecific competition. A most classical form in the literature is $f(u)=b u(1-\frac{u}{K})$, where $b$ is a birth rate and $K$ is the carrying capacity of species $u$. Here $b, K$ are pure constants. In the continuous time or ODE analog, one breaks up the landscape into two patches, say $\Omega_{1}, \Omega_{2}$ and approximates dispersal, $d u_{xx} \approx d(u_{1}-u_{2})$. Similarly one approximates drift, $c u_{x} \approx c(u_{1}-u_{2})$, where $u_{i}$ is the sub populations of $u$ in patch $\Omega_{i}$. Drift and dispersal are key ingredients governing species movements. Although there is a large body of work on AF mediated bio-control  \cite{SP07, SP10, SP11, SPV18, SPD17}, the effect of dispersal and drift on these models has not been well investigated. 

Dispersal is a key ecological mechanism that shapes communities \cite{M03, D99}. Species disperse in search of food, mates, resources, and even refuge \cite{CL03}. As human action continues to fragment natural habitat, more and more habitat ``patches" appear, and species subsequently will often disperse between these small interactive spatial components or ``patches." To this end, patch models have been intensely investigated \cite{W95, T11}. Thus, dispersal between these patches and its consequences on shaping ecosystem dynamics is of strong current relevance. The predator-prey model, with the dispersal of predator and prey between patches, has been intensely investigated. These include predator-only vs. prey-only dispersal, dispersal of predators (and prey) dependent on the strength of depredation, and dispersal mediated via Allee effects, which have all been thoroughly studied \cite{arditi2018asymmetric,messan2015two}.
Interestingly, the dynamics that one sees in the classic mean field sense, in predator-prey or competitive scenarios (where there is essentially one large patch) may not be true in the two (or multi) patch setting, with dispersal between these patches \cite{GK05, BSC10}. Another important mechanism of species movement is drift. This is mostly seen in aquatic systems, such as river or stream flows, where the flow of water will determine the direction of movement of species in the water body, typically from an upstream location to a downstream location. Thus, recent work has also focused on the drift in river/stream communities with a network structure and dispersal in eco-epidemic systems \cite{chen2024evolution, chen2023evolution, salako2024degenerate}. There is also interest in drift due to other vectors such as wind \cite{moser2009interacting}. Drift between patches of crop fields could increase in the foreseeable future due to climatic changes. Climate change is expected to increase flooding in the North Central United States over the next several years \cite{lee2023ipcc} - leading to stronger drift between patches.

Climate change is potentially as large a threat to biodiversity as habitat degradation \cite{malhi2020climate}. Current research focuses on range shifts following local colonization and/or extinction of species. If the environment changes, species will have to adapt, disperse, or otherwise will inadvertently be driven to extinction \cite{early2014climatic}. This becomes even more intriguing from an invasive species standpoint when climatic conditions turn favorable for an invader and unfavorable for the resident \cite{hellmann2008five}. A strong motivation for our current work is particularly driven by climate change effects on the population dynamics of invasive pests such as the soybean aphid \textit{Aphis glycines} (Hemiptera: Aphididae). This pest was first detected in the US in 2000 and since has become the chief insect pest of soybean in the major production areas of the north-central US \cite{tilmon2011biology}. It has a heteroecious holocyclic life cycle that utilizes a primary host plant for overwintering and a secondary host plant during the summer. As aforementioned, climate change is expected to increase flooding in the North Central United States \cite{lee2023ipcc}. Current research shows that increased rainfall is unlikely to cause the aphid pest to fall off the soybean plant, and survivability is greatly reduced if they do \cite{L1,example_conference}. However, ground predators (and biocontrol predators of the aphid, moving between soybean plants) would be subject to drift - particularly in flood scenarios and particularly if the tilt of land was from an AF patch (at higher ground) towards a crop field (at lower ground).

There is also a growing body of evidence that natural enemies of pests are many times more abundant in smaller patches surrounding crop fields and can depend on the connectivity of crop fields to non-crop habitats \cite{HL20}. Furthermore, multi-trophic diversity can be facilitated by increasing crop heterogeneity as well \cite{FG15, SF19}. From a pragmatic viewpoint, the utilization of these theoretical results to achieve target control scenarios of species requires the management of such landscapes. This is crucial for improved ecosystem structure and function \cite{KM18}, and current research strongly points towards the benefits of combining bio-control and landscape management \cite{klinnert2024landscape}. This has been coined ``landscape features supporting natural pest control" (LF-NPC). Among these landscape management practices are several innovative strategies, including the planting of prairie strips. These are small portions (about 10-20 $\%$) of the crop field that are removed from production for the primary crop and instead planted with several other plants, yielding conservation benefits for the entire crop field. This is essentially a conservation practice to protect soil and water while providing suitable habitat for wildlife. These enable increased productivity, which leads to higher crop yields, reduces sediment movement, and maintains high phosphorus and nitrogen levels. 

All in all, this is a high-value practice in agriculture \cite{LS17}. Recent innovation herein has led to programs such as STRIPS (Science-based Trails of Row crops Integrated with Prairie Strips) in the mid-western US. This has been pioneered and led in the state of Iowa, in particular, \cite{S16}. The row crops as earlier mentioned, provide a suitable habitat for several species and enhance biodiversity and ecosystem function. In particular, these include better habitats for bio-control agents such as predators that would disperse into the adjacent crop fields to target pests. However, to the best of our knowledge, an AF driven bio-control model, where a predator is housed or introduced into a prairie strip (or a STRIP), where it is provided with AF, and then subsequently disperses or is drifted into a neighboring crop field to target a pest, has not been investigated. The current manuscript is a first attempt in this direction.

The following \emph{general} model for an introduced predator population $y(t)$ depredating on a target pest population $x(t)$, while also provided with an additional food source of quality $\frac{1}{\alpha}$  and quantity $\xi$, has been proposed in the literature, 
\begin{equation}
\label{Eqn:1g}
\frac{dx}{dt} = x\left(1-\frac{x}{\gamma}\right) - f(x,\xi,\alpha) y, \    \frac{dy}{dt} =  g(x, \xi, \alpha) y  - \delta y.
\end{equation}
Here, $f(x,\xi,\alpha)$ is the functional response of the predator, which is pest dependent but also dependent on the additional food (hence the explicit dependence on $\xi, \alpha$). Likewise,
$g(x,\xi,\alpha)$ is the numerical response of the predator. 
If $\xi = 0$, that is, there is no additional food, and the model reduces to a classical predator-prey model of Gause type, that is $f(x,0,0)= g(x,0,0)$, where $f$ has the standard properties of a functional response. For these models, we know pest eradication \emph{is not possible}, as the only pest free state is $(0,0)$, which is typically unstable \cite{K01}. Thus modeling the dynamics of an introduced predator and its prey, a targeted pest via this approach, where the constructed 
$f(x,\xi,\alpha), g(x,\xi,\alpha)$ are used as a means to achieve a pest free state, has both practical and theoretical value, and thus has been well studied. Table \ref{Table:1},  summarizes some of the key literature in terms of the functional forms used in these models.

		The effect of increasing the quantity of AF $\xi$ in models of type \eqref{Eqn:1g} is to push the vertical predator nullcline to the left, thereby decreasing the equilibrium pest level. The extinction mechanism works by inputting   $\xi > \xi_{critical}$ (Here $\xi_{critical}$ depends on the model parameters and changes from model to model) s.t. the predator nullcline is pushed left past the y-axis (predator axis), into the $2^{nd}$ quadrant, annihilating any positive interior equilibrium. The earlier literature on AF purported that via this construct, due to the positivity of solutions, trajectories will converge onto the predator axis, yielding pest extinction. Furthermore, this can be done in minimal time \cite{SP10, SP11, VA22}. The pest extinction was shown to only occur in infinite time \cite{PWB20}, and furthermore pest extinction was shown to result in blow-up of the predator population as well \cite{PWB23}.
We outline further details on past literature on AF models; see Appendix \eqref{prior results}.

Our contributions in the current manuscript are as follows,

\begin{itemize}
\item Two new AF bio-control models with patch structure are introduced. Herein, a predator is introduced into a patch with AF (such as a prairie strip) and then drifts via \eqref{eq:patch_model_uni-i} or disperses via \eqref{eq:patch_model2} into a neighboring patch (such as a crop field), to target a pest. The mathematical analysis for the drift model is detailed in Section \ref{drift af model}, while the analysis for the dispersal model is provided in Section \ref{dispeersal af model}.

\item Unbounded growth of the predator population (without drift/dispersal) can be mitigated with sufficient drift and linear dispersal in the patch model. See Theorems \ref{thm:t1u1}, \ref{thm:t1}, and, \ref{thm:t1sd}.

\item Pest extinction in the crop field is globally asymptotically stable with drift via Theorem \ref{thm:gs1} and Lemma \ref{lem:pest_extinction_crop_field}.


\item The patch model enables pest extinction in the crop field with dispersal via Lemma \ref{lem:E1_stability}.

\item The patch model enables pest extinction in the prairie strip via Lemma \ref{lem:E2_stability}. To the best of our knowledge, this is the first result that shows AF mediated pest extinction is possible for pest-dependent functional responses.

\item The patch model enables chaotic dynamics; see Figure \eqref{fig:chaos_time_series}.

\item The patch model has other rich dynamics, such as the formation of limit cycles, via Theorem \ref{thm:hopf bifurcation}.

\item A special case of our results enables proving global stability of the interior equilibrium of the classical AF patch model via Corollary \ref{cor:g1}.

\item  The AF patch model, depending on parameter regimes, can do better or worse than a purely top-down bio-control model or an AF bio-control model without patch structure in terms of total pest density. To this end see Lemmas \ref{lem:rm_more_than_patch}-\ref{lem:af_less_than_patch}. Also see  Figure \eqref{fig:int_rm_vs_patch_fig} and Figure \eqref{fig:int_saf_vs_patch}.

\item We discuss the consequences of these results to AF-driven bio-control that would utilize prairie strips and so could be integrated with landscape management strategies and programs such as the STRIPs program in Iowa and the entire north-central US, as well as worldwide, given the recent 
global agricultural trends \cite{UNC}.
\end{itemize}
\begin{table}[H]
\caption{Dynamics of predator-pest models supplemented with AF}\label{Table:1}
	\scalebox{0.82}{
{\begin{tabular}{|c|c|c|c|c|}
		\hline
		& \mbox{Functional form} & \mbox{Relevant literature} & \mbox{AF requirement}& \mbox{ Effect on pest control}\\
		& &  & &  \\ \hline \hline
		(i) & $f(x,\xi,\alpha) = \frac{x}{1+\alpha \xi + x}$ & \cite{SP07, SP10, SP11,sen2015global} & $\xi > \frac{\delta}{\beta - \delta \alpha}$ & \mbox{Pest is eradicated, switching AF}\\
		& $g(x,\xi,\alpha) = \frac{\beta (x+\xi)}{1+\alpha \xi + x}$ &  &  &  \mbox{maintains/eliminates predator }\\          \hline
		(ii) & $f(x,\xi,\alpha) = \frac{x^{2}}{1+\alpha \xi^{2} + x^{2}}$  & \cite{SPV18}  & $\xi >\sqrt{ \frac{\delta}{\beta - \delta \alpha}}$  &  \mbox{Pest is eradicated, switching AF}\\
		& $g(x,\xi,\alpha) = \frac{\beta (x^{2}+\xi^{2})}{1+\alpha \xi^{2} + x^{2}}$ &  & & \mbox{ maintains/eliminates predator } \\ \hline
		(iii) & $f(x,\xi,\alpha) = \frac{x}{(1+\alpha \xi)( \omega x^{2}+1) + x}$   & \cite{SPD17} & $\xi > \frac{\delta}{\beta - \delta \alpha}$ &  \mbox{Pest is eradicated, switching AF} \\
		& $g(x,\xi,\alpha) = \frac{\beta (x+\xi(\omega x^{2}+1)}{(1+\alpha \xi)(\omega x^{2} + 1) + x}$ &  &  & \mbox{ maintains/eliminates predator}  \\ \hline
		(iv) & $f(x,\xi,\alpha) = \frac{x}{1+\alpha \xi + x + \epsilon y}$ & \cite{S02} &  $\beta \xi = \delta(1+\alpha)$  & \mbox{Pest extinction state stabilises,}  \\
		& $g(x,\xi,\alpha) = \frac{\beta (x+\xi)}{1+\alpha \xi + x + \epsilon y}$ &  & $\beta(\gamma + \xi) =\delta(1+\alpha \xi + \gamma)$  &\mbox{via TC}  \\ \hline
  
		\end{tabular}}
		}
\vspace{.1cm}		 
     \flushleft
\textbf{Note:}	Herein $\gamma$ is the carrying capacity of the pest, $\beta$ is the conversion efficiency of the predator, $\delta$ is the death rate of the predator, $\frac{1}{\alpha}$ is the quality of the additional food provided to the predator and $\xi$ is the quantity of additional food provided to the predator. 
		\end{table}

\section{Additional Food Patch Model Formulation - Drift}
\label{drift af model}

We assume that our landscape is divided into two sub-units, a crop field, which is the larger unit, and a prairie strip, which is the smaller unit. These are the two patches that make up our landscape. We will refer to the prairie strip as $\Omega_1$ and the crop field as $\Omega_2$. We assume that the prairie strip, by design, possesses row crops that will provide alternative/additional food to an introduced predator, which would enhance its efficacy in targeting a pest species that resides primarily in the crop field. The crop field has no AF. 

Thus, for the patch $\Omega_1$, we assume the quantity of AF as $\xi$ and the quality of AF as $\frac{1}{\alpha}$. Thus, essentially, in $\Omega_{1}$, we have an AF-driven predator-pest system for a pest density $x_{1}$ and a predator density $y_{1}$. We choose the classical Holling type II functional response for the functional response of the pest. We assume the predators will drift between $\Omega_{1}$ and $\Omega_{2}$ via a drift term, where the rate of drift out of $\Omega_{1}$ is $q_{1}$, and the rate at which predators ``arrive" or are drifted into $\Omega_{2}$ at the rate of $q_{2}$ see Figure \eqref{fig:patch_assumption_drift}. Here, the drift is driven by flooding or wind. This is dependent on the way the landscape slopes, and loss of individuals during drift require $q_{1} > q_{2}$. This is modeled as,

\begin{figure}[H]
\includegraphics[width = 7.2cm]{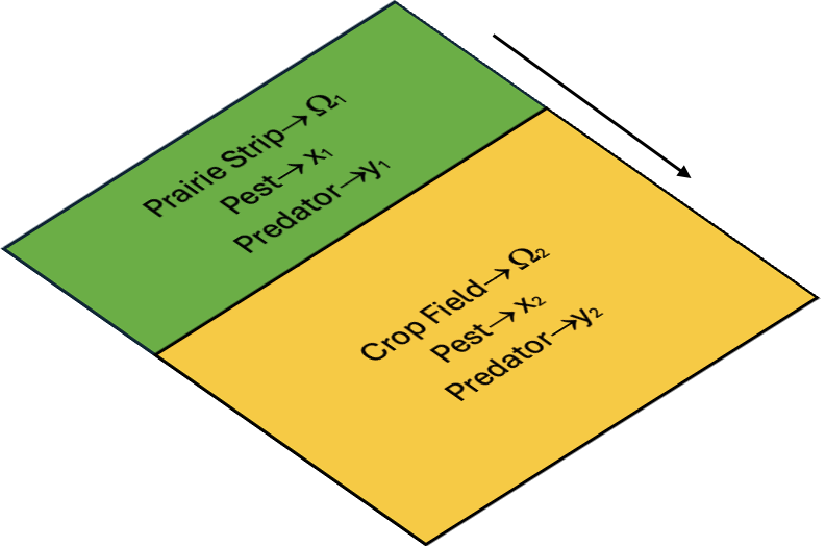}
\caption{Visual depiction of the underlying assumptions regarding patches $\Omega_1 \ \& \ \Omega_2$. }
\label{fig:patch_assumption_drift}
\end{figure}

\begin{equation} \label{eq:patch_model_uni-i}
\begin{aligned}
 \Dot{x_1}& = x_1 \left (1-\frac{x_1}{k_p}\right )-\frac{x_1 y_1}{1+x_1+\alpha \xi} \\
 \Dot{y_1}& = \epsilon_1 \left ( \frac{x_1+\xi}{1+x_1+\alpha \xi} \right)\ y_1 - \delta_1 y_1 - q_1 y_1   \\
 \Dot{x_2}& = x_2 \left (1- \frac{x_2}{k_c} \right)-\frac{x_2  y_2}{1+x_2} \\
\Dot{y_2}& = \epsilon_2 \left ( \frac{x_2}{1+x_2} \right ) \ y_2 - \delta_2 y_2 + q_2 y_1 
\end{aligned}
\end{equation}


 \subsection{Predator Blowup Prevention}

\begin{thm}
\label{thm:t1u1}
Consider the additional food patch model \eqref{eq:patch_model_uni-i} with $\epsilon_{1} > \delta_{1}, \epsilon_{2} > \delta_{2}$. If $q_{1}=q_{2}=0$, that is in the absence of drift, and $\xi > \xi_{critical} =  \dfrac{\delta_1}{\epsilon_1 - \alpha \delta_1}$, the predator population blows up in infinite time. However, with drift,  $q_{1}, q_{2} > 0$ s.t., $\max{\left(1,\dfrac{\epsilon_1 \xi}{1+\alpha \xi}\right)} < (\delta_1 + q_1 )$, and $\forall \ q_{2} > 0$, the predator population remains bounded for all time.
\end{thm}

\begin{proof}
In the absence of drift, $q_{1}=q_{2}=0$, blow-up in $y_{1}$ follows via \cite{PWB23}, and subsequently in $y_{2}$, via the form of the equations in the crop field. To show the damping effect of drift, we proceed via contradiction. Assume the predator population $y_{1}$ blows up in infinite time. 
WLOG we work under the assumption $\dfrac{\epsilon_1\xi}{1+\alpha \xi} > 1$. Now using the monotonocity property of $\dfrac{x_1+\xi}{1+x_1+\alpha \xi} $, we have, $\dfrac{x_1+\xi}{1+x_1+\alpha \xi}  \leq \dfrac{\xi}{1+\alpha \xi} $. Thus, via comparison, we have,

\begin{eqnarray}
   && \Dot{y_1}  \nonumber \\
   &=& \epsilon_1 \left ( \dfrac{x_1+\xi}{1+x_1+\alpha \xi} \right)\ y_1 - \delta_1 y_1 - q_1 y_{1} \nonumber \\
   &\leq&  \dfrac{\epsilon_1\xi}{1+\alpha \xi}y_{1}  - \delta_1 y_{1} - q_1y_{1}  \nonumber \\
   &=&  \left ( \dfrac{\epsilon_1\xi}{1+\alpha \xi}  - \delta_1  - q_1  \right)y_{1} \nonumber \\
   &\leq& 0. \nonumber \\
\end{eqnarray}

Thus via our assumption $\Dot{y_1} \leq  0$, and integration in $[0,t]$ yields $y_{1}(t) \leq y_{1}(0)$, so $y_{1}$ cannot blow-up, and we have a contradiction.
\end{proof}
\begin{rem}
Figure \eqref{fig:blowup_uni_patch} shows the blow-up dynamics in predator population when $\xi> \xi_{critical}$ and no drift in \eqref{eq:patch_model_uni-i} i.e., $q_1=q_2=0$. Figure \eqref{fig:blowup_prevention_uni_patch} shows the blow-up prevention in patch $\Omega_1$ when there is drift between the patches, i.e., $q_1,q_2 \neq 0$. In both simulations, the values of all other parameters remain constant, showing that the movement of predators can prevent blow-up in predator populations. The values of $q_1 \ \& \ q_2$ are selected based on the parametric constraints outlined in Theorem \ref{thm:t1u1}.
\end{rem}
\begin{figure}
  \begin{subfigure}{.44\textwidth}
\centering
  \includegraphics[width= 7cm, height=6cm]{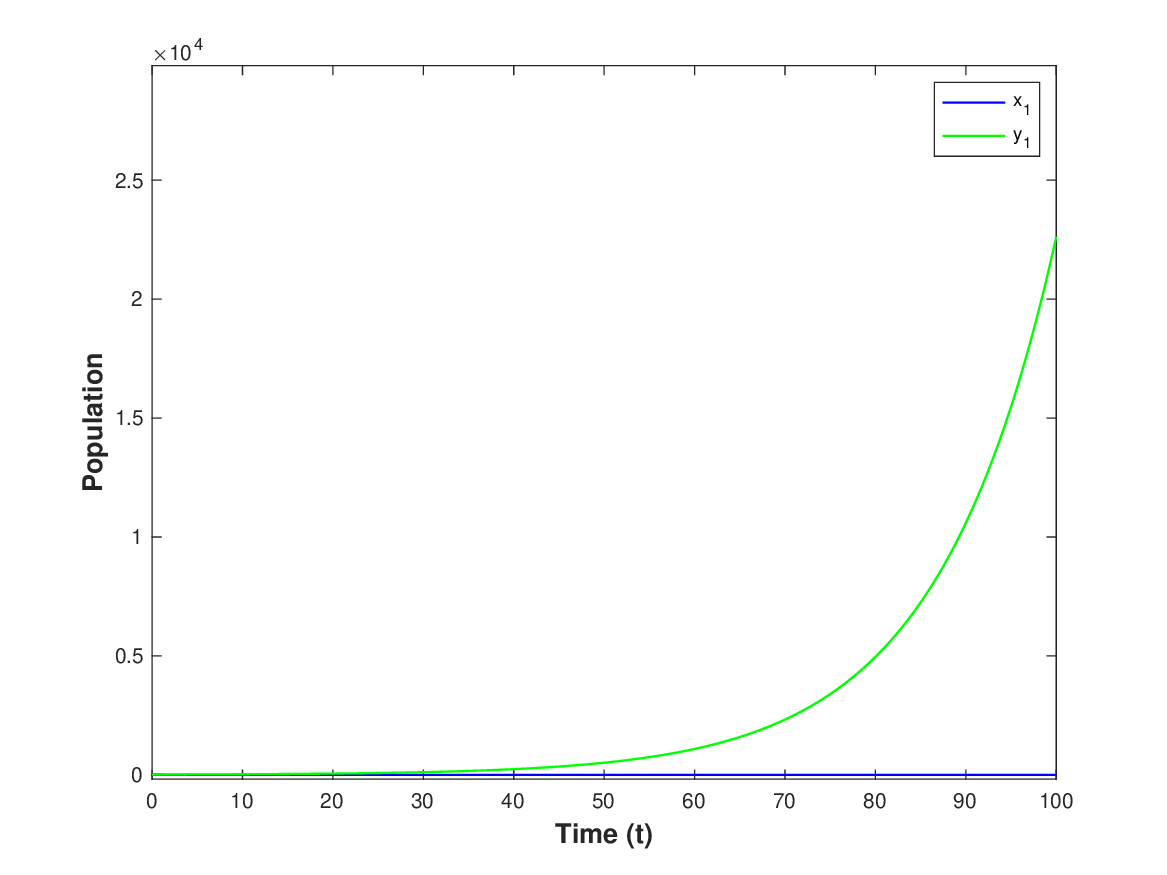}
 \subcaption{ $q_1=0,\ q_2=0$}
 \label{fig:blowup_uni_patch}
  \end{subfigure}
  \begin{subfigure}{.44\textwidth}
  \centering
  \includegraphics[width= 7cm, height=6cm]{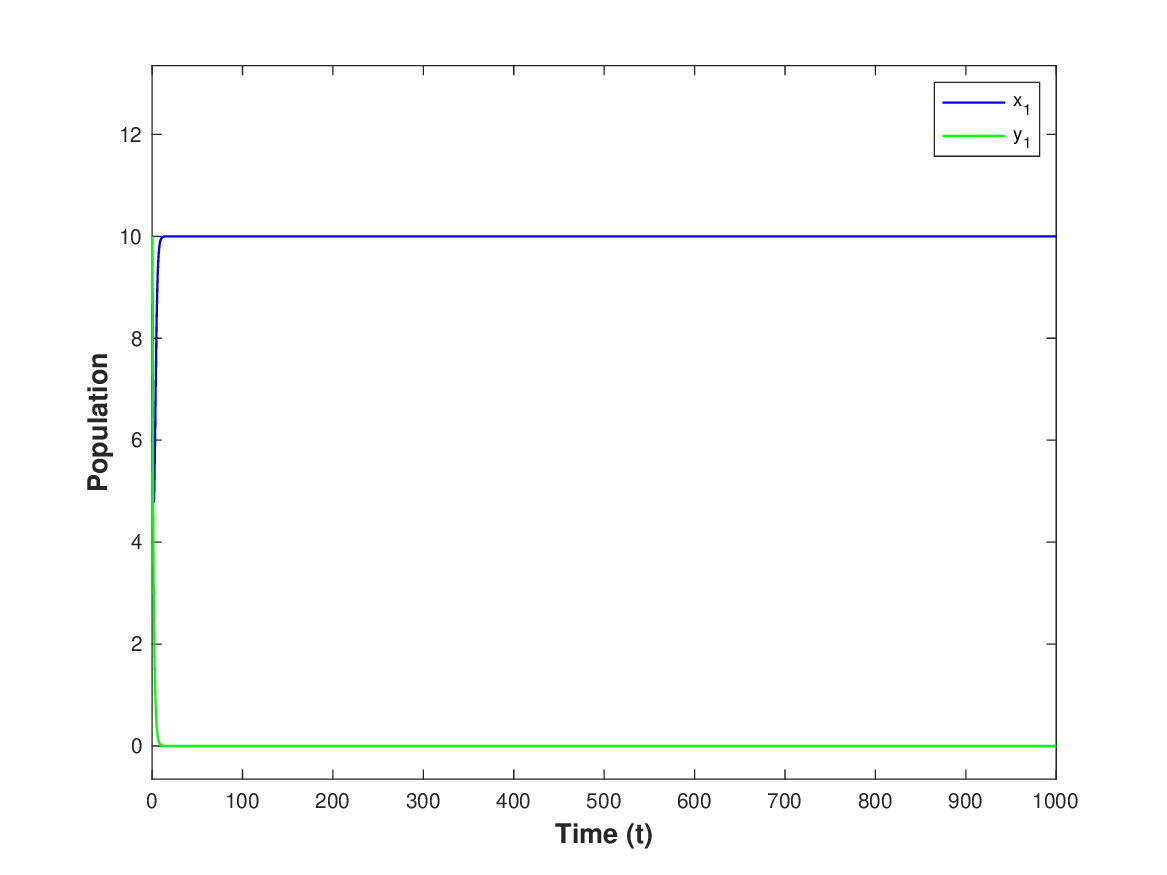}
 \subcaption{ $q_1=0.82, \ q_2=0.6$}
\label{fig:blowup_prevention_uni_patch}
 \end{subfigure}
 \caption{ The parameter set used is $k_p=10,k_c=10,\alpha=0.2, \xi=0.8, \epsilon_1 =\epsilon_2 =0.4, \delta_1= \delta_2=0.2 \ \text{with I.C.} = [10,10,10,10]$. }
 \label{blowup_drift_comp}
 \end{figure}
\subsection{Equilibrium Analysis}  

We now consider the existence and local stability analysis of the biologically relevant equilibrium points for the system \eqref{eq:patch_model_uni-i}.
The Jacobian matrix $(\hat{J})$ for \eqref{eq:patch_model_uni-i} is given by: 
 
\begin{equation} 
 \hat{J} = \begin{bmatrix}
1 - \dfrac{2 x_1}{k_p} -  \dfrac{y_1  \left(1+ \alpha \xi \right) } {\left(1+x_1+\alpha \xi\right)^2}  & \dfrac{- x_1}{1+x_1+\alpha \xi} & 0 & 0 \vspace{0.25cm}
  \\ 
\dfrac{\epsilon_1  \left(1 + \left(\alpha - 1\right) \xi\right) \ y_1}{(1+x_1+\alpha \xi)^2} &  \dfrac{\epsilon_1 \left(x_1 + \xi\right)}{1+x_1+\alpha \xi} - \delta_1 - q_1 & 0 & 0
\\
0 & 0 & 1 - \dfrac{2 x_2}{k_c} - \dfrac{y_2}{(1+x_2)^2} & 
  \dfrac{- x_2}{1+x_2} \vspace{0.25cm}
  \\
  0 & q_2 & \dfrac{\epsilon_2 y_2}{(1+x_2)^2} & \dfrac{\epsilon_2 x_2}{1+x_2} - \delta_2 
  \end{bmatrix}
\label{general_jacobian_k1_k3_0_uni}
 \end{equation}

\subsubsection{Pest free state in patch \texorpdfstring {$\Omega_2$}{Lg}\mbox{}}
\begin{lemma}
    The equilibrium point $ \hat{E_1} = (x_1^*,y_1^*,0,y_2^*)$ exists if  $\xi < \dfrac{\delta_1 + q_1 }{\epsilon_1 - \alpha (\delta_1 + q_1 )}$ and $ \epsilon_1 > \delta_1 + q_1 $.
\label{lem:E1_existence_unkidirectional}
  \end{lemma}
  
  \begin{proof}

See Appendix \eqref{crop_extinctio_uni_existence}.
\end{proof}

 \begin{lemma}
  \label{lem:E_stability_uni}
The equilibrium point $ \hat{E_1} = (x_1^*,y_1^*,0,y_2^*)$ is conditionally locally asymptotically stable. 

\end{lemma}
\begin{proof}
    See Appendix \eqref{lem proof:E_stability_uni}.
\end{proof}
\subsubsection{Pest free state in patch \texorpdfstring {$\Omega_1$}{Lg}}\mbox{}
\begin{lemma}
    The equilibrium point $ \hat{E_2} = (0,y_1^*,x_2^*,y_2^*)$ exists if  $\xi = \dfrac{\delta_1 + q_1 }{\epsilon_1- \alpha(\delta_1 + q_1 )}$ and $ \epsilon_1 > \delta_1 + q_1 $ and, $y_1^* > \dfrac{\delta_2}{q_2}$.
\label{lem:E3_existence_unkidirectional}
  \end{lemma}
  
  \begin{proof}

  See Appendix \eqref{proof_E3_existence_drift}.
  \end{proof}
  
  \begin{lemma}
\label{stability_pest_ext_af_drift}
The equilibrium point $ \hat{E_2} = (0,y_1^*,x_2^*,y_2^*)$ is a non-hyperbolic point. 
\end{lemma}

\begin{proof}
    See Appendix \eqref{stability_pest_ext_af_drift_proof}.
\end{proof}

\begin{rem}
The pest extinction state in the prairie strip exists under strict equality for $\xi$ since the predators are drifting from $\Omega_1$ to $\Omega_2$, leading to a decline in the predator population within the prairie strip over time. 
\end{rem}

\subsubsection{Pest free state in both \texorpdfstring{$\Omega_1 \ \& \ \Omega_2$}{Lg}\mbox{}}
\begin{lemma}
    The equilibrium point $ \hat{E_3} = (0,y_1^*,0,y_2^*)$ exists if  $\xi = \dfrac{\delta_1 + q_1 }{\epsilon_1 - \alpha (\delta_1 + q_1 )}$ and $ \epsilon_1 > \delta_1 + q_1 $.
\label{lem:E2_existence_unkidirectional}
  \end{lemma}
  \begin{proof}

See Appendix \eqref{proof_E2_existence_unkidirectional}.
\end{proof}
\begin{lemma}
\label{stability_pest_ext_drift}
The equilibrium point $ \hat{E_3} = (0,y_1^*,0,y_2^*)$ is a non-hyperbolic point. 
\end{lemma}

\begin{proof}
    See Appendix \eqref{proof:stability_pest_ext_drift}.
\end{proof}

\subsubsection{Coexistence state in both \texorpdfstring {$\Omega_1 \ \& \ \Omega_2$}{Lg}}
\begin{lemma}
The equilibrium point $ \hat{E_4} =(x_1^\star,y_1^\star,x_2^\star,y_2^\star)$ exists if  $\xi < \dfrac{\delta_1+q_1}{\epsilon_1 - \alpha (\delta_1+q_1)}$, {$ \epsilon_1 > \delta_1+q_1 $} and, $y_1^\star > \dfrac{\delta_2}{q_2}$. 
 \label{lem:co_existence}
\end{lemma}
\begin{proof}
 See Appendix \eqref{proof_co_existence_uni}.
 \end{proof}

\begin{lemma}
    \label{lem:loc_uni_coexistence} 
    The equilibrium point $  \hat{E_4} =(x_1^\star,y_1^\star,x_2^\star,y_2^\star)$ is conditionally locally asymptotically stable.
 \end{lemma}
\begin{proof}
    See Appendix \eqref{loc_uni_coexistence_proof}.
\end{proof}
\begin{figure}
\centering
\includegraphics[width = 10cm, height=7cm]{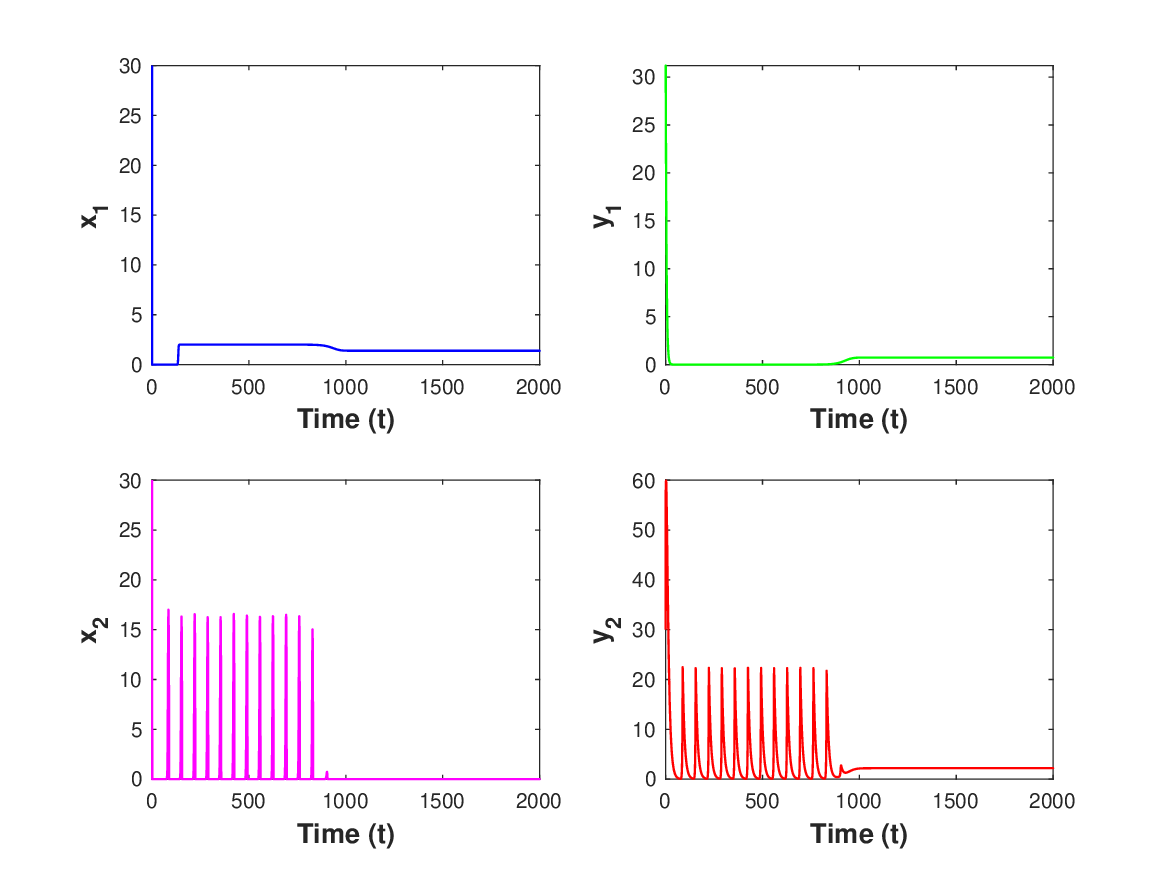}
\caption{The time series figure shows the pest extinction in the patch 
$\Omega_2$. The parameters used are $k_p=2,k_c=20,\alpha=0.02, \xi=0.1, \epsilon_1 = 0.45, \epsilon_2 =0.8, \delta_1= 0.08, \delta_2=0.1, q_1=0.2 , q_2=0.3 \ \text{with I.C.} = [30,30,30,30]$.}
 \label{fig:gs_x2_extinction}
 \end{figure}
 \begin{figure}
\centering
\includegraphics[width = 10cm, height=7cm]{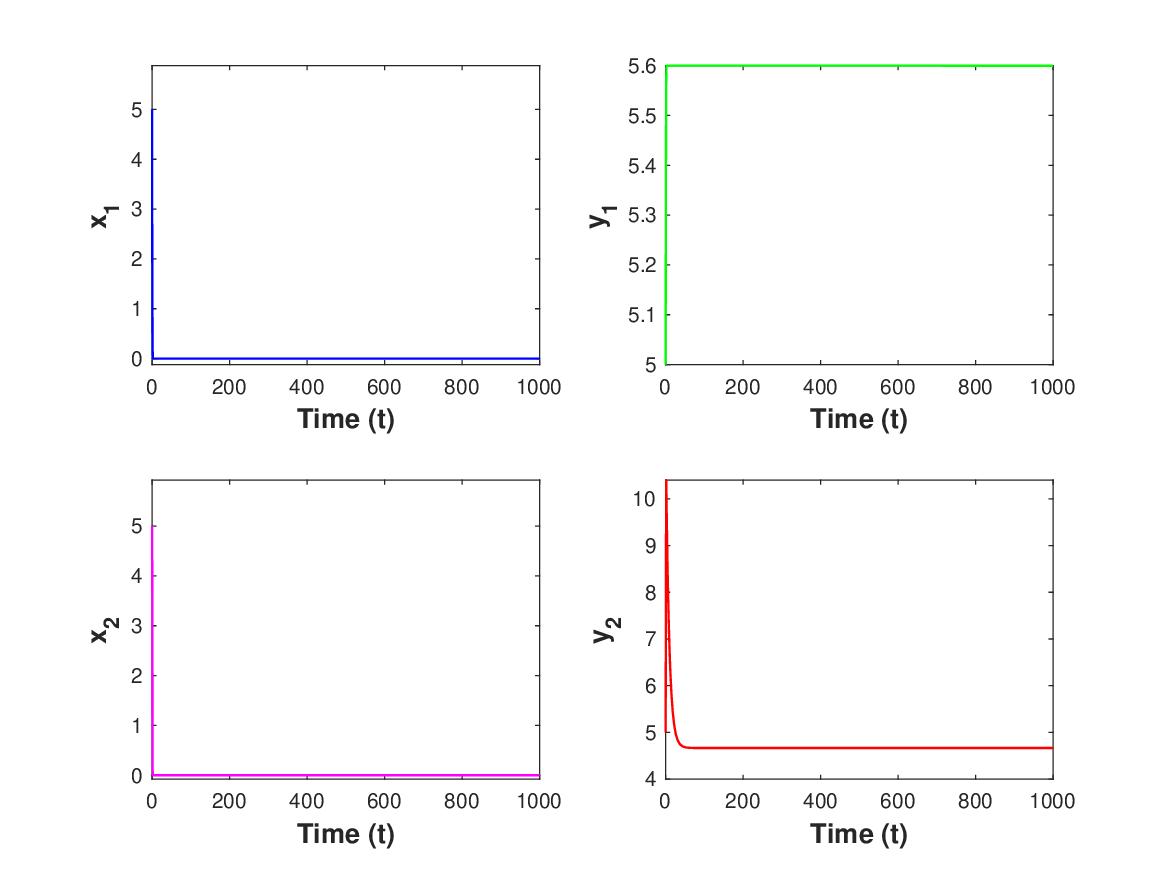}
\caption{The time series figure shows pest extinction in both patches
$\Omega_1$ \& $\Omega_2$ for an exact quantity of additional food. The parameters used are $k_p=2, k_c=10,
 \alpha=0.2,  \xi=0.798969, \epsilon_1=0.45, \epsilon_2 =0.8, \delta_1= 0.15, \delta_2=0.12, q_1=0.16, q_2=0.1 \ \text{with I.C.} = [5,5,5,5]$.
 }
 \label{fig:pest_extinction_drift}
 \end{figure}
 \begin{figure}
\centering
\includegraphics[width = 10cm, height=7cm]{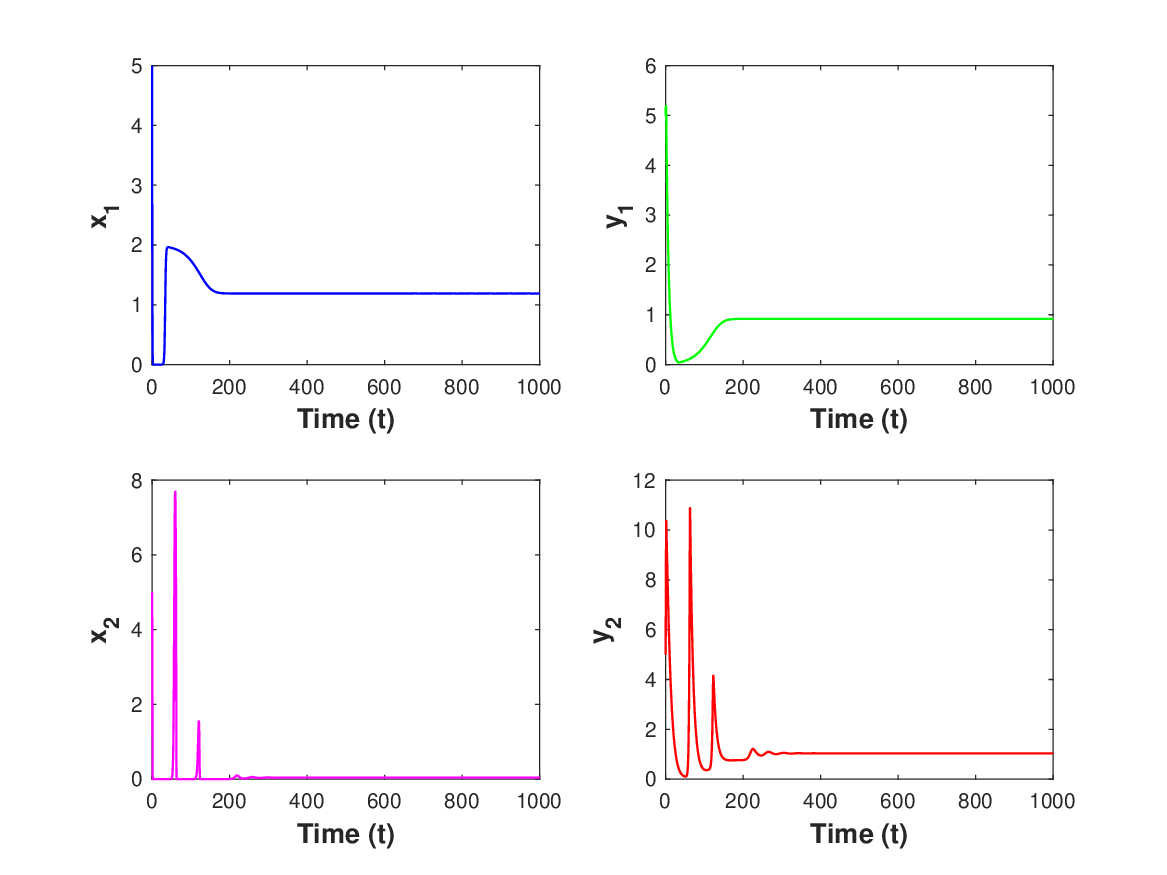}
\caption{The time series figure shows  the coexistence state in both patches $\Omega_1$  \& 
$\Omega_2$. The parameters used are $k_p=2, k_c=10,
 \alpha=0.2,  \xi=0.37, \epsilon_1=0.45, \epsilon_2 =0.8, \delta_1= 0.15, \delta_2=0.12, q_1=0.16, q_2=0.1 \ \text{with I.C.} = [5,5,5,5]$.
 }
 \label{fig:interior_equilibrium_drift}
 \end{figure}

\begin{rem}
Figure \eqref{fig:gs_x2_extinction} demonstrates that pest extinction in the crop field  $(\Omega_2)$ is possible, validating Lemma \ref{lem:E_stability_uni}.
Figure \eqref{fig:pest_extinction_drift} suggests that the pest-free state in both $\Omega_1 \ \& \ \Omega_2$ simultaneously is possible but requires strict parametric equality on the quantity of additional food as shown in Lemma \ref{stability_pest_ext_drift}. Additionally, Figure \eqref{fig:interior_equilibrium_drift} shows that the coexistence state is possible in both patches, supporting Lemma \ref{lem:loc_uni_coexistence}.  
\end{rem}
 
\subsection{Global stability of pest extinction state in  \texorpdfstring {$\Omega_2$}{Lg} }

\ 

Proving global stability of pest extinction in the crop field for system \eqref{eq:patch_model_uni-i} is problematic directly. 

\begin{rem}
The forms of the functional and numerical response of the introduced predator in \eqref{eq:patch_model_uni-i} are not symmetric. Thus, the standard Lyapunov function \cite{hsu2005survey} for such systems will not work. 
\end{rem}

\begin{rem}
Our strategy is to first derive a sub and super solution to the system of equations in $\Omega_{1}$, \eqref{eq:patch_model_uni-i}, where on each sub and super solution, we can use standard Lyapunov theory. We then employ a squeezing argument to yield the coexistence of the original system in $\Omega_{1}$. 
\end{rem}
\subsubsection{Coexistence state in patch \texorpdfstring {$\Omega_1$}{Lg}} 
\ 

The equations representing dynamics in patch $\Omega_1$ for \eqref{eq:patch_model_uni-i} are given by:
 \begin{equation} \label{eq:patch_model_uni_patch1}
\begin{aligned}
 \Dot{x_1}& = x_1 \left (1-\frac{x_1}{k_p}\right )-\frac{x_1 y_1}{1+x_1+\alpha \xi} \\
 \Dot{y_1}& = \epsilon_1 \left ( \frac{x_1+\xi}{1+x_1+\alpha \xi} \right)\ y_1 - \delta_1 y_1 - q_1 y_1  \end{aligned}
 \end{equation}
As, $ (1+\alpha\xi < 1+x_1+\alpha\xi) \implies \dfrac{\xi}{1+\alpha\xi}> \dfrac{\xi}{1+x_1+\alpha\xi}.   $
The model \eqref{eq:patch_model_uni_patch1} has a   super solution (in the predator component) given by \eqref{eq:patch_model_uni_patch_sup}:
 \begin{equation} 
\begin{aligned}
 \Dot{\tilde {x_1}}& = \tilde{x_1} \left (1-\frac{\tilde{x_1}}{k_p}\right )-\frac{\tilde{x_1} \tilde{y_1}}{1+\tilde{x_1}+\alpha \xi} \\
 \Dot{\tilde{y_1}}& = \epsilon_1 \left ( \frac{\tilde{x_1}}{1+\tilde{x_1}+\alpha \xi} \right)\ \tilde{y_1} - \left(\delta_1+ q_1 - \dfrac{\epsilon_1\xi  }{1+\alpha\xi} \right) \tilde{y_1}
 \end{aligned}
 \label{eq:patch_model_uni_patch_sup}
\end{equation}

Let $\tilde \delta = \delta_1+ q_1 - \dfrac{\epsilon_1\xi  }{1+\alpha\xi} $ then \eqref{eq:patch_model_uni_patch_sup} becomes,
\begin{equation}
\begin{aligned}
\Dot{\tilde{x_1}}& = \tilde{x_1} \left (1-\frac{\tilde{x_1}}{k_p}\right )-\frac{\tilde{x_1} \tilde{y_1}}{1+\tilde{x_1}+\alpha \xi} \\
 \Dot{\tilde{y_1}}& = \epsilon_1 \left ( \frac{\tilde{x_1}}{1+\tilde{x_1}+\alpha \xi} \right)\ \tilde{y_1} - \tilde\delta \tilde{y_1}  
\end{aligned}
\label{eq:patch_model_uni_patch_sup_delta_tilde}
\end{equation}

We can state the following lemma considering system \eqref{eq:patch_model_uni_patch_sup_delta_tilde},
\begin{lemma}
  The equilibrium point $\tilde{E} = (\tilde {x^*_1}, \tilde {y^*_1}) \  \text{exists if}  \ \epsilon_1 > \tilde{\delta} $. 
  \label{coexistence_existence_super}
\end{lemma}

\begin{proof}
   See Appendix \eqref{proof_coexistence_existence_super}. 
\end{proof}
\begin{lemma}
\label{lem:p1s}
    Consider system \eqref{eq:patch_model_uni_patch_sup_delta_tilde}. Then under the parametric restrictions, $k_p<\dfrac{(1+\alpha\xi)(\epsilon_1+\tilde{\delta})}{\epsilon_1-\tilde{\delta}}$, we have that \ $\tilde{E} = (\tilde {x^*_1}, \tilde {y^*_1})$, is globally asymptotically stable.
\end{lemma}

\begin{proof}
The coexistence equilibrium $(\tilde{x^*_1},  \tilde{y^*_1})$ is locally asymptotically stable  if,  \cite{Hal08}\\
\begin{equation*}
\tilde{x^*_1} > \dfrac{k_p-(1+\alpha\xi)}{2} \implies  \dfrac{2\tilde{\delta}(1+\alpha\xi)}{\epsilon_1- \tilde{\delta}} > k_p-(1+\alpha\xi) \implies k_p< \dfrac{(1+\alpha\xi)(\epsilon_1+\tilde{\delta})}{\epsilon_1-\tilde{\delta}}
\end{equation*}

Let $g(\tilde{x_1})= \left (1-\dfrac{\tilde{x_1}}{k_p}\right ), \ p(\tilde{x_1}) = \dfrac{\tilde{x_1}}{1+\tilde{x_1}+\alpha \xi} $ then, 
\begin{equation}
\begin{aligned}
\Dot{\tilde{x_1}}& = \tilde{x_1} g(\tilde{x_1})- p(\tilde{x_1})\tilde{y_1} \\
\Dot{\tilde{y_1}}& = (\epsilon_1 p(\tilde{x_1}) - \tilde\delta ) \ \tilde{y_1}    \end{aligned}
\label{p_g_general_sup}
\end{equation}
We consider the following Lyapunov function, \cite{hsu2005survey}
\begin{equation}
 V =  \int_{\tilde{x^*_1}}^{\tilde{x_1}} \dfrac{p(\phi)- \dfrac{\tilde\delta }{\epsilon_1}}{p(\phi)} \,d\phi  + \dfrac{1}{\epsilon_1}\int_{\tilde{y^*_1}}^{\tilde{y_1}}\dfrac{\eta-\tilde{y^*_1}}{\eta}\,d\eta
 \label{lyapunov_construction_sup}
\end{equation}
Differentiating $V$ with respect to time $t$ along the solutions of model \eqref{eq:patch_model_uni_patch_sup}, we get
\begin{equation*}
    \dot{V} = \dfrac{p(\tilde{x_1})- \dfrac{\tilde\delta }{\epsilon_1}}{p(\tilde{x_1})}\  \dot{\tilde{x_1}} + \dfrac{1}{\epsilon_1} \dfrac{\tilde{y_1}-\tilde{y^*_1}}{\tilde{y_1}} \ \dot{\tilde{y_1}}
\label{sup_lyapunov_differentiated}
\end{equation*}
Using system of equations \eqref{p_g_general_sup}, after some algebraic manipulations we have,
\begin{equation*}
\setlength{\jot}{10pt}
\begin{aligned}
&\dot{V} = \dfrac{p(\tilde{x_1})- \dfrac{\tilde\delta }{\epsilon_1}}{p(\tilde{x_1})}\  \{\tilde{x_1} g(\tilde{x_1})- p(\tilde{x_1})\tilde{y_1}\} + \dfrac{1}{\epsilon_1} \dfrac{\tilde{y_1}-\tilde{y^*_1}}{\tilde{y_1}} \ \{\epsilon_1 p(\tilde{x_1}) - \tilde\delta \} \ \tilde{y_1} \\
& = \dfrac{p(\tilde{x_1})- \dfrac{\tilde\delta }{\epsilon_1}}{p(\tilde{x_1})}\  \{\tilde{x_1} g(\tilde{x_1})- p(\tilde{x_1})\tilde{y_1}\} +(\tilde{y_1}-\tilde{y^*_1}) \ 
\left \{p(\tilde{x_1})-\dfrac{\tilde\delta }{\epsilon_1} \right \}
\\
 & = \left(p(\tilde{x_1})- \dfrac{\tilde\delta }{\epsilon_1}\right) \left( \dfrac{\tilde{x_1} g(\tilde{x_1})}{p(\tilde{x_1})} - \tilde{y^*_1}
 \right)\\
 \end{aligned}
\end{equation*}
Plugging in the value of $\dfrac{\tilde\delta }{\epsilon_1}$ then,
 \begin{equation}
\dot{V}\rvert_{({\tilde{x_1},{\tilde{y_1})} = (\tilde{x^*_1},\tilde{y^*_1})}} = \left(p(\tilde{x_1})- p(\tilde{x^*_1})\right) \left( \dfrac{\tilde{x_1} g(\tilde{x_1})}{p(\tilde{x_1})} - \tilde{y^*_1}
 \right) \leq 0
 \label{eq: lyapunov_super_system}
 \end{equation}

Since, $p(\tilde{x_1})$ is an increasing function in $\tilde{x_1}$, therefore $ \left( p(\tilde{x_1})-p(\tilde{x^*_1})\right)$ is negative if $\tilde{x_1} < \tilde{x^*_1} < k_p $ and it is positive for $ \tilde{x^*_1} <\tilde{x_1} < k_p $. Also, $\left( \frac{\tilde{x_1} g(\tilde{x_1})}{p(\tilde{x_1})} - \tilde{y^*_1} \right)$ is positive for $\tilde{x_1} < \tilde{x^*_1} < k_p $ and negative for $ \tilde{x^*_1} <\tilde{x_1} < k_p $. Therefore, we have that 

\begin{equation*}
     \dot{V}\rvert_{({\tilde{x_1} ,{\tilde{y_1})} = (\tilde{x^*_1},\tilde{y^*_1})}}  \leq 0 \ \text{in}  \ \mathbb{R}^4_{+}
\end{equation*}

Thus, the lemma is proved.
\end{proof}

Now as, $(1+x_1+\alpha\xi) >0 \implies  \dfrac{\xi}{1+x_1+\alpha\xi} > 0.$ The model \eqref{eq:patch_model_uni_patch1} has the following sub solution (in the predator component), \eqref{eq:patch_model_uni_patch_sub}:
 \begin{equation} 
\begin{aligned}
 \Dot{\Bar{x_1}}& = \Bar{x_1}\left (1-\frac{\Bar{x_1}}{k_p}\right )-\frac{\Bar{x_1} \Bar{y_1}}{1+\Bar{x_1}+\alpha \xi} \\
 \Dot{\Bar{y_1}}& = \epsilon_1 \left ( \frac{\Bar{x_1}}{1+\Bar{x_1}+\alpha \xi} \right)\ \Bar{y_1} - (\delta_1+ q_1 ) \Bar{y_1} 
 \end{aligned}
 \label{eq:patch_model_uni_patch_sub}
\end{equation}

Let $\Bar{\delta} = \delta_1+ q_1  $ then \eqref{eq:patch_model_uni_patch_sub} becomes,
\begin{equation}
\begin{aligned}
\Dot{\Bar{x_1}}& = \Bar{x_1}\left (1-\frac{\Bar{x_1}}{k_p}\right )-\frac{\Bar{x_1} \Bar{y_1}}{1+\Bar{x_1}+\alpha \xi} \\
 \Dot{\bar{y_1}}& = \epsilon_1 \left ( \frac{\Bar{x_1}}{1+\Bar{x_1}+\alpha \xi} \right)\ \Bar{y_1} - \Bar{\delta} \Bar{y_1}
 \end{aligned}
\label{eq:patch_model_uni_patch_sub_delta_bar}
\end{equation}
We can state the following lemma considering system \eqref{eq:patch_model_uni_patch_sub_delta_bar},
 \begin{lemma}
  The equilibrium point  $\bar{E} = (\bar {x^*_1}, \bar {y^*_1}) \  \text{exists if}  \ \epsilon_1 > \bar{\delta} $. 
  \label{coexistence_existence_sub}
\end{lemma}

\begin{proof}
   See Appendix \eqref{proof_coexistence_existence_sub}. 
\end{proof}

\begin{lemma}
\label{lem:p2s}
     Consider system \eqref{eq:patch_model_uni_patch_sub_delta_bar}. Then under the parametric restrictions, $k_p< \dfrac{(1+\alpha\xi)(\epsilon_1+\bar{\delta})}{\epsilon_1-\bar{\delta}} $, we have that $\bar{E} = (\Bar {x^*_1}, \Bar {y^*_1})$, is globally asymptotically stable.
\end{lemma}

\begin{proof}

The coexistence equilibrium $(\bar{x^*_1},  \bar{y^*_1})$ is locally asymptotically stable if, \cite{Hal08} \\ 
\begin{equation*}
\bar{x^*_1} > \dfrac{k_p-(1+\alpha\xi)}{2} \implies  \dfrac{2\bar{\delta}(1+\alpha\xi)}{\epsilon_1- \bar{\delta}} > k_p-(1+\alpha\xi) \implies k_p< \dfrac{(1+\alpha\xi)(\epsilon_1+\bar{\delta})}{\epsilon_1-\bar{\delta}}
\end{equation*}

Let $g(\Bar{x_1})= \left (1-\dfrac{\Bar{x_1}}{k_p}\right ), \ p(\Bar{x_1}) = \dfrac{\Bar{x_1}}{1+\Bar{x_1}+\alpha \xi} $ then,
\begin{equation}
\begin{aligned}
\Dot{\Bar{x_1}}& = \Bar{x_1} g(\Bar{x_1})- p(\Bar{x_1})\Bar{y_1} \\
 \Dot{\bar{y_1}}& = (\epsilon_1 p(\Bar{x_1}) - \bar\delta ) \ \Bar{y_1}   
 \end{aligned}
 \label{p_g_general_sub}
\end{equation}
We consider the following Lyapunov function  \cite{hsu2005survey},
\begin{equation}
 V =  \int_{\Bar{x^*_1}}^{\Bar{x_1}} \dfrac{p(\phi)- \dfrac{\bar\delta }{\epsilon_1}}{p(\phi)} \,d\phi  + \dfrac{1}{\epsilon_1}\int_{\Bar{y^*_1}}^{\Bar{y_1}}\dfrac{\eta-\Bar{y^*_1}}{\eta}\,d\eta
\label{lyapunov_construction_sub}
\end{equation}
Differentiating $V$ with respect to time $t$ along the solutions of model \eqref{p_g_general_sub}, we get
\begin{equation*}
\dot{V} = \dfrac{p(\Bar{x_1})- \dfrac{\bar\delta }{\epsilon_1}}{p(\Bar{x_1})}\  \dot{\Bar{x_1}} + \dfrac{1}{\epsilon_1} \dfrac{\Bar{y_1}-\Bar{y^*_1}}{\Bar{y_1}} \ \dot{\Bar{y_1}}
\label{sub_v_differentiated}
\end{equation*}
Using system of equations \eqref{p_g_general_sub}, after some algebraic manipulations we have,
\begin{equation*}
\setlength{\jot}{10pt}
\begin{aligned}
&\dot{V} = \dfrac{p(\bar{x_1})- \dfrac{\bar\delta }{\epsilon_1}}{p(\bar{x_1})}\  \{\bar{x_1} g(\bar{x_1})- p(\bar{x_1})\bar{y_1}\} + \dfrac{1}{\epsilon_1} \dfrac{\bar{y_1}-\bar{y^*_1}}{\bar{y_1}} \ \{\epsilon_1 p(\bar{x_1}) - \bar\delta \} \ \bar{y_1} \\
& = \dfrac{p(\Bar{x_1})- \dfrac{\bar\delta }{\epsilon_1}}{p(\Bar{x_1})}\  \{\Bar{x_1} g(\Bar{x_1})- p(\Bar{x_1})\Bar{y_1}\} +(\Bar{y_1}-\Bar{y^*_1}) \
\left \{p(\Bar{x_1})-\dfrac{\bar\delta }{\epsilon_1}\right \}
\\
 & = \left(p(\Bar{x_1})- \dfrac{\bar\delta }{\epsilon_1}\right) \left( \dfrac{\Bar{x_1} g(\Bar{x_1})}{p(\Bar{x_1})} - \Bar{y^*_1}
 \right)
 \end{aligned}
\end{equation*}
 Plugging in the value of $\dfrac{\Bar\delta }{\epsilon_1}$ then,
 \begin{equation}
 \dot{V}\rvert_{({\Bar{x_1} ,{\Bar{y_1})} = (\Bar{x^*_1},\Bar{y^*_1})}} = \left(p(\Bar{x_1})- p(\Bar{x^*_1})\right) \left( \dfrac{\Bar{x_1} g(\Bar{x_1})}{p(\Bar{x_1})} - \Bar{y^*_1}\right) 
 \label{eq:lyapunov_sub_system}
\end{equation}

Since, $p(\Bar{x_1})$ is an increasing function in $\Bar{x_1}$, therefore $\left(p(\Bar{x_1})-p(\Bar{x^*_1})\right) $ is negative if $\Bar{x_1} < \Bar{x^*_1} < k_p $ and it is positive for $ \Bar{x^*_1} <\Bar{x_1} < k_p $. Also, $\left( \frac{\Bar{x_1} g(\Bar{x_1})}{p(\Bar{x_1})} - \Bar{y^*_1} \right)$ is positive for $\Bar{x_1} < \Bar{x^*_1} < k_p $ and negative for $ \Bar{x^*_1} <\Bar{x_1} < k_p $. Therefore, we have that 
\begin{equation*}
     \dot{V}\rvert_{({\Bar{x_1} ,{\Bar{y_1})} = (\Bar{x^*_1},\Bar{y^*_1})}}  \leq 0 \ \text{in}  \ \mathbb{R}^4_{+}
\end{equation*}
\end{proof}

We next state the following auxilliary theorem from \cite{perko2013differential},

\begin{thm}
\label{thm:dp1} (Dependence on parameters)

Let $E$ be an open subset of $R^{n+m}$ containing the point $(\bf{x_0}, \bf{\mu_0})$ where $(\bf{x_0}) \in R^m$ and assume that $\bf{f} \in C^1(E)$. It then follows that there exists $a>0$, $\delta> 0 $ such that for all $\bf{y} \in N_{\delta}(x_0)$ and $\bf{\mu} \in N_{\delta}(\mu_0)$, the initial value problem,

\begin{equation*}
    \dot{\bf{x}} = \bf{f}(\bf{x}, \bf{\mu})
\end{equation*}
\begin{equation*}
    \bf{x}(0) = \bf{y}
\end{equation*}

has a unique solution $\bf{u}(t,\bf{y},\bf{\mu})$ with $\bf{u}$ $\in C^1(G)$ where $G = [-a,a] \times N_{\delta}(x_0) \times N_{\delta}(\mu_0)$.
\end{thm}

Now we state the following theorem,

\begin{thm}
\label{thm:gs1} 
Consider system \eqref{eq:patch_model_uni_patch1}. Then, under the parametric restrictions, $k_p<  \dfrac{(1+\alpha\xi)(\epsilon_1+\tilde{\delta})}{\epsilon_1-\tilde{\delta}},
$ \ we have the solution $ (x_{1},y_{1})$ to \eqref{eq:patch_model_uni_patch1} is globally asymptotically stable.
\end{thm}

\begin{proof}

We proceed by contradiction. WLOG assume $(x_{1}, y_{1})$ is not globally asymptotically stable, and large data does not converge to the $(x^{*}_{1}, y^{*}_{1})$ state.
Under the parametric restrictions of Lemma \ref{lem:p1s} and \ref{lem:p2s}, $k_p< \frac{(1+\alpha\xi)(\epsilon_1+\tilde{\delta})}{\epsilon_1-\tilde{\delta}} $ and $k_p<\frac{(1+\alpha\xi)(\epsilon_1+\bar{\delta})}{\epsilon_1-\bar{\delta}}$, we have that $\Bar {{y_1}} \leq y_{1} \leq \tilde{y_1}$, where both $(\Bar {{x_1}},\Bar {{y_1}}),$ and $(\tilde{x_1},\tilde{y_1})$, are globally asymptotically stable, and by direct comparison, $\Bar {{y_1}} \leq y_{1} \leq \tilde{y_{1}}$. We next choose a sequence $\{\xi_{n}, \alpha_{n} \}$ s.t., for a given parametric set $\{\xi, \alpha \}$ see Table \ref{Table:2}, where we have global stability of $(\tilde{x_1},\tilde{y_1})$ and $(\Bar{x_1},\Bar{y_1})$, and furthermore,

\begin{equation}
    \lim_{n \rightarrow \infty} \xi_{n} \alpha_{n} \rightarrow \alpha \xi, \ \lim_{n \rightarrow \infty} \xi_{n}  \rightarrow 0.
\end{equation}

For example $\{\xi_{n}=\frac{\xi}{n} \}$, $\{\alpha_{n}=n \alpha  \}$, would suffice.
Now, since the dependence of $(\tilde{x_1},\tilde{y_1})$ and $(\Bar{x_1},\Bar{y_1})$, on the parameters $\xi, \alpha$, is smooth, via standard $C^{1}$ theory of dynamical systems, with respect to parameters via Theorem \ref{thm:dp1}, we can consider a sequence of solutions, $(\tilde{x^{n}_1},\tilde{y^{n}_1})$ and $(\Bar{x^{n}_1},\Bar{y^{n}_1})$, that have the parameters $\{\xi, \alpha \}$ replaced by $\{\xi_{n}, \alpha_{n} \}$. Via comparison we have,

\begin{equation}
   \Bar {{y_1}} \leq \Bar {{y^{n}_1}} \leq y_{1} \leq \tilde{y^{n}_1} \leq \tilde{y_1}.
\end{equation}

Clearly, by construction we have that in the limit,

\begin{equation}
   \Bar {{y_1}} = \lim_{n \rightarrow \infty} \Bar {{y^{n}_1}} = \lim_{n \rightarrow \infty} \tilde{y^{n}_1} 
\end{equation}

Thus we take limits to yield,

\begin{equation}
   \Bar {{y_1}} = \lim_{n \rightarrow \infty} \Bar {{y^{n}_1}} \leq y_{1} \leq \lim_{n \rightarrow \infty} \tilde{y^{n}_1} = \lim_{n \rightarrow \infty} \Bar {{y^{n}_1}} = \Bar {{y_1}}.
\end{equation}

It also follows that,

\begin{equation}
   \Bar {{x_1}} \geq \Bar {{x^{n}_1}} \geq x_{1} \geq \tilde{x^{n}_1} \geq \tilde{x_1}.
\end{equation}

Thus, taking limits via the same argument as earlier,

\begin{equation}
   \Bar {{x_1}} = \lim_{n \rightarrow \infty} \Bar {{x^{n}_1}} \geq x_{1} \geq \lim_{n \rightarrow \infty} \tilde{x^{n}_1} = \lim_{n \rightarrow \infty} \Bar {{x^{n}_1}} = \Bar {{x_1}}.
\end{equation}


Thus $(x_{1},y_{1})$ will be driven to the globally stable $(\Bar {{x^{*}_1}}, \Bar {{y^{*}_1}})$ state, which can be made arbitrarily close to the $(x^{*}_{1}, y^{*}_{1})$ state; see Figure \eqref{nullcline_plot_comp}- this yields a contradiction, and the theorem is proved.

\end{proof}

\begin{table}[H]
\caption{Prey $\&$ Predator levels 
with $``\xi_{n}" \& ``\alpha_{n}"$}\label{Table:2}
	\scalebox{0.92}{
{\begin{tabular}{|c|c|c|c|}

		\hline
		Value of $(\xi_{n}, \alpha_{n})$  &  $(\tilde x_1,\tilde y_1):\eqref{eq:patch_model_uni_patch_sup}$  & $(x_1,y_1):\eqref{eq:patch_model_uni-i}$ & $(\bar{x_1},\bar{y_1}):\eqref{eq:patch_model_uni_patch_sub}$  \\
		& &  &   \\ \hline \hline
	$(\xi_{1}, \alpha_{1})$=(0.1,0.02) & (1.09609,0.948245)    & (1.38565,0.733429)  & (1.65035,0.463694)  \\
		&  &  &  \\ \hline
		$(\xi_{2}, \alpha_{2})$=(0.01,0.2) & (1.58209,0.539962) & (1.62388,0.49382) &(1.65035,0.463694)   \\
		&  &     &  \\ \hline
  $(\xi_{3}, \alpha_{3})$=(0.001,2) & (1.64336,0.471715)   & (1.64771,0.466738)   &  (1.65035,0.463694) \\
		&  &     &  \\ \hline
  $(\xi_{4}, \alpha_{4})$=(0.0001,20) & (1.64965,0.4645)   & (1.65009,0.463998)   &  (1.65035,0.463694) \\
		&  &     &  \\ \hline
		\end{tabular}}
		}

\vspace{0.25cm}
The parameter set used in Table \ref{Table:2} is $k_p=2,\epsilon_1 = 0.45, \epsilon_2 =0.8, \delta_1= 0.08, \delta_2=0.1, q_1=0.2 , q_2=0.3$.
Here $\left(\xi_{n}, \alpha_{n}\right)$ is chosen as $(\xi_{n} = \frac{\xi}{10^{n-1}}, \alpha_{n} = 10^{n-1} \alpha)$. We see very quick convergence to the globally stable sub-solution - within 3 approximations, error $\approx 10^{-4}$.
  \end{table}
  \begin{figure}
  \begin{subfigure}{.32\textwidth}
\centering
  \includegraphics[width=5.2cm, height=5.7cm]{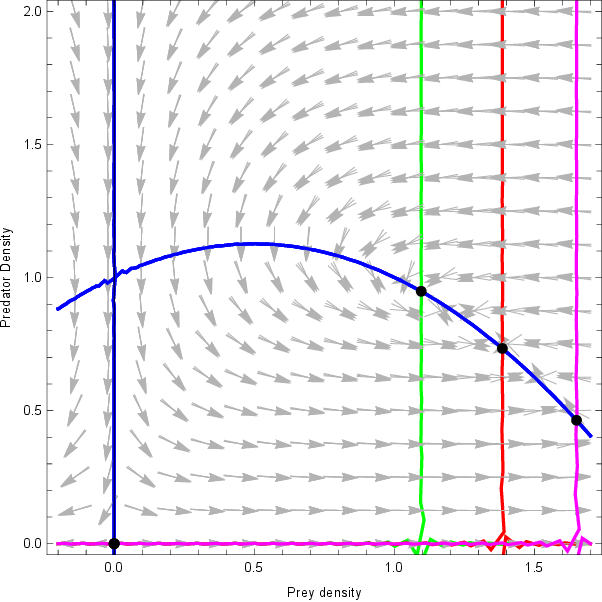}
 \subcaption{ $\xi_{1} = 0.1, \alpha_{1}=0.02
$ }
 \label{fig:table_first_entry}
  \end{subfigure}
  \begin{subfigure}{.32\textwidth}
  \centering
  \includegraphics[width= 5.2cm, height=5.7cm]{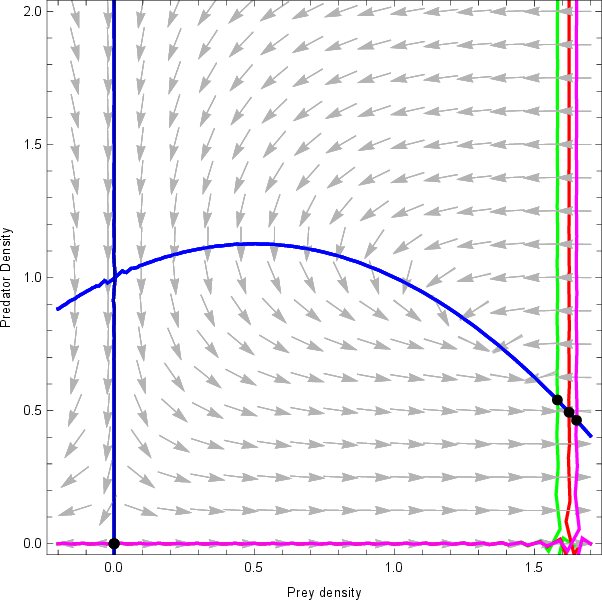}
 \subcaption{$\xi_{2} = 0.01, \alpha_{2}=0.2$}
  \label{fig:table_second_entry}
 \end{subfigure}
  \begin{subfigure}{.32\textwidth}
  \centering
  \includegraphics[width= 5.2cm, height=5.7cm]{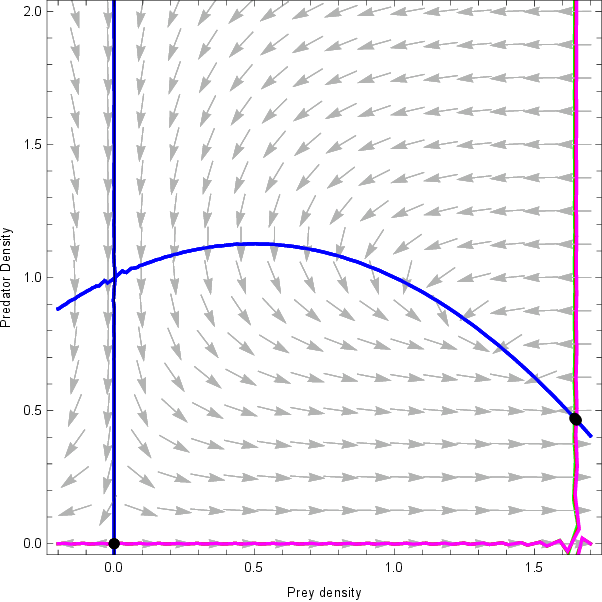}
 \subcaption{$\xi_{3} = 0.001, \alpha_{3}= 2$}
  \label{fig:table_third_entry}
 \end{subfigure}
 \caption{ 
The phase plots corresponding to the first, second, and third row of Table \ref{Table:2} are depicted in figure \eqref{fig:table_first_entry}, \eqref{fig:table_second_entry}, and \eqref{fig:table_third_entry}, respectively. In these plots, the blue curve represents the prey nullcline, while the green, red, and magenta curves represent the nullclines of $\tilde{y_1}$, $y_1$, and $\bar{y_1}$, respectively. The parameter set used is $k_p=2, \epsilon_1 = 0.45, \epsilon_2 =0.8, \delta_1= 0.08, \delta_2=0.1, q_1=0.2 , q_2=0.3 $. }
  \label{nullcline_plot_comp}
 \end{figure}
 \subsubsection{Pest extinction state in patch \texorpdfstring {$\Omega_2$}{Lg}}

\ 

The equations representing dynamics in patch $\Omega_2$ are given by: 

\begin{equation} \label{eq:patch_model_patch2_uni}
\begin{aligned}
 \Dot{x_2}& = x_2 \left (1- \frac{x_2}{k_c} \right)-\frac{x_2  y_2}{1+x_2} \\
\Dot{y_2}& = \epsilon_2 \left ( \frac{x_2}{1+x_2} \right ) \ y_2 - \delta_2 y_2 + q_2 y_1  
\end{aligned}
\end{equation}

\

Now plugging the value of $y_1^*$ using \eqref{x2_0_eq_1_uni} and \eqref{x2_0_eq_4_uni} we have, 

\begin{equation} \label{eq:patch_model_patch2_reduced_uni}
\begin{aligned}
 \Dot{x_2}& = x_2 \left (1- \frac{x_2}{k_c} \right)-\frac{x_2  y_2}{1+x_2} = x_2 \tilde f(x_2,y_2) \\
\Dot{y_2}& = \epsilon_2 \left ( \frac{x_2}{1+x_2} \right ) \ y_2 - \delta_2 y_2 - G = y_2 \tilde g(x_2,y_2)- G
\end{aligned}
\end{equation}

\ 

Here, $G$ is analogous to constant predator stocking 
  \cite{brauer1981constant} and is given by,
\begin{equation*}
    G = -q_2 y_1^* =  - \epsilon_1 q_2 \left ( \dfrac{1+\alpha\xi-\xi}{\epsilon_1- (\delta_1+q_1)}\right)  \left ( 1-\dfrac{(\delta_1+q_1)(1+\alpha\xi)-\epsilon_1\xi}{k_p(\epsilon_1-(\delta_1+q_1))}    \right) 
\end{equation*}

We state the following lemma,

\begin{lemma}
\label{lem:pest_extinction_crop_field}
Consider the system \eqref{eq:patch_model_patch2_reduced_uni}. Then under the parametric restriction $\epsilon_2 > \delta_2\left(1-\dfrac{1}{k_c}\right) \ and, \\  \delta_2 < \epsilon_1 q_2 \left ( \dfrac{1+\alpha\xi-\xi}{\epsilon_1- \bar{Q}}\right)  \left ( 1-\dfrac{\bar{Q}(1+\alpha\xi)-\epsilon_1\xi}{k_p(\epsilon_1-\bar{Q})}    \right)  $ where, $\bar{Q} = \delta_1+q_1$, we have that the pest free state $(0,y^{*}_{2})$ in the crop field to 
\eqref{eq:patch_model_patch2_reduced_uni}, is globally attracting for any positive $(x_{2}(0),y_{2}(0))$. 
\end{lemma}

\begin{proof}
From $\dot x_2 = 0 $,

\begin{equation}
    y_2^*  =  \left (1- \frac{x_2^*}{k_c} \right) \left(1+x_2^*\right) 
    \label{y2_x2}
\end{equation}

 and from $\dot y_2 = 0$ we have, 
\begin{equation*}
  y_2^* \tilde g(x_2^*,y_2^*) =  G \implies y_2^* \left(   \epsilon_2 \left ( \frac{x_2^*}{1+x_2^*} \right ) - \delta_2 \right) =  G = -q_2 y_1^*
  \end{equation*}

\

The critical stocking rate for which equilibrium reaches the pest extinction  is given by \cite{brauer1981constant},

\begin{equation*}
  -G_{critical} = -\tilde g(0,1) \implies G_{critical} = -\delta_2 
\end{equation*}

\

Now, the pest extinction state in the crop field patch $\Omega_2$ is possible for all initial conditions if,

\begin{equation*}
\epsilon_2 > \delta_2\left(1-\dfrac{1}{k_c}\right) \  \text{and,} \ G < G_{critical} \implies
\boxed{
    y_2^* > 1}
\end{equation*}

\begin{equation*}
\setlength{\jot}{10pt}
\begin{aligned}
& \implies  - \epsilon_1 k_4 \left ( \dfrac{1+\alpha\xi-\xi}{\epsilon_1- (\delta_1+q_1)}\right)  \left ( 1-\dfrac{(\delta_1+q_1)(1+\alpha\xi)-\epsilon_1\xi}{k_p(\epsilon_1-(\delta_1+q_1))}    \right)  < -\delta_2   \\
&  \implies  \epsilon_1 q_2 \left ( \dfrac{1+\alpha\xi-\xi}{\epsilon_1- \bar{Q}}\right)  \left ( 1-\dfrac{\bar{Q}(1+\alpha\xi)-\epsilon_1\xi}{k_p(\epsilon_1-\bar{Q})}    \right)  > \delta_2  \ \text{Where,}\  \bar{Q} = \delta_1+q_1\\
& \implies \xi^2 (\epsilon_1-\alpha \bar{Q})(\alpha-1) + \xi \{(\alpha-1)(k_p(\epsilon_1-\bar{Q})-\bar{Q}) + (\epsilon_1-\alpha \bar{Q})\}- \frac{k_p \delta_2(\epsilon_1-\bar{Q})^2}{\epsilon_1 q_2} +k_p (\epsilon_1-\bar{Q}) - \bar{Q} > 0 \\
&  \implies \xi^2 (\epsilon_1-\alpha \bar{Q})(1-\alpha) - \xi \{(\alpha-1)(k_p(\epsilon_1-\bar{Q})-\bar{Q}) + (\epsilon_1-\alpha \bar{Q})\} + \frac{k_p \delta_2(\epsilon_1-\bar{Q})^2}{\epsilon_1 q_2} - k_p (\epsilon_1-\bar{Q}) + \bar{Q} < 0 \\
\end{aligned}
\end{equation*}

\

The above inequality is quadratic in $\xi$; on comparing with a standard quadratic equation, we have, 
\begin{equation*}
\begin{aligned}
 & A = (\epsilon_1-\alpha \bar{Q})(1-\alpha),  \  B = - \{ ( \alpha-1)(k_p(\epsilon_1-\bar{Q})-\bar{Q}) +(\epsilon_1-\alpha \bar{Q})\}, \ C = \frac{k_p \delta_2(\epsilon_1-\bar{Q})^2}{\epsilon_1 q_2} - k_p (\epsilon_1-\bar{Q}) + \bar{Q} \\
& \text{If}\   A > 0 , \ B^{2}- 4 AC >0 \ \implies \xi_1 < \xi < \xi_2 \ \  \text{where,} \ \ \xi_1 = \dfrac{-B - \sqrt{B^2-4AC}}{2A} \ \text{and} \ \ \xi_2 = \dfrac{-B + \sqrt{B^2-4AC}}{2A}
\end{aligned}
\end{equation*}
\end{proof}

\begin{rem}
Lemma \ref{lem:pest_extinction_crop_field} can be rephrased, and the parametric conditions can be expressed in terms of the carrying capacities of both the prairie strip and the crop field.
    Consider the system \eqref{eq:patch_model_uni-i}. Then, under the parametric restrictions, $k_{c} > 0 $ 
    and, 
   \begin{equation*}
    k_{p} >\dfrac{\epsilon_1 q_2 \left(\epsilon_1\xi-\bar{Q}(1+\alpha\xi)\right)\left( 1+\alpha\xi-\xi\right)}{(\epsilon_1-\bar{Q})\left( \delta_2(\epsilon_1-\bar{Q})- \epsilon_1 q_2 (1+\alpha\xi-\xi)\right) }
    \end{equation*}
     where, $\bar{Q} = \delta_1+q_1$, we have that the pest free state $(0,y^{*}_{2})$ in the crop field to 
\eqref{eq:patch_model_patch2_reduced_uni}, is globally attracting for any positive $(x_{2}(0),y_{2}(0))$. 
\end{rem}
\begin{rem}
  Figure \eqref{region_plots} represents the relation between the carrying capacity of the prairie strip and the quantity of AF, along with the drift rates, while keeping other parameters fixed so that the pest-free state in the crop field is globally attracting. Note that the plot of $k_p \ \text{vs.} \ q_2$ is not shown here because region (I) is independent of $q_2$. Therefore, region (III), being the intersection of the region (I) and (II), can not be plotted. 
  
  Region I: $k_p<  \frac{(1+\alpha\xi)(\epsilon_1+\tilde{\delta})}{\epsilon_1-\tilde{\delta}} $, 
  
  Region II: $k_{p} >\frac{\epsilon_1 q_2 \left(\epsilon_1\xi-\bar{Q}(1+\alpha\xi)\right)\left( 1+\alpha\xi-\xi\right)}{(\epsilon_1-\bar{Q})\left( \delta_2(\epsilon_1-\bar{Q})- \epsilon_1 q_2 (1+\alpha\xi-\xi)\right) }$, 
  
  Region III: $ \frac{\epsilon_1 q_2 \left(\epsilon_1\xi-\bar{Q}(1+\alpha\xi)\right)\left( 1+\alpha\xi-\xi\right)}{(\epsilon_1-\bar{Q})\left( \delta_2(\epsilon_1-\bar{Q})- \epsilon_1 q_2 (1+\alpha\xi-\xi)\right) }< k_p < \frac{(1+\alpha\xi)(\epsilon_1+\tilde{\delta})}{\epsilon_1-\tilde{\delta}}$.
  \end{rem}

 \begin{figure}
  \begin{subfigure}{.44\textwidth}
  \centering
  \includegraphics[width= 7cm, height=6cm]{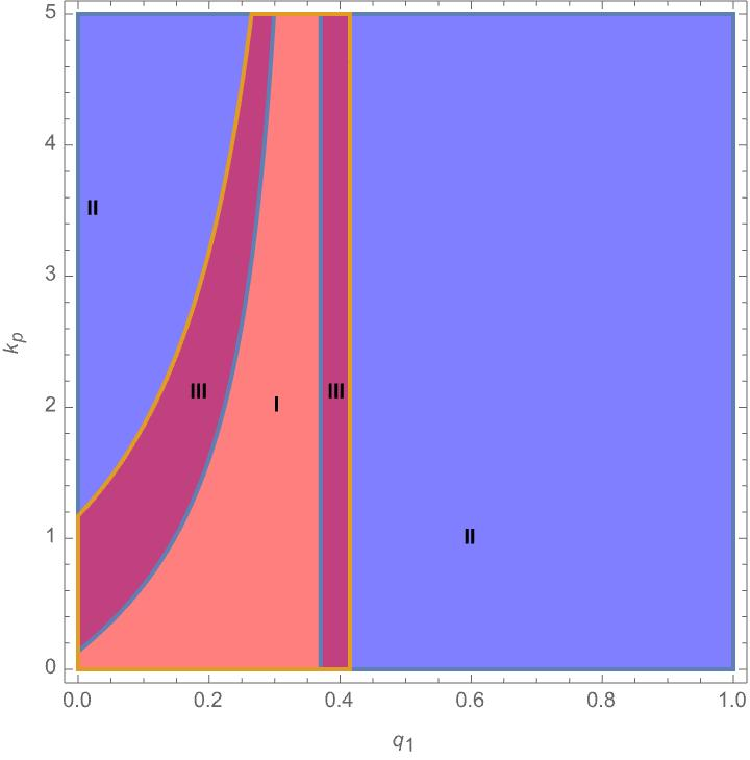}
 \subcaption{$k_p$\text{vs.}  $q_1$}
  \label{fig:kp vs k2_none}
 \end{subfigure}
  \begin{subfigure}{.44\textwidth}
  \centering
\includegraphics[width= 7cm, height=6cm]{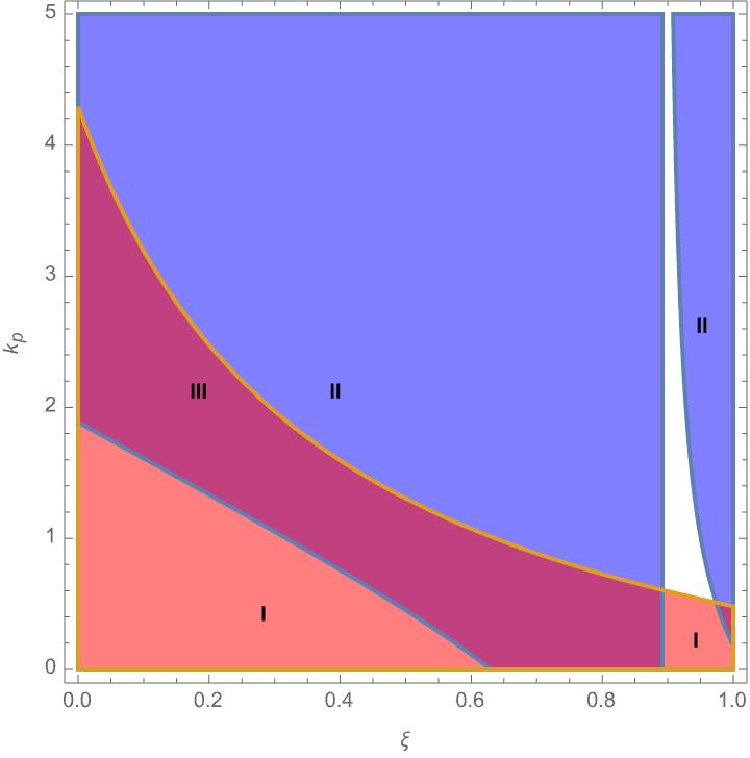}
 \subcaption{$k_p \text{vs.}\  \xi$}
  \label{fig:kp vs xi_none}
 \end{subfigure}
 \caption{  The red color represents region (I), the blue color represents region (II), and the magenta color represents region (III).
 In \eqref{fig:kp vs k2_none}, the relationship between $k_p \ \& \  q_1$ is depicted for $\xi = 0.1, q_2=0.3$. In \eqref{fig:kp vs xi_none}, the relationship between $k_p \ \& \  \xi$ is shown with $q_1=0.2, q_2=0.3$.  The fixed parameter set used is $k_c=20, \alpha = 0.02, \epsilon_1 = 0.45, \epsilon_2 =0.8, \delta_1= 0.08, \delta_2=0.1, q_2 = 0.3$.
}

 \label{region_plots}
 \end{figure}

 \subsection{Global Stability in Classical AF Model}
A special case of the earlier results, Lemmas \ref{lem:p1s}, \ref{lem:p2s} and Theorem \ref{thm:gs1} in particular, enable us to improve the criterion for global stability of the classical AF model, \eqref{eq:patch_model_uni_patch1}. Note if there is no dispersal assumed, that is $q_{1}=q_{2}=0$, we have that \eqref{eq:patch_model_uni_patch1} reduces to the classical AF model in \cite{SP07}.

\begin{rem}
In \cite{SP07}, it is stated that the coexistence equilibrium to the classical AF model is globally stable if $\gamma < 1 + \alpha \xi$. The condition $\gamma < 1 + \alpha \xi$ is highly restrictive and limits pest carrying capacity to essentially the amount of AF available. geometrically, this condition shifts the ``hump" or maximum in the prey nullcline over to the second quadrant. Thus, one is not modeling a one-hump prey nullcline any longer.
\end{rem}

We next state a corollary that improves this restriction, allowing for larger $\gamma$,

\begin{cor}
\label{cor:g1}
     Consider the classical AF model in \cite{SP07}. Then, under the parametric restrictions 
   $\gamma < \dfrac{(1+\alpha\xi)(\beta + \hat{\delta})}{\beta - \hat{\delta}}$, \ where, $\hat{\delta} = \delta -\dfrac{\beta \xi}{1+\alpha\xi}$, 
      we have the solution $ (x,y)$ to the classical AF model is globally asymptotically stable.
\end{cor}
\begin{proof}
    The proof is a direct application of Lemmas \ref{lem:p1s}, \ref{lem:p2s} and Theorem \ref{thm:gs1}, applied to the special case when $q_{1}=q_{2}=0$.
\end{proof}

\begin{rem}
Under our derived condition for global stability, the carrying capacity $\gamma$, can be as large as one wants, as long as $\left(\frac{\beta + \hat{\delta}}{\beta - \hat{\delta}}\right)$, is scaled to be as large as one requires. All in all, the derived condition is a clear improvement over the condition derived in \cite{SP07}, as long as $\beta > \hat{\delta}$, because $\left(\frac{\beta + \hat{\delta}}{\beta - \hat{\delta}}\right) > 1$, this follows geometrically as long as there exists a positive interior equilibrium, then $\hat{\delta} > 0$.
    \end{rem}

\section{Additional Food Patch Model Formulation - Dispersal}
\label{dispeersal af model}

\begin{figure}
\includegraphics[width = 5cm]{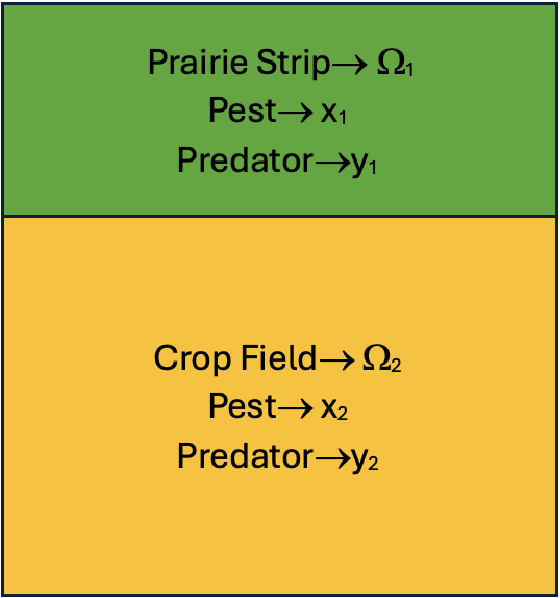}
\caption{Visual depiction of the underlying assumptions regarding patches $\Omega_1 \ \& \ \Omega_2$. }
\label{fig:patch_assumption}
\end{figure}
We assume that our landscape is divided into two sub-units, a crop field, which is the larger unit, and a prairie strip, which is the smaller unit. These are the two patches that make up our landscape. We will refer to the prairie strip as $\Omega_1$ and the crop field as $\Omega_2$. We assume that the prairie strip, by design, possesses row crops that will provide alternative/additional food to an introduced predator, which would enhance its efficacy in targeting a pest species that resides primarily in the crop field. The crop field has no AF. 
We further assume that the pest does not move between the patches. That is, the pest prefers its primary host, the crop in the crop field; however, there could be some presence of the pest in the prairie strip due to long-range dispersal effects such as wind (not modeled herein) or the presence of certain row crops in the strip, as alternative hosts. 

Thus, for the patch $\Omega_1$, we assume the quantity of AF as $\xi$ and the quality of AF as $\frac{1}{\alpha}$. Thus, essentially, in $\Omega_{1}$, we have an AF-driven predator-pest system for a pest density $x_{1}$ and a predator density $y_{1}$. We choose the classical Holling type II functional response for the functional response of the pest. In $\Omega_2$, since there is no AF, we have a classical predator-prey (pest) dynamics for a pest density $x_{2}$ and a predator density $y_{2}$.  We assume the predators will disperse between $\Omega_{1}$ and $\Omega_{2}$, via linear dispersal see Figure \eqref{fig:patch_assumption}. We assume that the net predator dispersal from $\Omega_{1}$ into $\Omega_{2}$ takes place at rate $k_{2}$ while the net predator dispersal from $\Omega_{2}$ into $\Omega_{1}$ takes place at rate $k_{4}$. Dispersal (and linear dispersal in particular) has been modeled in the literature in two main directions. In one approach, even if dispersal is asymmetric, the net (dispersing) population is conserved \cite{arditi2018asymmetric}. On the other, this may not be true, where a relative dispersal rate for each patch is considered \cite{messan2015two}. 
In
our current modeling framework, we can reproduce either of the above by choosing dispersal rates
accordingly. 

\begin{equation} \label{eq:patch_model2}
\begin{aligned}
\Dot{x_1}& = x_1 \left (1-\frac{x_1}{k}\right )-\frac{x_1 y_1}{1+x_1+\alpha \xi} \\
 \Dot{y_1}& = \epsilon_1 \left ( \frac{x_1+\xi}{1+x_1+\alpha \xi} \right)\ y_1 - \delta_1 y_1 + k_2 \left (y_2-y_1 \right)  \\
 \Dot{x_2}& = x_2 \left (1- \frac{x_2}{k} \right)-\frac{x_2  y_2}{1+x_2} \\
\Dot{y_2}& = \epsilon_2 \left ( \frac{x_2}{1+x_2} \right ) \ y_2 - \delta_2 y_2 + k_4 \left (y_1-y_2  \right)  
\end{aligned}
\end{equation}

\par There is a strong biological motivation for the asymetric dispersal case, $k_{2}>k_{4}$. We assume $k_{2} = k_{4} +  \tilde{k_{2}} \xi$. That is, we assume that the AF is a driver of additional dispersal of the predator out of $\Omega_{1}$ into $\Omega_{2}$, also the AF can also act as an attractant for predators wanting to disperse out of $\Omega_{2}$ into $\Omega_{1}$. The above assumptions yield,

\begin{equation} 
\label{eq:patch_model2pp}
\begin{aligned}
 \Dot{x_1}& = x_1 \left (1-\frac{x_1}{k}\right )-\frac{x_1 y_1}{1+x_1+\alpha \xi} \\
 \Dot{y_1}& = \epsilon_1 \left ( \frac{x_1+\xi}{1+x_1+\alpha \xi} \right)\ y_1 - \delta_1 y_1 + k_4 \left (y_2-y_1 \right) + \tilde{k_{2}} \xi \left (y_2-y_1 \right) \\
 \Dot{x_2}& = x_2 \left (1- \frac{x_2}{k} \right)-\frac{x_2  y_2}{1+x_2} \\
\Dot{y_2}& = \epsilon_2 \left ( \frac{x_2}{1+x_2} \right ) \ y_2 - \delta_2 y_2 + k_4 \left (y_1-y_2  \right)  
\end{aligned}
\end{equation}

Here the $\tilde{k_{2}} \xi$, is the ``extra" or increased dispersal driven by the predator energized by the AF - this motivates the $\boxed{-\tilde{k_{2}} \xi y_{1}}$ term. Furthermore, AF can also act as an attractant for predators wanting to disperse out of $\Omega_{2}$ into $\Omega_{1}$ to replenish on the AF and recharge before targeting the pest in the crop field again - this motivates the $\boxed{\tilde{k_{2}} \xi y_{2}}$ term. Since prairie strips with the AF also serve as refuge for the predators, there is strong biological motivation for the predators to return to the AF patch \cite{snyder2019give, S99}.
If there was no AF, $\xi = 0$, then we are in the symmetric dispersal case.

\subsection{Predator Blowup prevention} \label{blowup-prevention}
As seen in the earlier section, the Blow-up phenomenon is seen in many AF models in the literature. That is, they enable infinite time \emph{blow-up} in the predator population, so must be applied to bio-control scenarios with caution. A basic AF model with a monotone pest-dependent response is unable to yield $(0,y^{*})$, the pest extinction state for $y^{*} < \infty$, and what is seen (see Table \ref{Table:1}) that if $\xi > \xi_{critical}$, $(x,y) \rightarrow (0,\infty)$. Thus, blow-up prevention and the attainability of the (finite) $(0,y^{*})$ pest extinction state is much sought after. To this end, the introduced patch model is useful. 
\subsubsection{Asymmetric dispersal}
\ 

We are considering the case of asymmetric dispersal. 
We state the following theorem,
\begin{thm}
\label{thm:t1}
Consider the additional food patch model \eqref{eq:patch_model2pp} with $\epsilon_{1} > \delta_{1}, \epsilon_{2} > \delta_{2}$. If $\tilde{k_{2}}=k_{4}=0$, that is there is no dispersal between patches and $\xi > \xi_{critical} =  \dfrac{\delta_1}{\epsilon_1 - \alpha \delta_1}$, then the predator population blows up in infinite time. However with any dispersal, $k_{4}>0$, and $\tilde{k_{2}}$ s.t. $ \tilde{k_{2}} \xi > \frac{\xi}{1+\alpha \xi}  - \delta_1 $
and $\delta_{1}+\delta_{2} > \frac{\epsilon_{1} \xi}{1+ \alpha \xi} > \delta_{1}$, the predator population remains bounded for all time.
\end{thm}

\begin{proof}
    See Appendix \eqref{thm proof:t1}.
\end{proof}

\begin{figure}
  \begin{subfigure}{.44\textwidth}
\centering
  \includegraphics[width= 7cm, height=6cm]{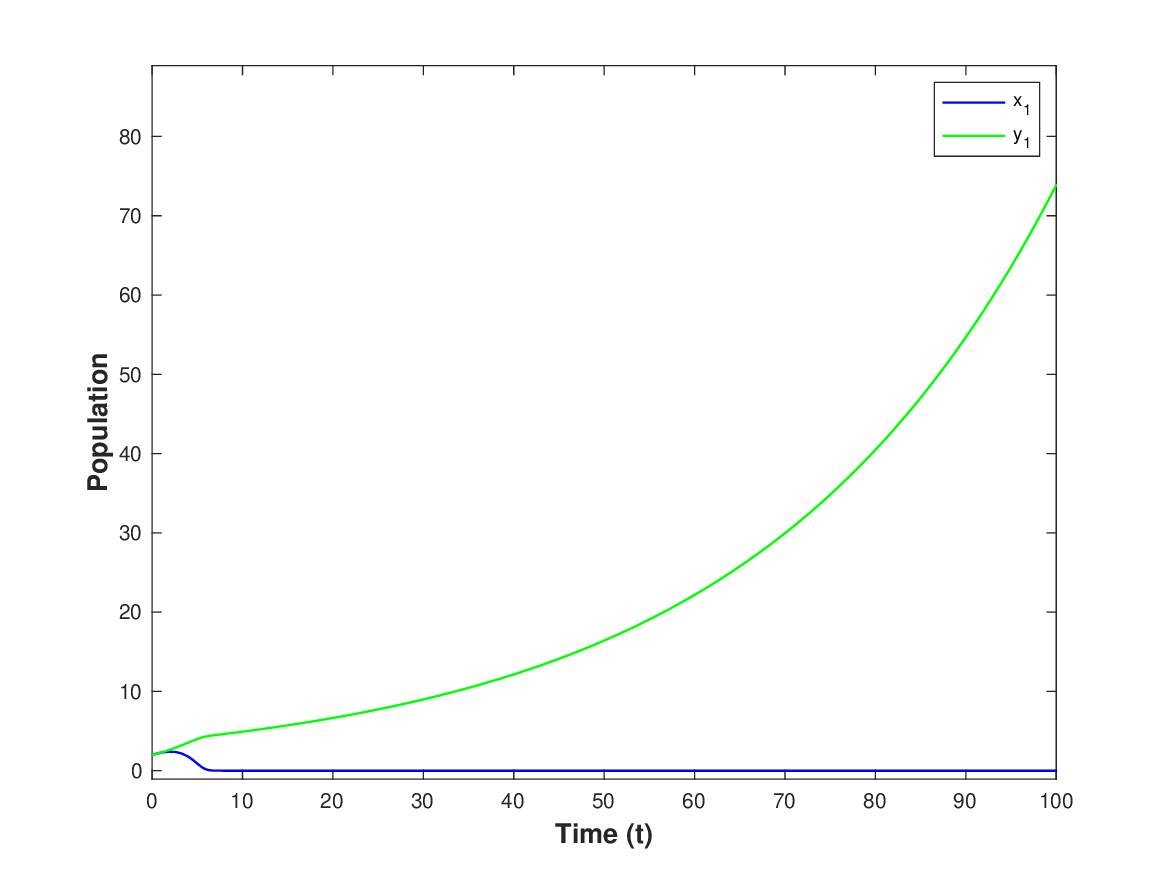}
 \subcaption{ $\tilde{k_2}= k_4=0$}
 \label{fig:blowup_no_patch}
  \end{subfigure}
  \begin{subfigure}{.44\textwidth}
  \centering
  \includegraphics[width= 7cm, height=6cm]{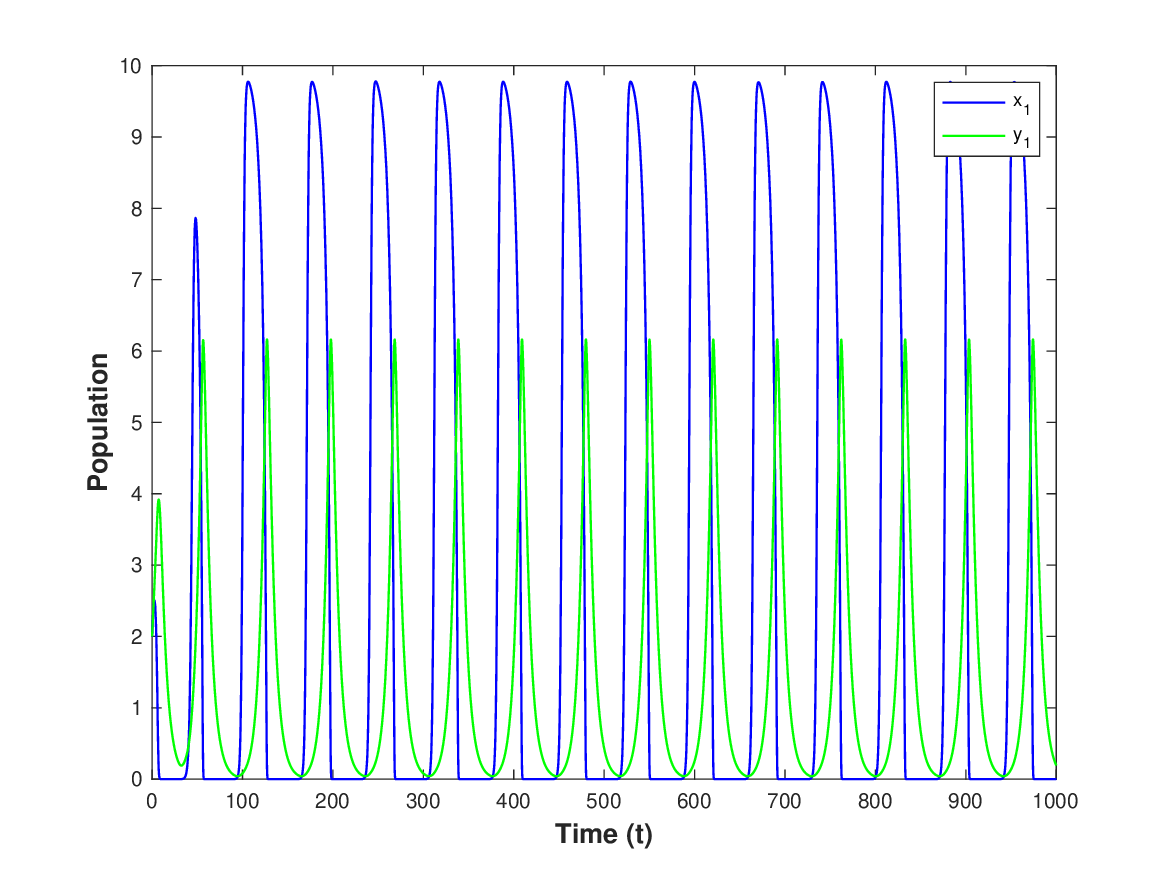}
 \subcaption{ $\tilde{k_2}=0.6, \ k_4=0.1$}
  \label{fig:blowup_prevention_patch_1}
 \end{subfigure}
 \caption{ The parameter set used is $k=10,\alpha=0.2, \xi=0.65, \epsilon_1 =\epsilon_2 =0.4, \delta_1= \delta_2=0.2 \ \text{with I.C.} = [2,2,2,2]$. }
 \end{figure}

\begin{figure}
  \begin{subfigure}{.44\textwidth}
  \centering
\includegraphics[width= 7cm, height=6cm]{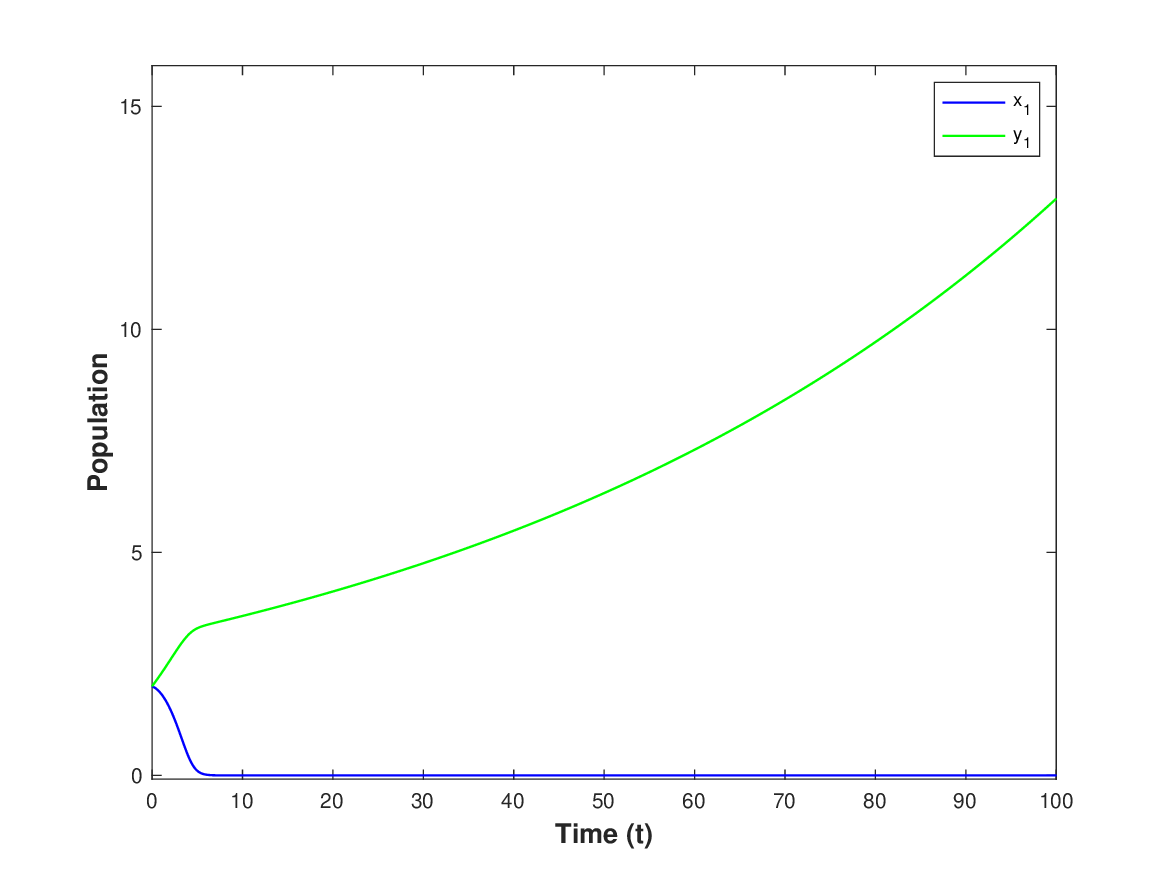}
 \caption{ $ k_4=0$}
 \label{fig:blowup}
  \end{subfigure}
  \begin{subfigure}{.44\textwidth}
  \centering
  \includegraphics[width= 7cm, height=6cm]{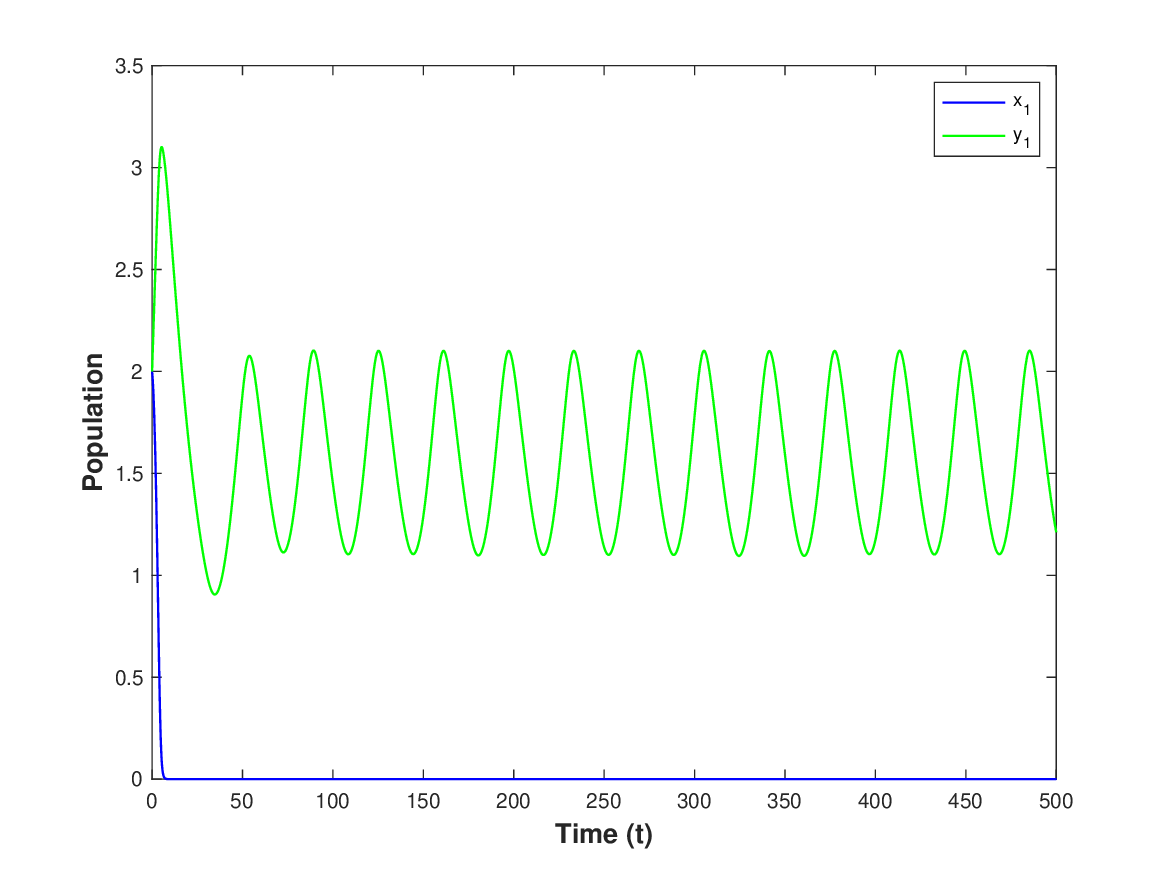}
 \caption{ $ k_4=0.12$}
 \label{fig:blowup_2}
 \end{subfigure}
 \caption{The parameter set used is $k=5,\alpha=0.2, \xi=0.6, \epsilon_1 =\epsilon_2 =0.4, \delta_1= \delta_2=0.2 \ \text{with I.C.} = [2,2,2,2]$.}
\end{figure}



\subsubsection{Symmetric dispersal}
\

Herein, we derive conditions under which blow-up prevention is possible in the case of symmetric dispersal. Again, the prey populations $x_{1},x_{2}$ are always bounded via comparison to the logistic equation. In the symmetric dispersal case, if $\xi > \xi_{critical} =  \frac{\delta_1}{\epsilon_1 -\alpha \delta_1}$, then blow-up in infinite time follows from \cite{PWB23}. This is equivalent to saying, $\frac{\xi}{1+\alpha \xi} > \delta_{1}$.

We state the following theorem,

\begin{thm}
\label{thm:t1sd}
Consider the additional food patch model \eqref{eq:patch_model2pp} with $\epsilon_{1} > \delta_{1}, \epsilon_{2} > \delta_{2}$. If $\tilde{k_{2}}= k_{4}=0$, that is there is no dispersal between patches and $\xi > \xi_{critical} =  \dfrac{\delta_1}{\epsilon_1 - \alpha \delta_1}$, then the predator population blows up in infinite time. However with symmetric dispersal, that is $k_{4}>0$ s.t., $\delta_1 + \delta_2 > \frac{\xi}{1+\alpha \xi} > \delta_{1}$, the predator population remains bounded for all time.
\end{thm}

\begin{proof}
    See Appendix \eqref{thm proof:t1sd}.
\end{proof}

\begin{rem}
Figure \eqref{fig:blowup_no_patch} shows the blow-up dynamics in predator population when $\xi> \xi_{critical}$ and no dispersal in \eqref{eq:patch_model2pp} i.e., $\tilde{k_2}=k_4=0$. Figure \eqref{fig:blowup_prevention_patch_1} shows the blow-up prevention in patch $\Omega_1$, when there is asymmetric dispersal between the patches, i.e., $\tilde{k_2},k_4 \neq 0$. 
Figure \eqref{fig:blowup} shows the blow-up dynamics in predator population when $\xi> \xi_{critical}$ and no dispersal in \eqref{eq:patch_model2pp}. Figure \eqref{fig:blowup_2} shows the blow-up prevention in patch $\Omega_1$ when there is symmetric dispersal between the patches, i.e., $\tilde{k_2} = 0, k_4 > 0$. 
In both the figures, the values of all other parameters remain constant, showing that the movement of predators with asymmetric as well as symmetric dispersal can prevent blow-up in predator populations.  The values of parameters are selected based on the parametric constraints outlined in Theorem \ref{thm:t1} and \ref{thm:t1sd}.

\end{rem}
 \subsection{Equilibrium Analysis }
We now consider the existence and local stability analysis of the biologically relevant equilibrium points for the system \eqref{eq:patch_model2}. The Jacobian matrix $(J)$ for the  additional food patch model \eqref{eq:patch_model2} is given by: 
 
\begin{equation} 
J = \begin{bmatrix}
1 - \dfrac{2 x_1}{k} -  \dfrac{y_1  \left(1+ \alpha \xi \right) } {\left(1+x_1+\alpha \xi\right)^2}  & \dfrac{- x_1}{1+x_1+\alpha \xi} & 0 & 0 \vspace{0.25cm}
  \\ 
\dfrac{\epsilon_1  \left(1 + \left(\alpha - 1\right) \xi\right) \ y_1}{(1+x_1+\alpha \xi)^2} &  \dfrac{\epsilon_1 \left(x_1 + \xi\right)}{1+x_1+\alpha \xi} - \delta_1 - k_2 & 0 & k_2 
\\
0 & 0 & 1 - \dfrac{2 x_2}{k} - \dfrac{y_2}{(1+x_2)^2} & 
  \dfrac{- x_2}{1+x_2} \vspace{0.25cm}
  \\
  0 & k_4 & \dfrac{\epsilon_2 y_2}{(1+x_2)^2} & \dfrac{\epsilon_2 x_2}{1+x_2} - \delta_2 - k_4
  \end{bmatrix}
 \label{general_jacobian_k1_k3_0}
 \end{equation}
 
By evaluating this Jacobian matrix at each equilibrium point, we obtain the local stability conditions of  $ E_1, E_2, E_3$ and $E_4$.


\subsubsection{Pest free state in patch \texorpdfstring {$\Omega_2$}{Lg}}\mbox{}

\begin{lemma}
    The equilibrium point $E_1 = (x_1^*,y_1^*,0,y_2^*)$ exists if  $\xi < \dfrac{Q}{\epsilon_1 - \alpha Q}$ and $ \epsilon_1 > Q $ where, $ Q = \dfrac{\delta_1 \delta_2  + \delta_1 k_4  + k_2 \delta_2  }{\delta_2 + k_4 } $.
 \label{lem:E1_existence}
  \end{lemma}
  \begin{proof}
  See Appendix \eqref{proof_E1_existence}.
    \end{proof}
\begin{rem}
    A critical condition for the existence of a pest free state in the crop field ($\Omega_{2}$) is that $\epsilon_1 > Q$. Note $ Q = \dfrac{\delta_1 \delta_2  + \delta_1 k_4  + k_2 \delta_2  }{\delta_2 + k_4 } = \delta_1+ \dfrac{k_2 \delta_2}{\delta_2 + k_4}$.
    Thus $\delta_{1} \searrow \implies Q \searrow$. In this sense, even if $\epsilon_1 < Q$, and we do not have a pest free state, it could be created if one is able to bring down $\delta_{1}$ sufficiently enough so that now $\epsilon_1 > Q$.
    A similar argument with $\delta_{2}$ holds in the case of pest extinction in prairie strip ($\Omega_{1}$).
\end{rem}
\begin{lemma}
The equilibrium point $E_1 = (x_1^*,y_1^*,0,y_2^*)$ is conditionally locally asymptotically stable. 
\label{lem:E1_stability}
\end{lemma}
\begin{proof}
    See Appendix \eqref{lem proof:E1_stability}.
\end{proof}

\begin{figure}
\centering
\includegraphics[width = 10cm, height=7cm]{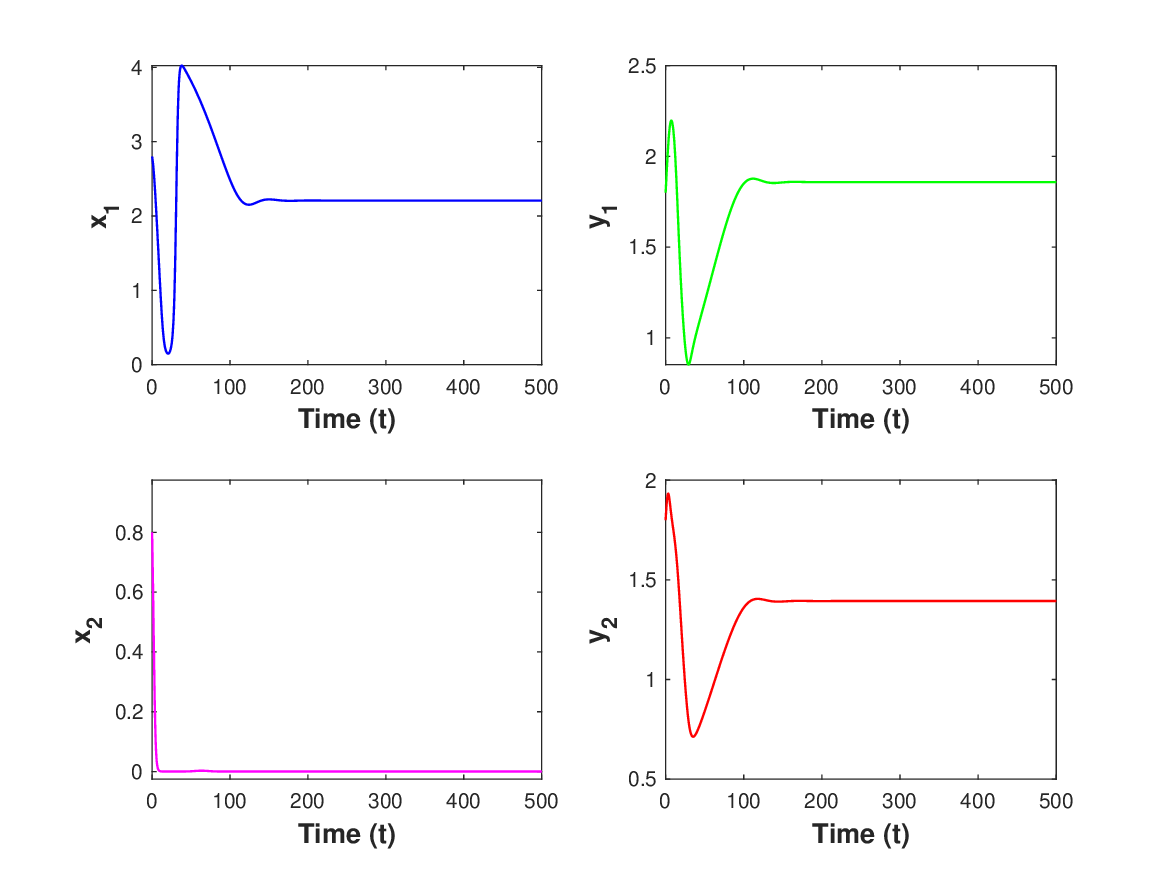}
\caption{The time series figure shows pest extinction in patch $\Omega_2$ with symmetric dispersal. The parameters used are $k=5, \alpha=0.2, \xi=0.6, \epsilon_1 =0.4, \epsilon_2 =0.2, \delta_1= 0.3, \delta_2=0.05,k_2=k_4=0.15$ with I.C.= $[2.8,1.8,0.8,1.8]$.}
\label{fig:x2_extinction_equal_dispersal}
\end{figure}

\begin{figure}
\begin{subfigure}{.48\textwidth}
   \centering
\includegraphics[width = 9cm]{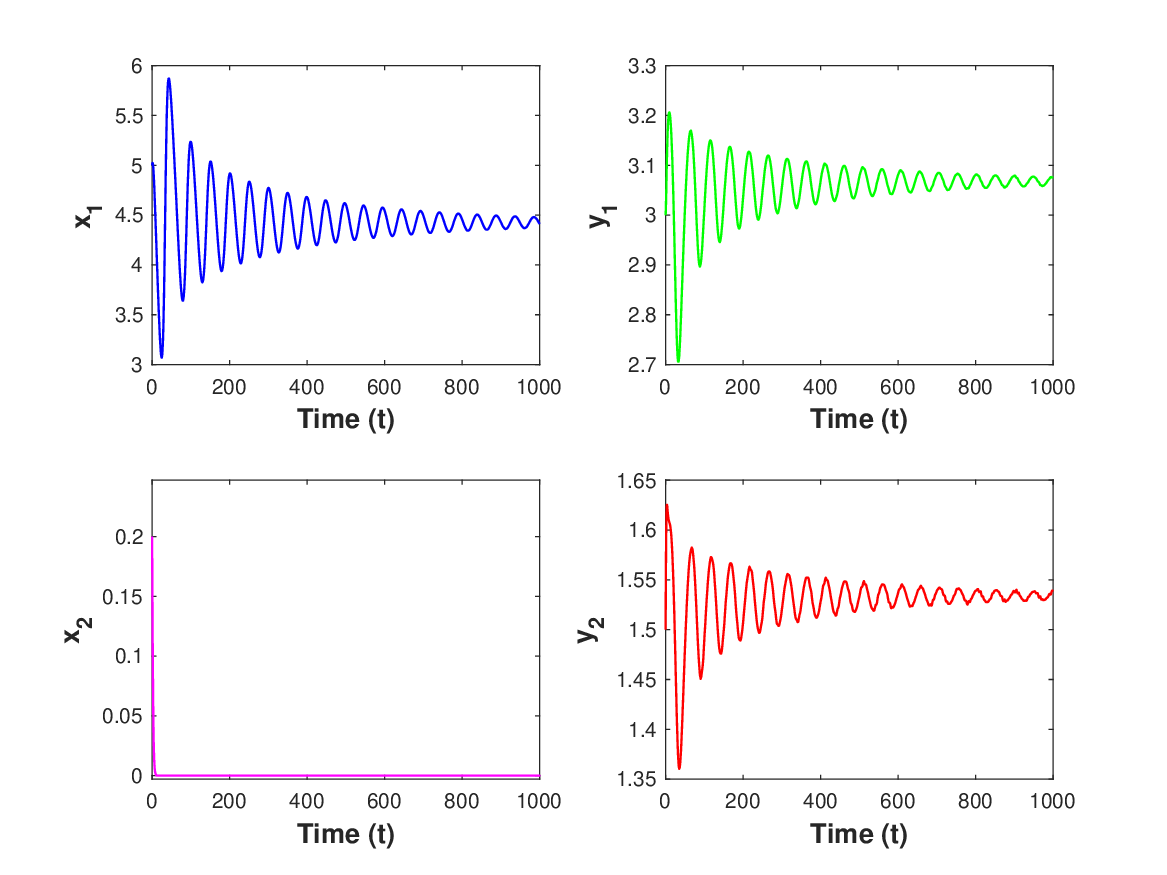}
\subcaption{  I.C. $= [ 5,3,0.2,1.5]$}
 \label{fig:x2_extinction}
\end{subfigure}
\begin{subfigure}{.48\textwidth}
\centering
\includegraphics[width = 9cm]{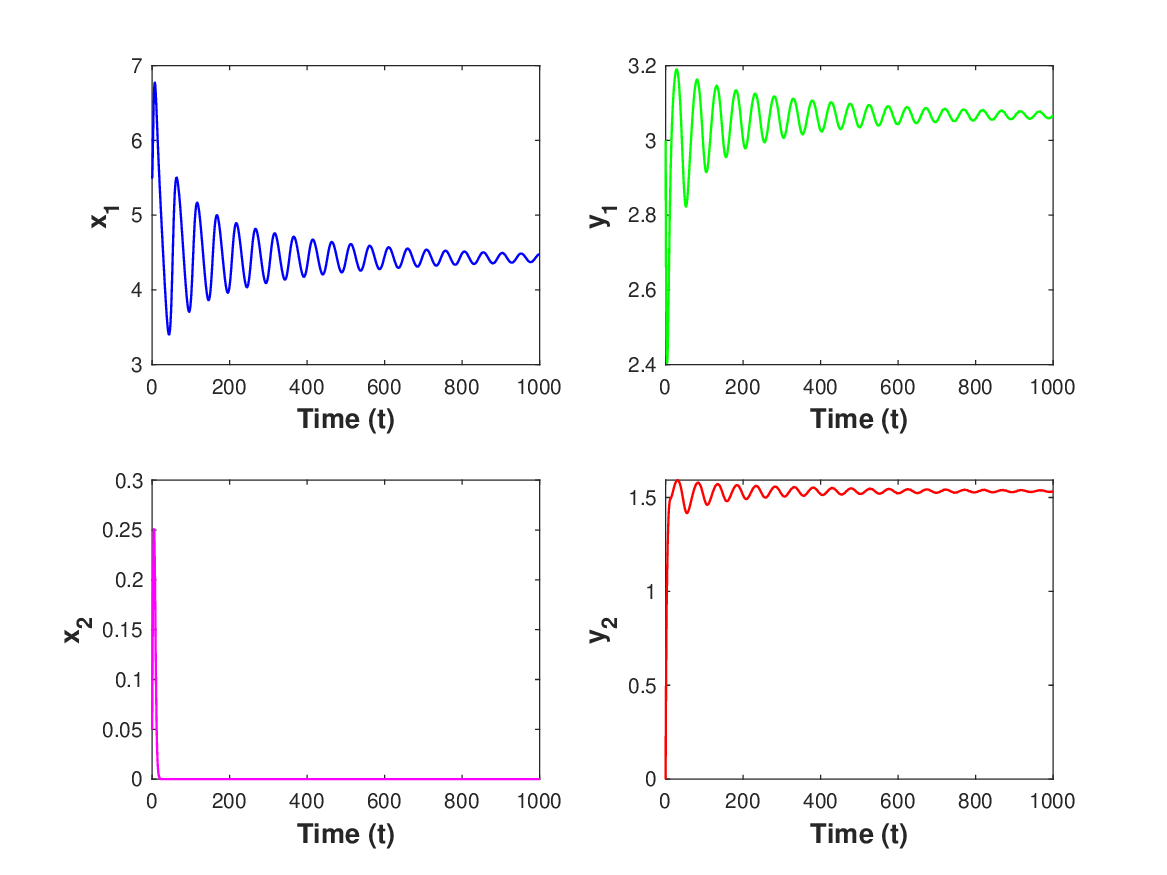}
\subcaption{I.C. $= [ 5,3,0.05,0]$}
 \label{fig:x2_extinct_y2_0}
\end{subfigure}
\caption{The time series figures \eqref{fig:x2_extinction} and, \eqref{fig:x2_extinct_y2_0} shows pest extinction in patch $\Omega_2$. 
 In Figure \eqref{fig:x2_extinct_y2_0}, the extinction of pests in patch $\Omega_2$ occurs without any initial predator population in  $\Omega_2$. Initially, the predators are introduced solely in $\Omega_1$, and due to their dispersion, predators enter the crop field and subsequently eliminate pests. The parameters used are $k=10, \alpha=0.2, \xi=0.39, \epsilon_1 = \epsilon_2 =0.4, \delta_1= \delta_2=0.2,k_2=0.3,k_4=0.2$. }
 \label{x2_extinction_comp}
\end{figure}
\begin{figure}
\begin{subfigure}{.48\textwidth}
   \centering
\includegraphics[width = 9cm]{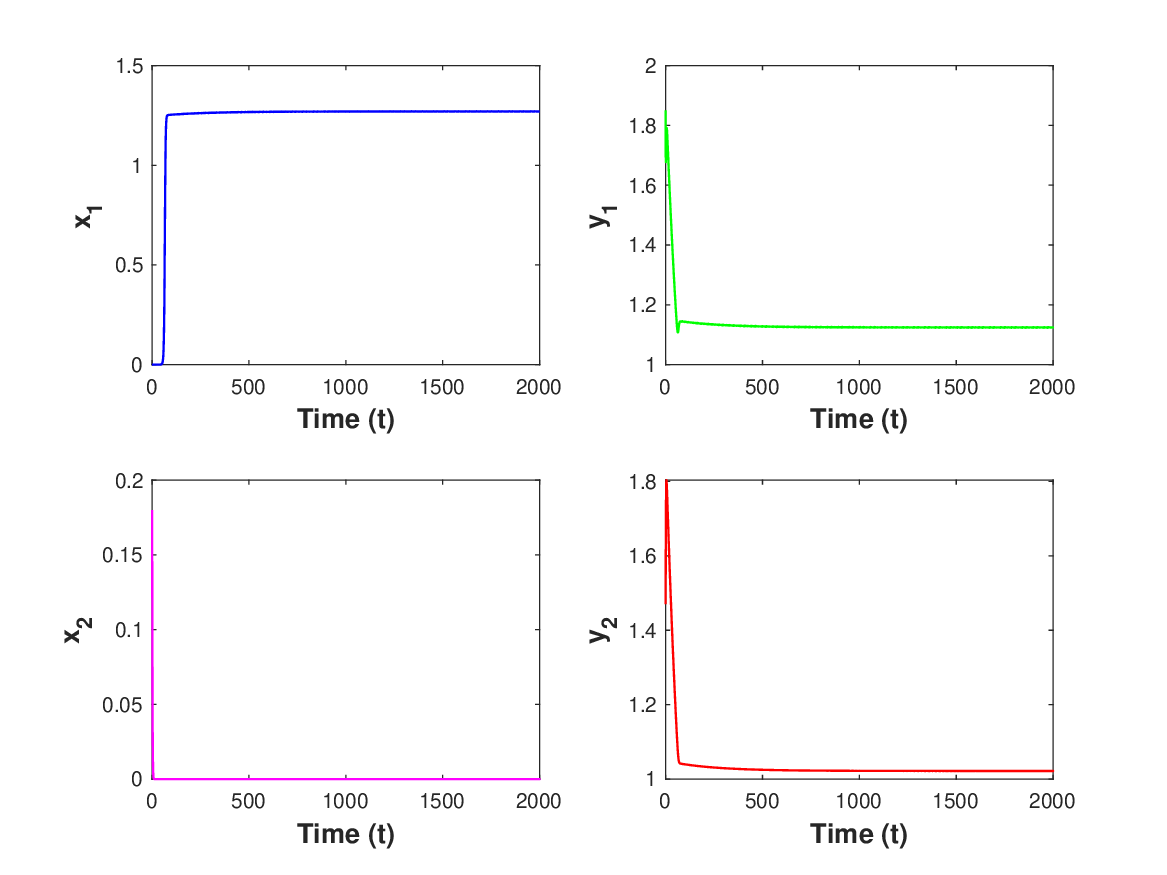}
\subcaption{ I.C. $= [ 0.00001,1.85,0.18,1.47]$}
 \label{low_ic_x1_x2_extinction}
\end{subfigure}
\begin{subfigure}{.48\textwidth}
\centering
\includegraphics[width = 9cm]{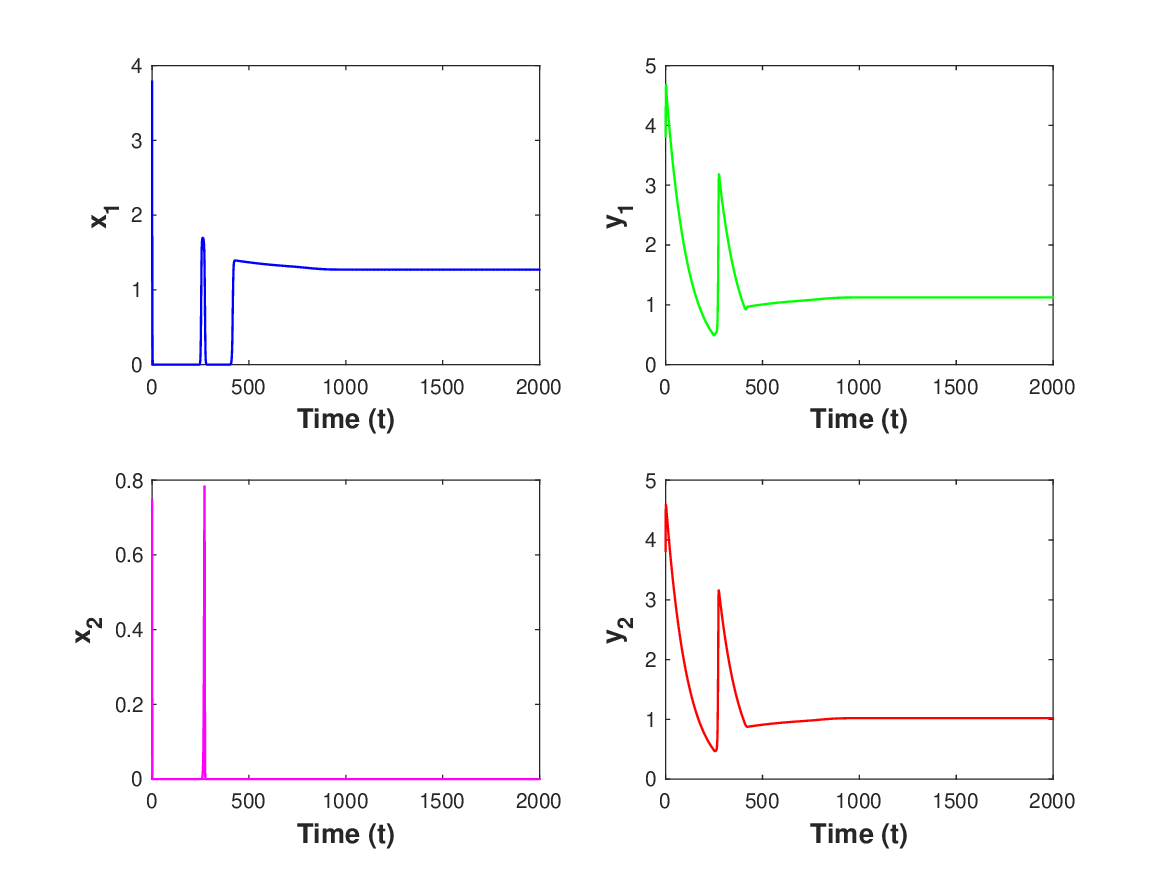}
\subcaption{I.C. $= [3.8,3.8,0.75,3.8]$}
\label{x2_extinction_xp_xc_updated}
\end{subfigure}
\caption{The time series figures \eqref{low_ic_x1_x2_extinction} and \eqref{x2_extinction_xp_xc_updated}, shows pest extinction in patch $\Omega_2$ under two different carrying capacities. Here, $k_p$ and $k_c$ denote the carrying capacity in the prairie strip and the crop field, respectively, with $k_p<k_c$. For Figure \eqref{low_ic_x1_x2_extinction}, pest extinction in patch $\Omega_2$ occurs for a very low initial population of pest in $\Omega_1$. The parameters used are $k_p=2,k_c=20,\alpha=0.9, \xi=0.9, \epsilon_1 = 0.4, \epsilon_2 =0.8, \delta_1= 0.2, \delta_2=0.01, k_2=0.9 ,k_4=0.1$.}
\label{x2_extinction_comp_sec}
\end{figure}
\subsubsection{Pest free state in patch \texorpdfstring {$\Omega_1 $}{Lg}}\mbox{}

\begin{lemma}
The equilibrium point $E_2 = (0,y_1^*,x_2^*,y_2^*)$ exists if, $  \tilde{Q} > 0$  and  $\epsilon_2 > \tilde{Q}$ where, $\tilde{Q} =    \delta_2 + k_4 + \dfrac{k_2 \ k_4}{  \dfrac{\epsilon_1 \xi}{1+\alpha \xi} \  - \delta_1 - k_2 }$.
\label{lem:E2_existence}
\end{lemma}
\begin{proof}
See Appendix \eqref{proof_E2_existence}.
\end{proof}
\begin{lemma}
      The equilibrium point $E_2 = (0,y_1^*,x_2^*,y_2^*)$ is conditionally locally asymptotically stable.
    \label{lem:E2_stability}
 \end{lemma}
\begin{proof}
    See Appendix \eqref{lem proof:E2_stability}.
\end{proof}
 
\begin{figure}
\begin{subfigure}{.48\textwidth}
   \centering
\includegraphics[width = 9cm]{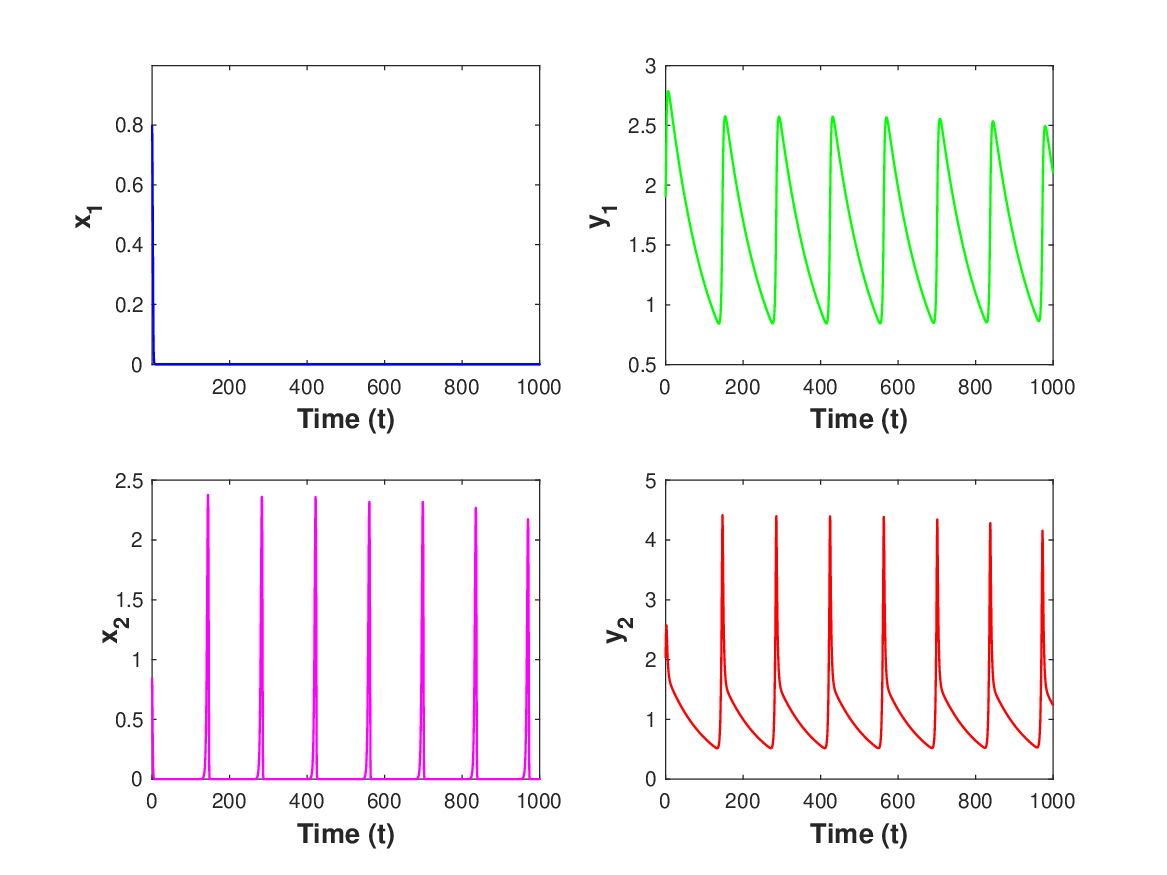}
\subcaption{  $k=10$}
 \label{x1_extinction} 
\end{subfigure}
\begin{subfigure}{.48\textwidth}
\centering
\includegraphics[width = 9cm]{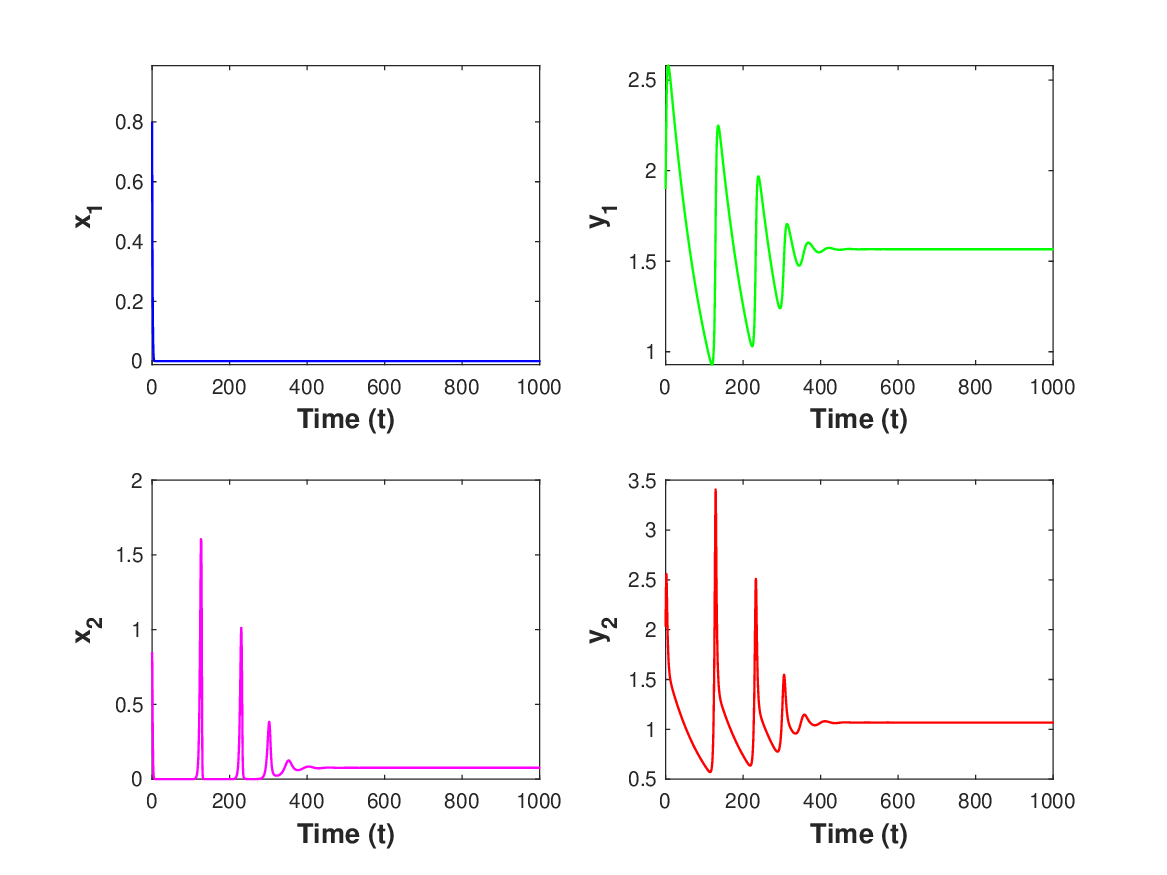}
\subcaption{$k_p=1, k_c=10$}
 \label{kp_kc_x1_extinction}
\end{subfigure}
\caption{The time series figures show pest extinction in patch $\Omega_1$. 
In reference to \eqref{x1_extinction}, we posit that the carrying capacities are equivalent in both patches. However, in \eqref{kp_kc_x1_extinction}, the carrying capacity of the prairie strip is lower than that of the crop field. Here, $k_p$ and $k_c$ denote the carrying capacity in the prairie strip and the crop field, respectively. The parameters used are $ \alpha=0.2, \xi=0.5, \epsilon_1 = 0.4, \epsilon_2 =0.8, \delta_1= \delta_2=0.15, k_2=0.1,k_4=0.2 \ \text{with I.C.} = [0.8,1.9,0.85,2.03]$.}
\label{x1_extinction_comp}
\end{figure}

\subsubsection{Pest free state in both \texorpdfstring{$\Omega_1 \ \& \ \Omega_2$}{Lg}}\mbox{}
\begin{lemma}
    The equilibrium point $ E_3 =(0,y_1^*,0,y_2^*)$ exists if $\xi = \dfrac{{Q}}{\epsilon_1 - \alpha {Q}}$ and $ \epsilon_1 > Q $ where, $    Q = \dfrac{\delta_1 \delta_2  + \delta_1 k_4  + k_2 \delta_2  }{\delta_2 + k_4 } $.
    \label{lem:E3_existence}
 \end{lemma}
 \begin{proof}
 See Appendix \eqref{proof_E3_existence}.
\end{proof}

\begin{lemma}
   The equilibrium point $E_3 = (0,y_1^*,0,y_2^*)$ is a non-hyperbolic  point.
   \label{lem:E3_stability}
\end{lemma}
\begin{proof}
    See Appendix \eqref{lem proof:E3_stability}.
\end{proof}
\begin{figure}
\centering
\includegraphics[width = 10cm, height=7cm]{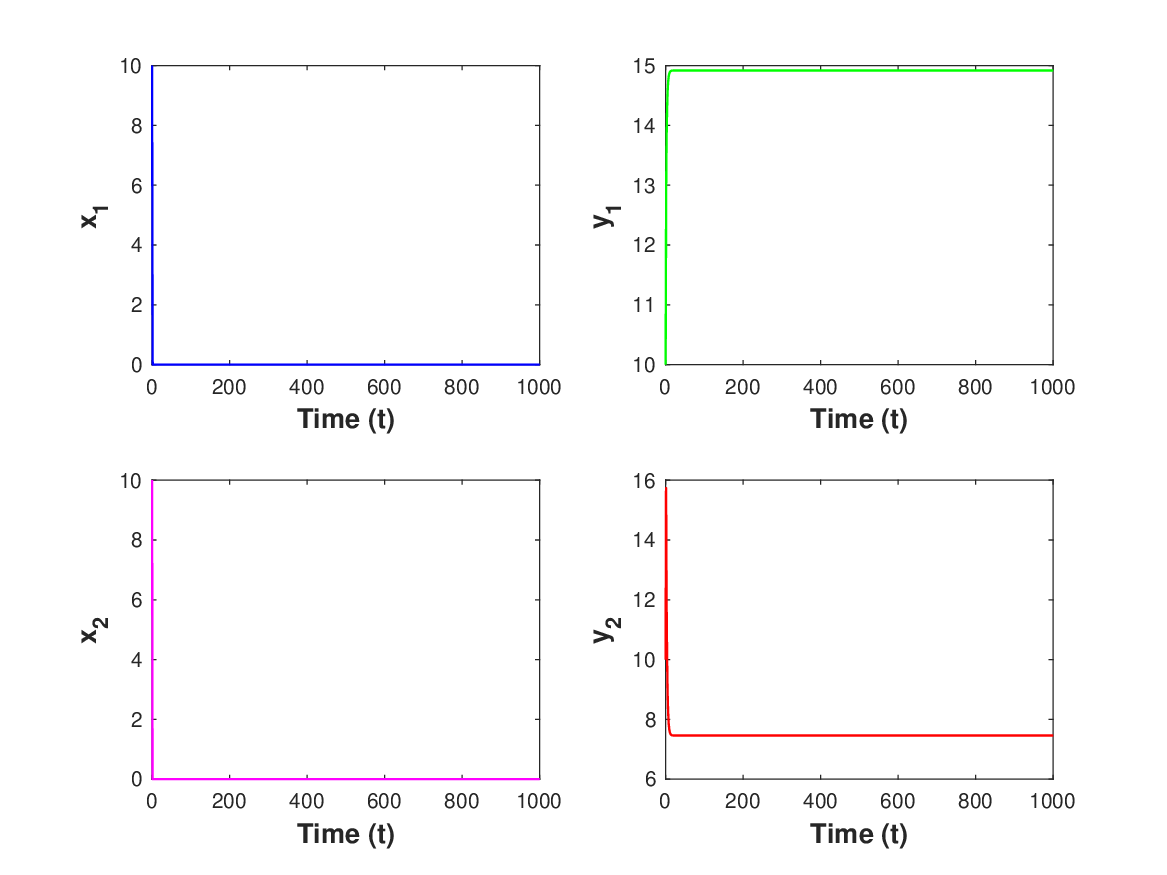}
\caption{The time series figure shows pest extinction in both patches
$\Omega_1$ \& $\Omega_2$ for an exact quantity of additional food. The parameters used are
 $k=10,\alpha=0.2,  \xi= 0.714286, \epsilon_1 = 0.4, \epsilon_2 =0.8, \delta_1= \delta_2=0.2, k_2=0.1, k_4=0.2 \ \text{with I.C.}= [10,10,10,10]$.
 }
 \label{fig:x1_x2_extinction}
\end{figure}

\subsubsection{Coexistence state in both \texorpdfstring{$\Omega_1 \ \& \ \Omega_2$}{Lg}}\mbox{}
\begin{lemma}
The equilibrium point $ E_4 =(x_1^\star,y_1^\star,x_2^\star,y_2^\star)$ exists if $ 0< x_1^\star < \dfrac{\delta_1 + k_2}{\epsilon_1 - (\delta_1 + k_2)} - \alpha \xi $, $ 0<  x_2^\star < \dfrac{\delta_2 + k_4}{\epsilon_2 - (\delta_2 + k_4)} $, where $x_2^\star = f(x_1^\star) $ and $x_1^\star$ can be solved from nullcline equations. 
 \label{lem:E4_existence}
\end{lemma}
\begin{proof}
 See Appendix \eqref{proof_E4_existence}.
 \end{proof}
\begin{lemma}
    The equilibrium point $E_4 = (x_1^\star,y_1^\star,x_2^\star, y_2^\star)$ is conditionally locally asymptotically stable. 
    \label{lem:E4_stability}
\end{lemma}

\begin{proof}
    See Appendix \eqref{lem proof:E4_stability}.
\end{proof}

\begin{figure}[H]
\centering
\includegraphics[width = 10cm, height=7cm]{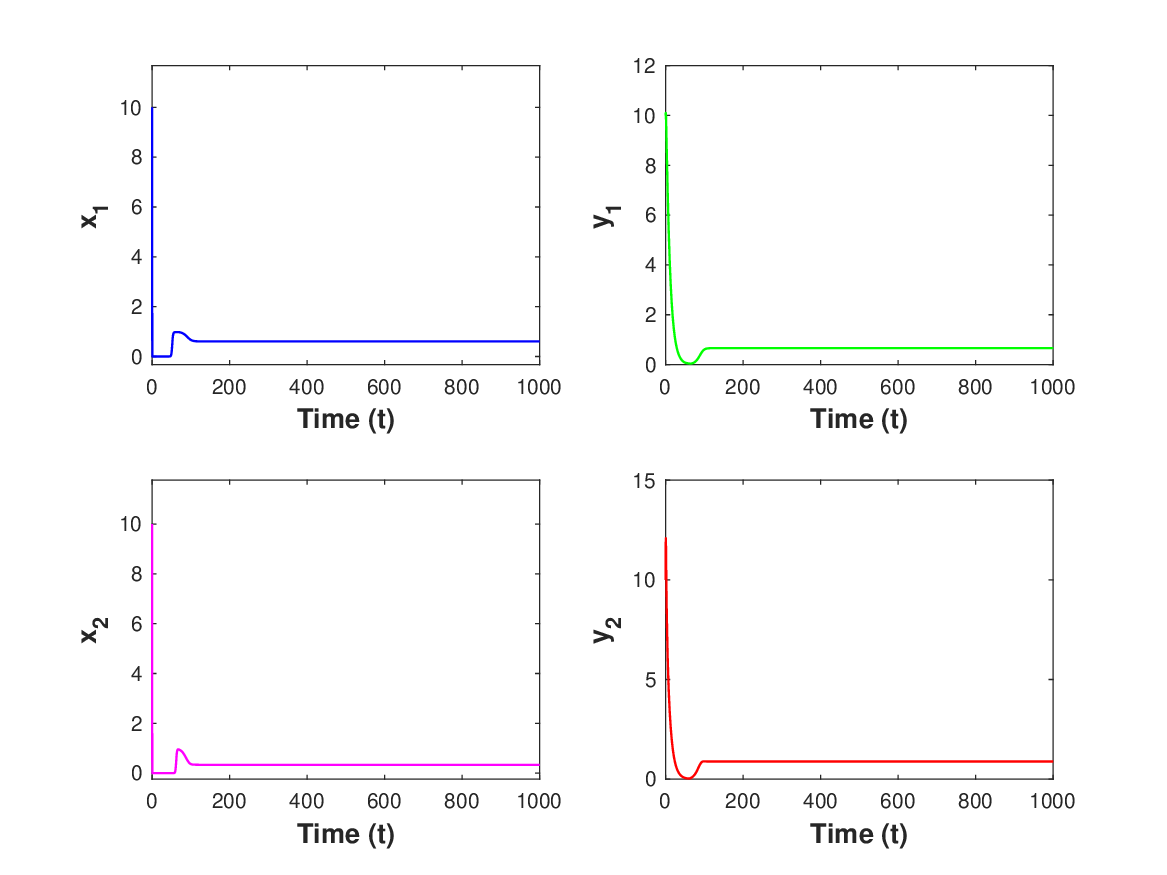}
\caption{The time series figure shows  the coexistence state in both patches $\Omega_1$  \& 
$\Omega_2$. The parameters used are $k=1, 
 \alpha=0.2,  \xi=0.37, \epsilon_1=0.2, \epsilon_2 =0.8, \delta_1= \delta_2=0.15, k_2=0.1, k_4=0.2 \ \text{with I.C.} = [10,10,10,10]$.
 }
 \label{fig:interior_equilibrium}
 \end{figure}

 \begin{rem}
Numerical simulations suggest that it is possible to eradicate the pests within the crop field $(\Omega_2)$. Figure \eqref{fig:x2_extinction_equal_dispersal} shows the pest extinction in the patch $\Omega_2$ with symmetric dispersal. Additionally, Figures \eqref{x2_extinction_comp} and \eqref{x2_extinction_comp_sec} illustrate pest extinction in $\Omega_2$ under different scenarios.  Figure \eqref{x1_extinction_comp} shows that pest extinction in the prairie strip $\Omega_1$ is possible. Moreover, Figure \eqref{fig:x1_x2_extinction} shows pest extinction in both $\Omega_1 \ \& \ \Omega_2$ simultaneously can happen under strict parametric equality on the quantity of additional food, validating Lemma \ref{lem:E3_existence}. Furthermore, Figure \eqref{fig:interior_equilibrium} demonstrates that coexistence is also possible under certain parametric restrictions as mentioned in Lemma \ref{lem:E4_stability}.
 \end{rem}

\

\

\subsection{Hopf Bifurcation}
A bifurcation occurs in a dynamical system when a parameter change leads to a sudden qualitative change in the system's behavior.
The bifurcation value is the critical parameter value at which the qualitative dynamics of a model change. 
We use the method developed in  \cite{Liu94} for Hopf Bifurcation. We extended the conditions for the 4-dimensional patch model \eqref{eq:patch_model2} according to the criteria outlined in \cite{Liu94}. The bifurcation parameter in our mathematical model  \eqref{eq:patch_model2} is $\xi$, the quantity of additional food. The following conditions need to hold for the bifurcation parameter $\xi = \xi^*$.

\begin{thm}
 The necessary and sufficient conditions for the system \eqref{eq:patch_model2} undergoes Hopf bifurcation with respect
to the parameter $\xi = \xi^*$ around the equilibrium point $E_1 = (x_1^*,y_1^*,0,y_2^*)$ are stated as follows:

\begin{enumerate} [(i)]
 \item $A_0(\xi^*) > 0 $, $A_3(\xi^*) > 0 $, $(A_3A_2-A_1)(\xi^*)>0$
    \\
\item $f(\xi^*) = (A_1A_2A_3- A_0A_3^2-A_1^2)(\xi^*) = 0 $
    \\
\item $ \dv{(f(\xi))}{\xi} \rvert_{\xi = \xi^*} \neq 0$
    \label{}
\end{enumerate}
where the eigenvalues of the system at equilibrium point $E_1$ satisfy $\lambda^4+A_3\lambda^3+A_2\lambda^2+A_1\lambda+A_0=0$
\label{thm:hopf bifurcation}
\end{thm}

\begin{proof}
\label{thm proof:hopf bifurcation}
The characteristic equation for $J_1 = J(E_1)$ is given by:
\begin{equation*}
\setlength{\jot}{10pt}
\begin{aligned}
&  \hspace{6cm} \lambda^4 +A_3 \lambda^3 + A_2 \lambda^2 + A_1 \lambda + A_0  = 0  \\
&\text{where, } A_3 = - \left( A - P + D \right), \ A_2 =  D(A - P) - (AP-BC),\  A_1 =  D(AP - BC) +BCE ,  \ A_0 = -BCDE\\
&\text{ and the expressions for } A,B,C,D,E \  \text{and}, P \ \text{can be found in Lemma} \  \ref{lem:E1_stability}.\\
& \text { for } A_0(\xi^*) >0  \implies BCDE(\xi^*)< 0 \text{   ( which is always true from the feasibility condition )}   \\
& \text { for } A_3(\xi^*) > 0  \implies (A-P+D)(\xi^*) < 0 \\
& \text { for } (A_3A_2-A_1)(\xi^*)>0 \implies (A-P+D)\left( D(A - P) - (AP-BC)  \right) +  \left(D(AP - BC) +BCE \right)< 0 \rvert_{\xi = \xi^*} \\
& \text { for } f(\xi^*) = (A_1A_2A_3- A_0A_3^2-A_1^2)(\xi^*) = 0 \text{ we have,}\\
&\implies- (A-P+D) \left( D(A - P) - (AP-BC)  \right) \left(D(AP - BC) +BCE \right)\\
& \hspace{1.0cm}+BCDE(A-P+D)^2 - \left(D(AP - BC) +BCE \right)^2 \rvert_{\xi = \xi^*}= 0\\
&\implies \{ D^3 + D^2(A-P)  - D(AP-BC) -BCE\} * \{(A-P)(AP-BC)-BCE \}\rvert_{\xi = \xi^*} = 0 \\
& \text { for } \dv{(f(\xi))}{\xi} \rvert_{\xi = \xi^*} \neq 0 \implies \dv{(A_1A_2A_3- A_0A_3^2-A_1^2)}{\xi}\rvert_{\xi = \xi^*} \neq 0\\
& \text{Now, using the above expression we solve, } 
\dv{ (\phi * \psi)}{\xi}\rvert_{\xi = \xi^*} =  \phi * \dv{\psi}{\xi} + \psi * \dv{\phi}{\xi}\rvert_{\xi = \xi^*} \neq 0\\
& \text{where, } \phi(\xi) = D^3 + D^2(A-P)  - D(AP-BC) -BCE , \hspace{0.3cm} \psi(\xi) =  (A-P)(AP-BC)-BCE 
\end{aligned}
\end{equation*}
Thus, conditions mentioned in \textit{(i),(ii)}, and \textit{(iii)} are satisfied, so \eqref{eq:patch_model2} undergoes Hopf Bifurcation and hence the theorem is proved. 
\end{proof}

\begin{rem}
     A commonly seen dynamic in multi-component ODE systems, that is, typically systems of three or more ODEs, is chaos. This is defined as aperiodic behavior exhibiting sensitive dependence to initial condition \cite{D18}. The possibility of chaotic dynamics in our proposed system \eqref{eq:patch_model2} can be seen via simulations shown in Appendix \eqref{chaos}.
\end{rem}

\subsection{Comparison of AF patch model to classical bio-control models}
 \

In this section, we compare the classical models with our additional food two-patch model to show the benefits of this model in pest control. So, we study the possibilities of coexistence state in the models. We have studied the $(x^*,{y}^*)$ state in the Rosenzweig-MacArthur model and Holling type II additional food model and have studied the $(0,{y_1^*},{x_2^*},{y_2^*})$  state to compare the total pest population between all these models.

We also numerically investigate the comparison of pest levels concerning the interior state for all the above-mentioned models, i.e., 
comparing the $(x_1^\star, y_1^\star, x_2^\star, y_2^\star)$ of \eqref{eq:patch_model2} with $(x_{r}^*,y_{r}^*)$ of \eqref{rmmodel} and $(x_{h}^*,y_{h}^*)$ of \eqref{org_model}.
\subsubsection{Rosenzweig-MacArthur Predator-Prey model vs. Patch driven Additional food model}
\mbox{}
\ 

We first study the classical Rosenzweig-MacArthur model \cite{Hal08} which, if non-dimensionalized, is given by:
\begin{equation} 
\begin{aligned}
\dot{x_r} &=x_r\bigg(1 - \dfrac{x_r}{k}\bigg) - \dfrac{x_r y_r}{1 + x_r} \\
\dot{y_r} &= \dfrac{\epsilon_r x_r y_r}{1 + x_r } - \delta_r y_r\\
\end{aligned}
\label{rmmodel}
\end{equation}

The dynamics of the Rosenzweig-MacArthur (RM) model show that we cannot have pest extinction. So, the least pest population we can have is when interior equilibrium $(x_r^*,y_r^*)$ where $x_r^*=\dfrac{\delta_r}{\epsilon_r-\delta_r}$ and the interior is stable when $k<\dfrac{\delta_r+\epsilon_r}{\epsilon_r-\delta_r}$.

\begin{lemma} \label{lem:rm_more_than_patch}
When $\dfrac{\delta_1}{\epsilon_1-\delta_1\alpha} < \xi <\dfrac{\delta_1+k_2}{\epsilon_1-\alpha(\delta_1+k_2)}$ then system \eqref{eq:patch_model2} can achieve lower pest population than system \eqref{rmmodel}. 
\end{lemma}
\begin{proof}
\label{lem proof:rm_more_than_patch}
We will compare the $(x_r^*,y_r^*)$ equilibrium of system \eqref{rmmodel} and the $E_2 = (0,y_1^*,x_2^*,y_2^*)$ equilibrium of system \eqref{eq:patch_model2} where total pest population in the system is $x_2$. 
\par We know the equilibrium points can be stated as, 
$x^*_r=\dfrac{\delta_2}{\epsilon_2-\delta_2}$ and $x_2^*=\dfrac{\tilde{Q}}{\epsilon_2-\tilde{Q}}$, where $\tilde Q$ is,  
$ \tilde{Q} = \delta_2 + k_4 + \dfrac{k_2 \ k_4}{  \dfrac{\epsilon_1 \xi}{1+\alpha \xi} \  - \delta_1 - k_2}.$ So, if $  \dfrac{\epsilon_1 \xi}{1+\alpha \xi} \  - \delta_1 - k_2   >0$ then $x_2^*>x_r^*$. Thus, for pest population to be lower in the system \eqref{eq:patch_model2}, we need $\tilde{Q}<\delta_2$. Solving, we get $\xi<\dfrac{\delta_1+k_2}{\epsilon_1-\alpha(\delta_1+k_2)}$.

\par Now, if we solve $x_2^*<x_r^*$ we get, $\delta_2> \tilde{Q}$  i.e., $k_2<\delta_1+k_2-\dfrac{\epsilon_1\xi}{1+\alpha\xi}$. 

  
\par Thus, we have a lower total pest population in the system \eqref{eq:patch_model2} than in the system \eqref{rmmodel} when,
\begin{equation*}
\dfrac{\delta_1}{\epsilon_1-\delta_1\alpha} < \xi <\dfrac{\delta_1+k_2}{\epsilon_1-\alpha(\delta_1+k_2)}
 \label{eq: rm_more_patch}
\end{equation*}
\end{proof}

\begin{figure}
\centering
\includegraphics[width = 10cm, height=7cm]{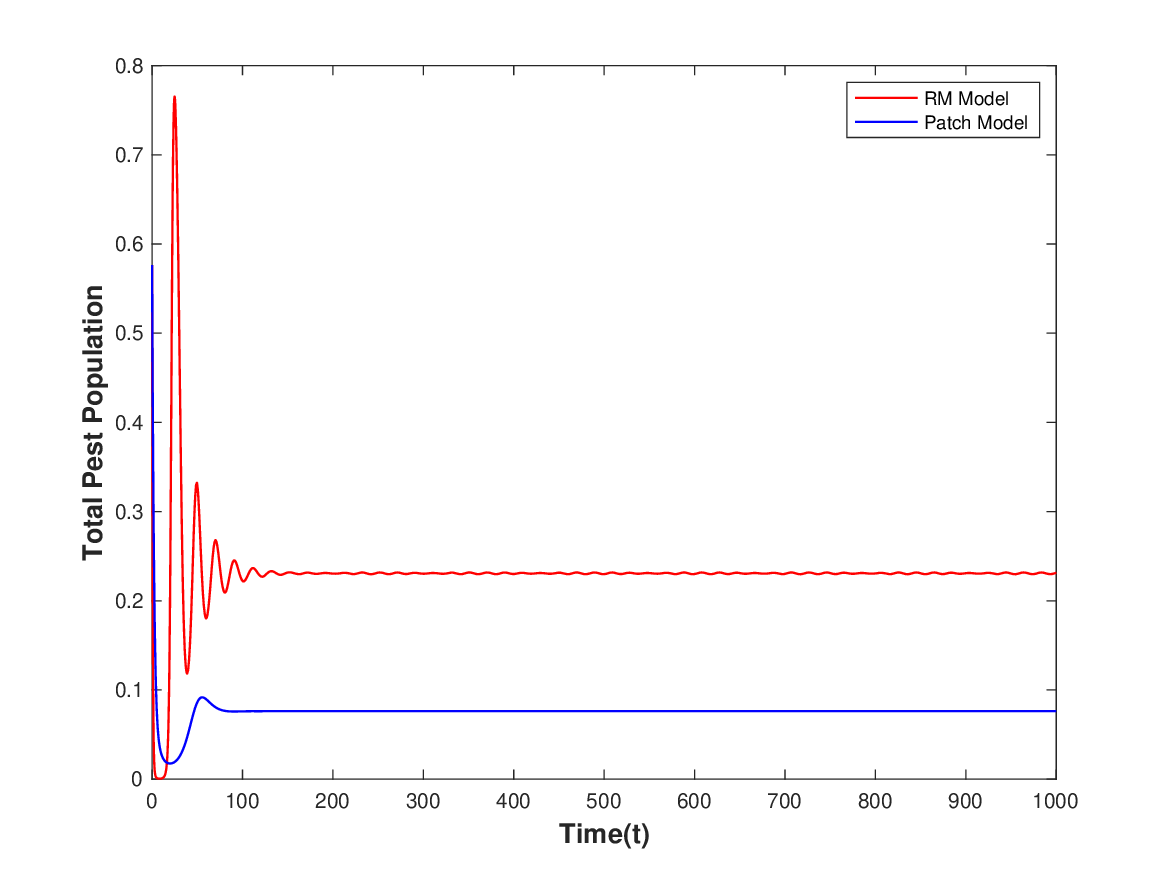}
\caption{The time series figure showing the total pest population comparison of the system \eqref{rmmodel} and \eqref{eq:patch_model2}. The initial condition of the population for the patch model is [0.5, 1.56644, 0.0762332, 1.06803], and that for the RM model is used as total pest and predator population respectively, i.e. [0.5762, 2.6345].}
\label{fig:rm_vs_patch_fig}
\end{figure}

\begin{lemma} \label{lem:rm_less_than_patch}
When   $0 < \xi <   \dfrac{\epsilon_1 \ (Q-\delta_1)}{(\epsilon_1-Q\alpha)(\epsilon_1-\delta_1)}$ then, the system \eqref{eq:patch_model2} can achieve higher pest population than the system \eqref{rmmodel} where,   $Q=\dfrac{\delta_1\delta_2+\delta_1 k_4+k_2\delta_2}{\delta_2+k_4}$. 
\end{lemma}
\begin{proof}
\label{lem proof :rm_less_than_patch}
We prove the result using comparison. We will compare the $(x_r^*,y_r^*)$ equilibrium of system \eqref{rmmodel} and the $E_1 = (x_1^*,y_1^*,0,y_2^*)$ equilibrium of system \eqref{eq:patch_model2} where total pest population in the system is $x_1$.

\par We know the equilibrium points can be stated as $x_r^*=\dfrac{\delta_1}{\epsilon_1-\delta_1}$ and $x_1^*=\dfrac{Q(1+\alpha \xi)-\epsilon_1 \xi}{\epsilon_1 - Q}$  where, 
\begin{equation*}
    Q=\dfrac{\delta_1\delta_2+\delta_1 k_4+k_2\delta_2}{\delta_2+k_4}= \delta_1 + \dfrac{k_2\delta_2}{\delta_2+k_4}  \implies Q > \delta_1 
\end{equation*}

\par So, if we need $x_r^*<x_1^*$ then,
\begin{equation*}
    \dfrac{\delta_1}{\epsilon_1-\delta_1}<\dfrac{Q(1+\alpha \xi)-\epsilon_1 \xi}{\epsilon_1-Q}
\end{equation*}

\par From the classical definition, we have $\epsilon_1>\delta_1$ and from the existence of equilibrium point $E_1$ (see Lemma \ref{lem:E1_existence}) we have $\epsilon_1>Q$.  Simplifying the inequality, 

\begin{equation*}
\setlength{\jot}{10pt}
\begin{aligned}
& \implies \delta_1( \epsilon_1-Q)< (Q(1+\alpha \xi)-\epsilon_1 \xi)( \epsilon_1-\delta_1)\\
    & \implies \xi(\alpha Q-\epsilon_1)(\epsilon_1-\delta_1)>(\delta_1-Q)\epsilon_1\\
    & \implies \xi(\epsilon_1-\alpha Q)(\epsilon_1-\delta_1)<(Q-\delta_1)\epsilon_1
    \end{aligned}
\end{equation*}  

$\text{Since}\   \epsilon_1>Q \implies \epsilon_1>\alpha Q \ \text{as} \ (0<\alpha<1) \ \text{and,} \ Q>\delta_1. $ In conclusion, we have a higher total pest population in system \eqref{eq:patch_model2} than \eqref{rmmodel} when,

            
            
        

\begin{equation*}
0< \xi < \dfrac{\epsilon_1 (Q-\delta_1)}{ (\epsilon_1-Q\alpha)(\epsilon_1-\delta_1)}
 \label{eq: rm_less_patch}
\end{equation*}
\end{proof}

\begin{figure}
\centering
\includegraphics[width = 10cm,height=7cm]{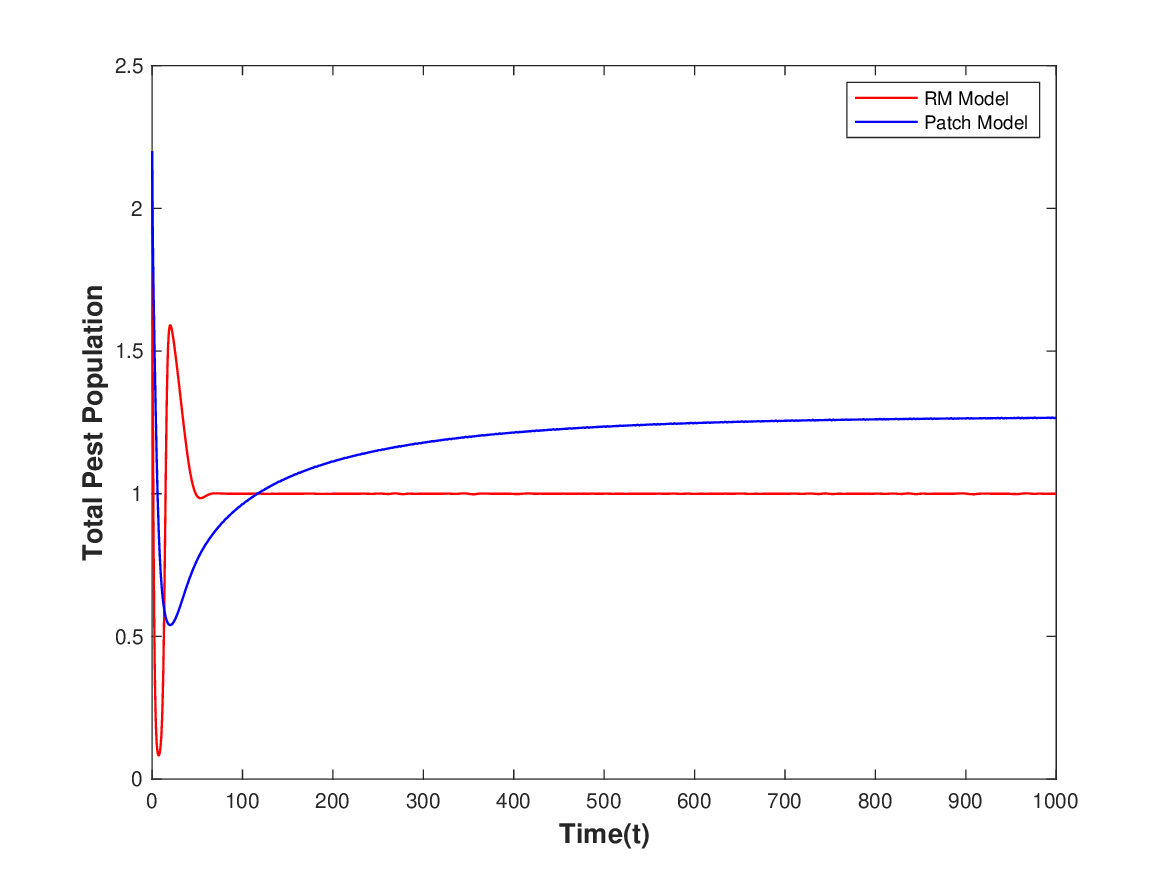}
\caption{The time series figure showing the total pest population comparison of the system \eqref{rmmodel} and \eqref{eq:patch_model2}. The initial condition of the population for the patch model is [2,1,0.2,1], and that for the RM model is used as total pest and predator population, respectively, i.e. [2.2, 2].}
\label{fig:rm_less_patch}
\end{figure}

\begin{figure}
\centering
\includegraphics[width = 10cm]{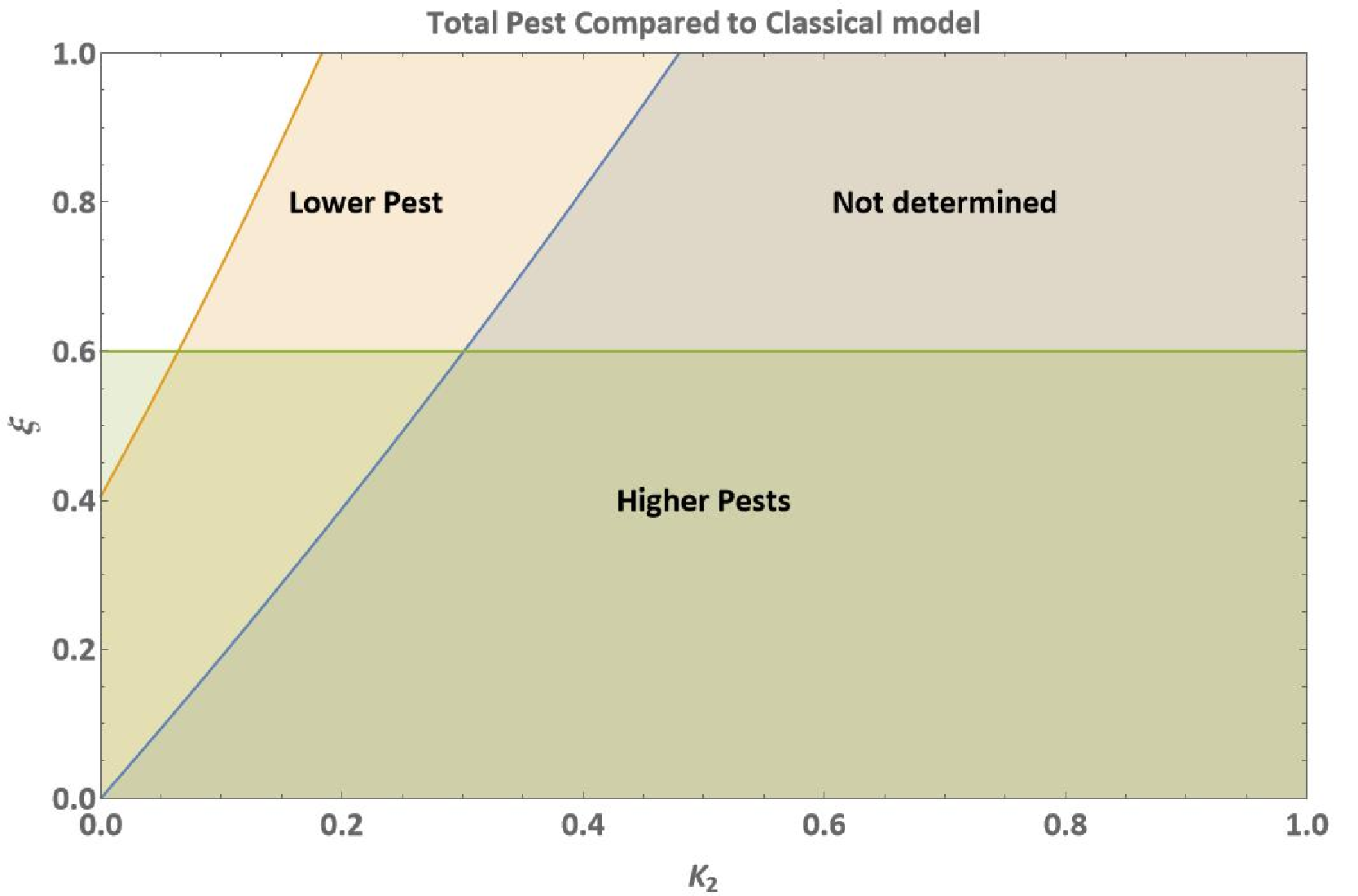}
\caption{Comparison of total pest population: RM model vs. Patch model. For fixed parameters, $(\xi,k_2)$ are varied to see regions where the total pest population in the RM model differs from the Patch model. ``Lower Pest" means the Patch model has a lower pest population than the RM model and vice versa. }
\label{fig:region_of_comparison}
\end{figure}

\par We choose parameter sets for numerical simulations to adhere to the stability criteria of the equilibrium points in both systems.
In Figure (\ref{fig:rm_vs_patch_fig}) the parameter set is $k=1, \alpha=0.2,  \xi=0.5, \epsilon_1=0.4, \epsilon_2=0.8, \delta_1= \delta_2=0.15, k_2=0.1, k_4=0.2$. Figure (\ref{fig:rm_vs_patch_fig}) validates Lemma  \ref{lem:rm_more_than_patch}; we see that the overall pest population in a two-patch setting with additional food is lower than the RM model.

\par The time series shown in Figure (\ref{fig:rm_less_patch}) describes the comparison between the RM model and the new patch model under parameter set $k=2, \alpha=0.9,\xi=0.9, \epsilon_1= 0.4, \epsilon_2=0.8, \delta_1= 0.2, \delta_2=0.01, k_2=0.9,  k_4=0.1.$ We validate Lemma \ref{lem:rm_less_than_patch} numerically, as it can be seen that the total pest population in the classical model is lower than the patch model with some choice of $(\xi,k_2)$. The model has the capability to have both a higher and lower combined pest population in both patches than the RM model.
The region corresponding to $\xi$ and $k_2$ can be broken down into different regions where for any fixed parameter set, there can be a choice of $(\xi,k_2)$ for which we can have both higher or lower total pest population in the system under conditions given in Lemma \ref{lem:rm_more_than_patch} and \ref{lem:rm_less_than_patch}, as it can be seen in Figure (\ref{fig:region_of_comparison}).

\subsubsection{Classical two-species additional food model vs. Patch driven Additional food model}
\

In this section, we will compare the classical pest-predator additional food model and the two-patch additional food model to compare the total pest population density. Secondly, we study classical Holling type II additional food model \cite{SP07}, which is given by
\begin{equation}\label{org_model}
    \begin{aligned}
\dot{x_h} &= x_h\bigg(1 - \dfrac{x_h}{k}\bigg) - \bigg(\dfrac{x_h y_h}{1 + x_h + \alpha  \xi}\bigg)\\ 
\dot{y_h} &= \epsilon_h \bigg(\dfrac{x_h + \xi}{1 +  x_h + \alpha \xi} \bigg)y_h - \delta_h y_h
    \end{aligned}
\end{equation}

Interior equilibrium for $\eqref{org_model}$ are:

\begin{equation*}
    x_h^* = \dfrac{\delta_h + \xi(\alpha\delta_h-\epsilon_h)}{(\epsilon_h-\delta_h)},\hspace{0.2cm}
    y_h^* = (1+x_h+\alpha\xi)\left(1-\dfrac{x_h}{k}\right)
\end{equation*} 

We are considering the case of boundary equilibrium $E_2 = (0,y_{1}^*,x_2^*, y_{2}^*)$ in \eqref{eq:patch_model2}, so total pest population in the system is $x_2$. 


\begin{lemma}
   When $\dfrac{ k_4 (\epsilon_1 -\alpha\delta_1) + (\delta_1 + k_2)(\alpha \delta_2 - \epsilon_2)}{(\epsilon_1  - \alpha(\delta_1 + k_2)) \ (\alpha \delta_2 - \epsilon_2)}< \xi < \dfrac{\delta_2}{\epsilon_2- \alpha\delta_2}$, then  system \eqref{eq:patch_model2} can achieve lower pest population than system \eqref{org_model}. 
\label{lem:af_more_than_patch}
\end{lemma}

\begin{proof}
\label{lem proof:af_more_than_patch}
    We will compare the $(x_h^*,y_h^*)$ equilibrium of system \eqref{org_model} and the $(0,y_1^*,x_2^*,y_2^*)$ equilibrium of system \eqref{eq:patch_model2}.\\


We know equilibrium points can be written as,  $ x_h^* = \dfrac{\delta_2 + \xi(\alpha\delta_2-\epsilon_2)}{(\epsilon_2-\delta_2 )},$ and $x_2^*=\dfrac{\tilde{Q}}{\epsilon_2-\tilde{Q}}$, where  $ \tilde{Q} =    \delta_2 + k_4 + \dfrac{k_2 \ k_4}{  \dfrac{\epsilon_1 \xi}{1+\alpha \xi} \  - \delta_1 - k_2    }$.

Thus, in order for the pest population to be lower in the system \eqref{eq:patch_model2}, we need 
\begin{equation*}
 \dfrac{\delta_2 + \xi(\alpha\delta_2-\epsilon_2)}{(\epsilon_2-\delta_2 )} > \dfrac{\tilde{Q}}{\epsilon_2-\tilde{Q}}
\end{equation*}
We now need,
 \begin{equation}
 \tilde{Q} <  \delta_2 + \xi(\alpha\delta_2-\epsilon_2) \ \text{ and as,} \  (\alpha\delta_2-\epsilon_2) <0 \implies \tilde{Q} <     \delta_2 
 \label{eq:safcomp1}
 \end{equation} 

Putting the value $\tilde{Q}$ in the above equation \eqref{eq:safcomp1} we have, 
\begin{equation*}
\setlength{\jot}{10pt}
\begin{aligned}
&  \delta_2 + k_4 + \dfrac{k_2 \ k_4}{  \dfrac{\epsilon_1 \xi}{1+\alpha \xi} \  - \delta_1 - k_2} < \delta_2 + \xi(\alpha\delta_2-\epsilon_2)\\
& k_4 + \dfrac{k_2 \ k_4\ (1+\alpha \xi)}{\epsilon_1 \xi - (1+\alpha \xi)(\delta_1 + k_2)} < \xi(\alpha \delta_2-\epsilon_2)\\
 \end{aligned}
\end{equation*}

since $\epsilon_1 \xi - (1+\alpha \xi)(\delta_1 + k_2) < 0, \ \text{from Lemma} \ \ref{lem:E2_existence}\  \text{we have,}$ 

\begin{equation*}
\setlength{\jot}{10pt}
\begin{aligned}
& \implies k_4\epsilon_1 \xi-k_4 \delta_1 \ (1+\alpha \xi) > \xi (\alpha \delta_2-\epsilon_2) \ \{\epsilon_1 \xi - (1+\alpha \xi)(\delta_1 + k_2)\} \\
& \implies \xi  \{\epsilon_1 \xi - (1+\alpha \xi)(\delta_1 + k_2)\}(\alpha \delta_2-\epsilon_2) < k_4 \xi(\epsilon_1 -\alpha\delta_1)-k_4 \delta_1 <  k_4 \xi(\epsilon_1 -\alpha\delta_1) \ \ \ \because k_4 \delta_1 > 0\\
&  \implies \xi \{\epsilon_1  - \alpha(\delta_1 + k_2)\}(\alpha \delta_2 - \epsilon_2)- (\delta_1 + k_2)(\alpha \delta_2 - \epsilon_2) <  k_4 (\epsilon_1 -\alpha\delta_1) \ \ \ \because \xi>0 \\
& \implies \xi \{\epsilon_1  - \alpha(\delta_1 + k_2)\}(\alpha \delta_2 - \epsilon_2)<  k_4 (\epsilon_1 -\alpha\delta_1) + (\delta_1 + k_2)(\alpha \delta_2 - \epsilon_2)
  \end{aligned}
\end{equation*}

 If   $\epsilon_1  - \alpha(\delta_1 + k_2) >0$ then, 
 
\begin{equation*}
   \begin{aligned}
& \xi > \dfrac{ k_4 (\epsilon_1 -\alpha\delta_1) + (\delta_1 + k_2)(\alpha \delta_2 - \epsilon_2)}{(\epsilon_1  - \alpha(\delta_1 + k_2)) \ (\alpha \delta_2 - \epsilon_2)} \ \ \ \because (\alpha \delta_2 - \epsilon_2) < 0 \\
\end{aligned}
\end{equation*}
\vspace{0.2cm}

Also we know that, $ x_h^* = \dfrac{\delta_2 + \xi(\alpha\delta_2-\epsilon_2)}{(\epsilon_2-\delta_2 )}$ exists if, $\xi < \dfrac{\delta_2}{\epsilon_2- \alpha\delta_2}.$ So, we can have a low pest population in \eqref{eq:patch_model2} than in \eqref{org_model} when, 

\begin{equation*}
\dfrac{ k_4 (\epsilon_1 -\alpha\delta_1) + (\delta_1 + k_2)(\alpha \delta_2 - \epsilon_2)}{(\epsilon_1  - \alpha(\delta_1 + k_2)) \ (\alpha \delta_2 - \epsilon_2)}< \xi < \dfrac{\delta_2}{\epsilon_2- \alpha\delta_2} 
\label{eq: sf_more_patch}
 \end{equation*}
\end{proof}

\begin{figure}
\centering
\includegraphics[width = 10cm,height=7cm]{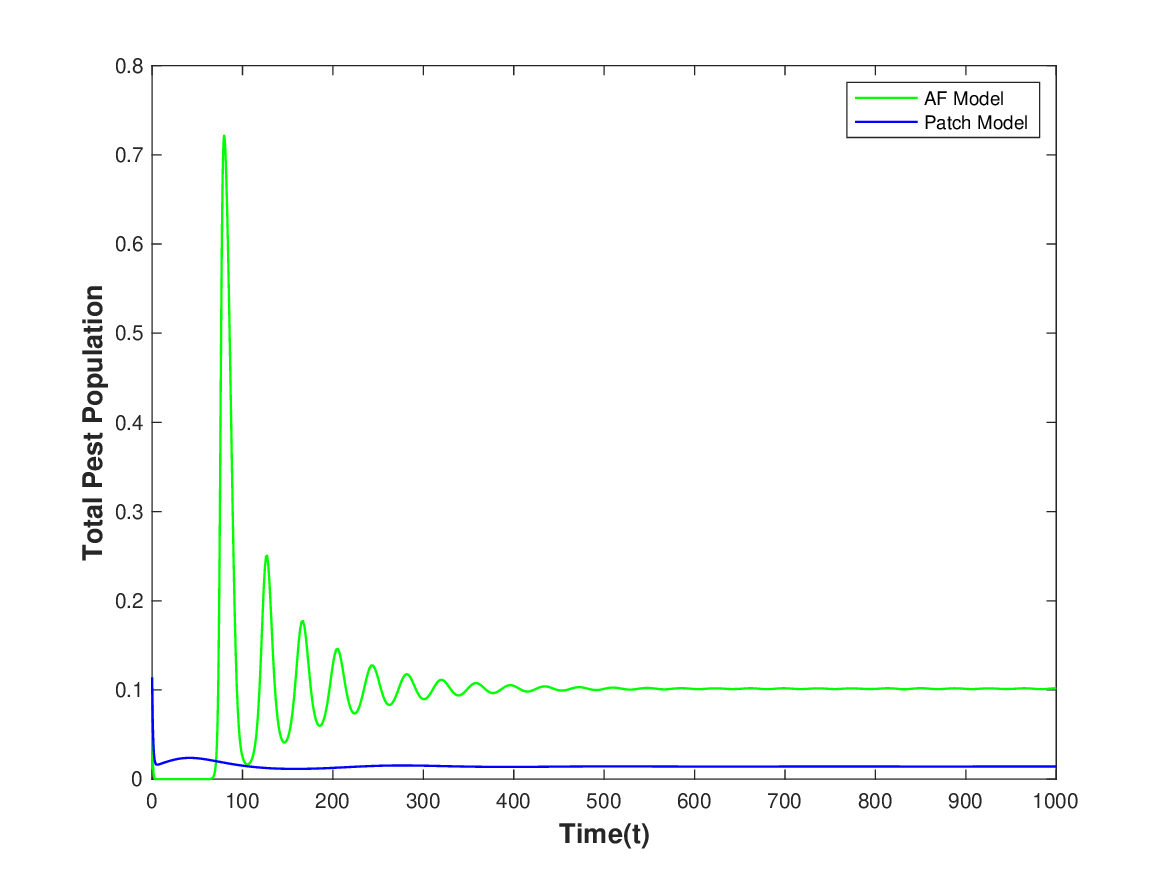}
\caption{Comparison of total pest population: Classical two species additional food model $(x_h^*,y_h^*)$  vs. Patch model $(0,y_1^*,x_2^*,y_2^*)$.  The initial condition of the population for the patch model is $[0.1,2.1,0.014,0.99]$, and that for the classical additional food model is $[0.114, 3.09]$. }
\label{fig:boundary_saf_vs_patch_fig}
\end{figure}

\begin{lemma}
\label{lem:af_less_than_patch}
When $\dfrac{\epsilon_1 Q}{\delta_1 (\alpha-1)(\epsilon_1 - \alpha Q)}  < \xi < \dfrac{\delta_1}{\epsilon_1- \alpha\delta_1}$, then system \eqref{eq:patch_model2} can achieve higher pest population than system \eqref{org_model}. 
\end{lemma}

\begin{proof}
    \label{lem proof:af_less_than_patch}
We will compare the $(x_h^*,y_h^*)$ equilibrium of system \eqref{org_model} and the $(x_1^*,y_1^*,0,y_2^*)$ equilibrium of system \eqref{eq:patch_model2}.\\  \par We know the equilibrium points can be stated as $x_h^* = \dfrac{\delta_1 + \xi(\alpha\delta_1-\epsilon_1)}{(\epsilon_1-\delta_1 )}$   and $x_1^*=\dfrac{Q(1+\alpha \xi)-\epsilon_1 \xi}{\epsilon_1 - Q}$  where, 
\begin{equation*}
    Q=\dfrac{\delta_1\delta_2+\delta_1 k_4+k_2\delta_2}{\delta_2+k_4}= \delta_1 + \dfrac{k_2\delta_2}{\delta_2+k_4}  
\end{equation*}
Thus, in order for the pest population to be lower in the system \eqref{org_model}, we need 
\begin{center}
    $\dfrac{\delta_1 + \xi(\alpha\delta_1-\epsilon_1)}{(\epsilon_1-\delta_1 )} < \dfrac{Q(1+\alpha \xi)-\epsilon_1 \xi}{\epsilon_1 - Q} $
\end{center}

\par From the classical definition, we have $\epsilon_1>\delta_1$ and from the existence of $E_1$ (see Lemma \ref{lem:E1_existence}) we have $\epsilon_1>Q$.  Simplifying the inequality, 

\begin{equation*}
\setlength{\jot}{10pt}
\begin{aligned}
& \implies \delta_1 \epsilon_1 + \xi(\alpha\delta_1-\epsilon_1)(\epsilon_1 - Q)< \epsilon_1 Q + \xi (\epsilon_1-\delta_1 )(\alpha Q -\epsilon_1) \\
&\because \  \delta_1\epsilon_1>0 \implies \  \xi(\alpha\delta_1-\epsilon_1)(\epsilon_1 - Q)< \epsilon_1 Q + \xi (\epsilon_1-\delta_1 )(\alpha Q -\epsilon_1)\\
& \implies \xi \ \{(\epsilon_1-\delta_1 )( \epsilon_1-\alpha Q )-(\epsilon_1 - \alpha\delta_1)(\epsilon_1 - Q) \}< \epsilon_1 Q \\
& \text{we have,} \ \epsilon_1-Q\alpha < \epsilon_1-\alpha Q \implies -(\epsilon_1-Q)(\epsilon_1 - \alpha\delta_1) > - (\epsilon_1-\alpha Q)(\epsilon_1 - \alpha\delta_1) \\
& \implies  \xi \ (\epsilon_1 - \alpha Q) \ \{(\epsilon_1-\delta_1 )-(\epsilon_1 - \alpha\delta_1) \}< \epsilon_1 Q  \implies \xi \delta_1(\epsilon_1 - \alpha Q) (\alpha  - 1 ) < \epsilon_1 Q \ \ \because 0<\alpha<1 \ \text{then, } \\
& \implies \xi> \dfrac{\epsilon_1 Q}{\delta_1 (\alpha-1)(\epsilon_1 - \alpha Q) }
\end{aligned}
\end{equation*}

\

Also we know, $ x_h^* = \dfrac{\delta_1 + \xi(\alpha\delta_1-\epsilon_1)}{(\epsilon_1-\delta_1 )}$ exists if, $\xi < \dfrac{\delta_1}{\epsilon_1- \alpha\delta_1}.$ Thus, in conclusion, we have a higher total pest population in system \eqref{eq:patch_model2} than \eqref{org_model} when, 

\begin{equation*}
  \dfrac{\epsilon_1 Q}{\delta_1 (\alpha-1)(\epsilon_1 - \alpha Q)}  < \xi < \dfrac{\delta_1}{\epsilon_1- \alpha\delta_1}
  \label{eq: sf_less_patch}
\end{equation*}
\end{proof}
\begin{figure}
\includegraphics[width = 10cm,height=7cm]{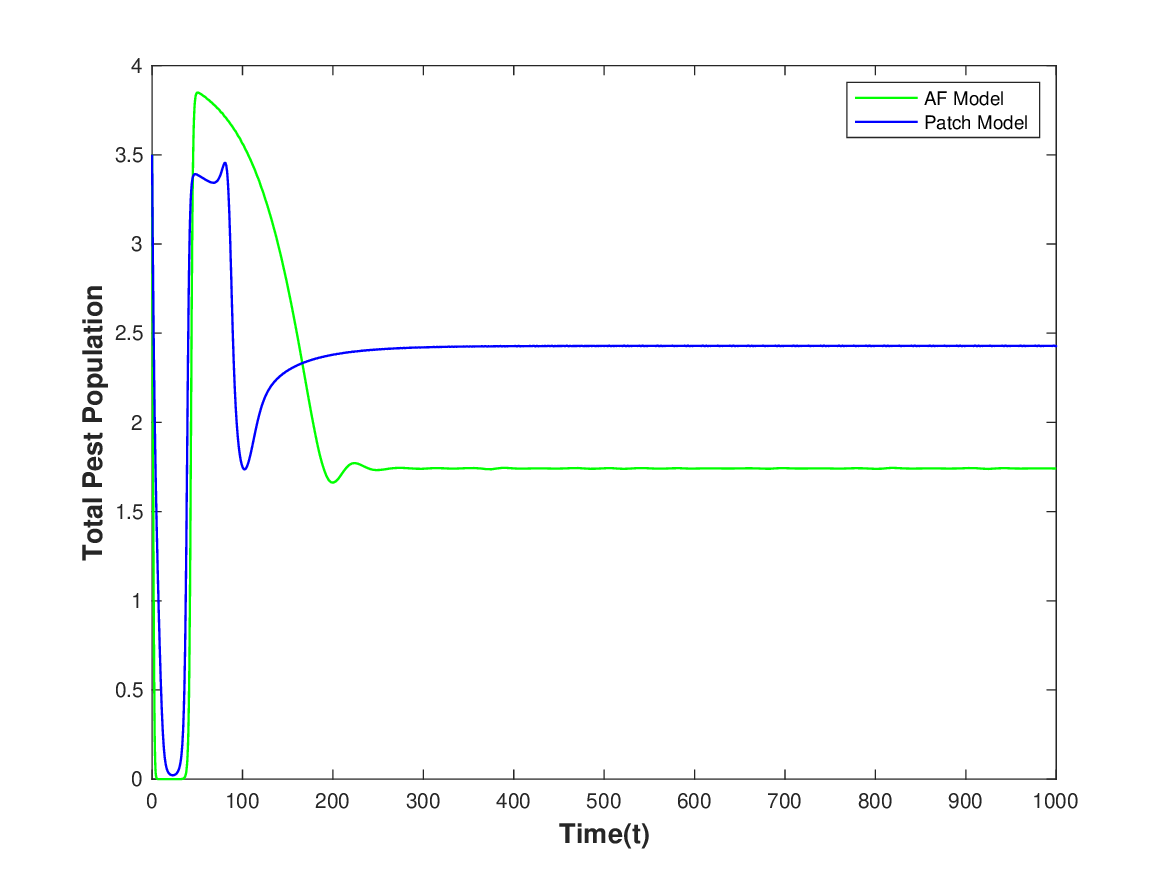}
 \caption{Comparison of total pest population: Classical two-species additional food model $(x_h^*,y_h^*)$  vs. Patch model $(x_1^*,y_1^*,0,y_2^*)$. The initial condition for the patch model is $[3,2,0.5,1.5]$, and that for the classical additional food model is $[3.5, 3.5]$. }
\label{fig:x2_0_saf_vs_patch}
\end{figure}

We choose parameter sets for numerical simulations to adhere to the stability criteria of the equilibrium points in both systems. In Figure \eqref{fig:boundary_saf_vs_patch_fig} the parameter set is $k=1, \alpha=0.2, \xi=0.39, \epsilon_1 = 0.4,\epsilon_2 =0.6, \delta_1=0.09,\delta_2=0.25,k_2=0.1,k_4=0.2$. Figure \eqref{fig:boundary_saf_vs_patch_fig} validates Lemma  \ref{lem:af_more_than_patch}; we see that the overall pest population,  in a two-patch setting with additional food which is $(x_2^*)$, in this case, is lower than the classical AF model.
The time series shown in Figure (\ref{fig:x2_0_saf_vs_patch}) describes the comparison between the classical AF model and the new AF patch model under parameter set $k=4, \alpha=0.2, \xi=0.37, \epsilon_1 =0.2, \epsilon_2 =0.4, \delta_1 = 0.15, \delta_2=0.01,k_2=0.5,k_4=0.5$. Figure \eqref{fig:x2_0_saf_vs_patch} validates Lemma \ref{lem:af_less_than_patch}.
We also conducted a numerical investigation to compare the interior equilibrium between systems \eqref{eq:patch_model2}, \eqref{rmmodel}, and \eqref{org_model}. See figures \eqref{fig:int_rm_vs_patch_fig} and \eqref{fig:int_saf_vs_patch}.


\section{Discussion and conclusion}

Motivated by recent innovative developments in landscape management, particularly in the North-Central US, such as the STRIPS program and the recent research movement towards LF-NPC, we have presented a bio-control model for a two-patch system with AF. Herein, a predator is introduced into a ``patch" such as a prairie strip or STRIP, which has AF. The AF boosts the predator efficacy, and it disperses or is drifted into a second ``patch," which is an adjacent or neighboring crop field, to target a pest. The predator will also return to the STRIP intermittently, due to the AF and refuge effects that it provides \cite{snyder2019give}.
To the best of our knowledge, such a study of the dynamics of a two-patch system with AF in one patch has not been analyzed earlier. The effect of drift and dispersal of the predator population between the patches has been studied extensively, and the biological implications of the results have also been analyzed. We have shown that linear drift and linear (local) dispersal of the predator population between the patches can have interesting dynamics on the stability of the equilibrium states.

It is seen that drift between patches prevents blow-up via Theorem \ref{thm:t1u1}; see Figure \eqref{blowup_drift_comp}. Similarly, it can be seen that \emph{linear} predator dispersal between patches can help in blow-up prevention; see Figure (\ref{fig:blowup_no_patch}) and Figure (\ref{fig:blowup_prevention_patch_1}),
which is prevalent in many AF models with pest dependent functional response \cite{PWB23}. One needs to differentiate between the cases of symmetric vs asymmetric dispersal, to understand the blow-up prevention. In the case of asymmetric dispersal, dispersal out of the AF patch needs to happen at a greater rate than the rate at which predators enter the crop field ($k_{2} > k_{4}$) - this is intuitive with the fact that predator blow-up can only occur in the AF patch, and so positive net dispersal out of the AF patch, could possibly act as a damping term to prevent blow-up. In the symmetric dispersal case, ($k_{2} = k_{4}$), and net movement is 0, so in is dealing with a conservative case. In this setting preventing blow-up depends on the predator death rate $\delta_{2}$, in the crop field. We require $\delta_1 + \delta_2 > \frac{\xi}{1+\alpha \xi} > \delta_{1}$, see Theorem \ref{thm:t1sd}. Here, integrated pest management programs need to keep this in mind, with this estimate on the death rate providing qualitative features on the life history of use of possible introduced predators, and the relation of the important life history death rate parameter, and how it scales with the AF the predator needs, to provide the most efficient biocontrol service.
All in all, unlike earlier studies where a higher-order term like intra-specific predator competition, predator dependent function response, or group defense in the pest is required to prevent blow-up, an AF driven patch system can also lead to such prevention, solely with a linear drift/dispersal term.   

The stability of the equilibrium states has also been studied extensively, as they are of great importance for bio-control. We show that pest extinction under drift is possible in the crop field via Lemma \ref{lem:E_stability_uni}. Furthermore this state is shown to be globally stable. Proving this is challenging, and a standard Lyapunov function argument \cite{hsu2005survey} does not work due to the way the AF term is modeled - that is, the structure of the predator's numerical response, due to the AF. We, however, are able to achieve global stability of the pest extinction state in the crop field. This follows via a series of lemmas and theorems; in particular, see Lemma \ref{lem:p1s}, \ref{lem:p2s}, Theorem \ref{thm:gs1} and Lemma \ref{lem:pest_extinction_crop_field}. A special case of this result enables us to prove the global stability of the coexistence equilibrium in the classical AF model via Corollary \ref{cor:g1} - this was left open in \cite{SP07}. Note that our Lyapunov function, developed here, will possibly work to prove the global stability of the classical AF model in the case of several other functional responses that have been recently considered. This is the subject of current and future investigations. 

Considering drift, pest eradication from the prairie strip can only be achieved under strict equality for the quantity of additional food via Lemma \ref{lem:E3_existence_unkidirectional}. Coexistence within the prairie strip and crop field is feasible in the patch model with drift, as demonstrated by Lemma \ref{lem:loc_uni_coexistence}. In order to design efficient bio-control strategies, we investigate the global stability of the pest extinction state in the crop field in relation to how large a prairie strip has to be to facilitate extinction. Thus, we study the relation of $k_{p}$ to other parameters. What is clearly seen in Figure \eqref{region_plots} is that there are 3 distinct regions in the parameter space that separate whether the pest extinction state in the crop field is global-attracting or not. If one studies Figure \eqref{fig:kp vs k2_none}, then global attraction needs $g(q_{1}) < k_{p} < f(q_{1})$, where $f(q_{1}), g(q_{1})$ are monotone functions of the drift out of the  AF patch. So we must ensure a certain amount of predators drift out of the AF patch and into the crop field to cause pest extinction there - and these should be in appropriate ranges. If enough of them are not drifted into the crop field, they may not be successful in eradicating the pest therein. Similarly, if one looks at Figure \eqref{fig:kp vs xi_none}, region III in particular, one sees a narrow band again where $g_{1}(q_{1}) < k_{p} < f_{1}(q_{1})$. Here $f_{1}, g_{1}$ are monotonically decreasing functions of the AF.
This tells us that as AF increases, the size of the prairie strip or AF patch will decrease, so as to yield global stability of the pest extinction state in the crop field. 
Also, this tells us that a minimum level of AF is required for pest eradication in the crop field, but ``too much" AF can act to the contrary - thus, managers maintaining or designing an integrated pest management program need to be mindful of this.
Working out the optimal ratio of the AF patch, and the AF crop within that, and its effect on pest eradication, remains a subject
of current investigation.

Similar results are observed in the dispersal case. It is observed that pest extinction is possible in either patch under certain parametric restrictions. In particular, pest extinction is possible in the crop field, where no predators are present initially, see Figure (\ref{fig:x2_extinct_y2_0}). The predators are present/introduced only in the prairie strip and move into the second patch, that is, the crop field, via dispersal. As mentioned earlier, a classic example of the kinetic term $f(u)$ in population dynamics is $f(u)=b u(1-\frac{u}{K})$, where $b$ is a birth rate and $K$ the carrying capacity of species $u$. Here $b, K$ are pure constants. However, a more realistic assumption is that of changing environmental pressures - driven by climate change. Thus, future work could lead to modeling $K=K(t)$, as possibly time dependent. 
Also, the higher dimension of the model makes it unpragmatic to have a closed-form parametric value-based expression for the equilibrium population, so numerical evidence is given to validate our quantitative studies. Total pest eradication from both patches, considering both dispersal and drift, is possible but requires strict parametric equality via Lemma \ref{lem:E2_existence_unkidirectional} and Lemma \ref{lem:E3_existence}. This is highly unlikely from a biological point of view.

Furthermore, several other kinds of dynamics have been uncovered for the proposed model. It is sensitive to initial conditions and can give rise to stable limit cycles. The stable limit cycles occur due to a Hopf bifurcation, via which a stable interior equilibrium loses stability, resulting in the occurrence of a stable limit cycle. What is also seen is a further period doubling process, which, under certain parametric windows, gives rise to chaotic dynamics. Chaos is studied numerically and validated by testing the Lyapunov coefficient. It is seen that for a fixed parameter set, some specific amount of additional food can give rise to chaotic behavior in the proposed model.
The benefits of the model as an effective bio-control tool are studied by comparing it with the classical Rosenzweig-MacArthur prey-predator model and the classical AF model with Holling Type II functional response. 
In each case, we compare the pest population density to conclude that the total pest population is lower in the two-patch AF system than in a classical predator-pest (essentially one-patch) system. In particular, the pest population in the AF patch driven system can be less than that in the classical bio-control system for several parameter sets where interior equilibrium is stable for a choice of dispersal rate of predators.

In conclusion, an AF patch-driven prey-pest model can result in an effective bio-control strategy. The drift/dispersal of the predator population between patches can result in pest extinction in \emph{certain} patches - in particular, in our setting, it can do so in the crop field, which is the main patch of interest as far as pest density goes. This essentially results in successful bio-control. Also, the patch model can keep the total prey population lower than in models without patch structure, thus having the benefit of checking the pest population using predator dispersal. 
We note there are several open questions at this juncture. Both prey and predator population dispersal could be considered. Furthermore, non-linear or non-local drift, say akin to drift by strong winds, could be considered. Also, one can consider the spatially explicit patch model, where various shapes of the prairie strip could be modeled. This is part of the larger question of what shape/density of resource distribution is optimal. There is recent progress in these directions in spatial ecology, in the single species case \cite{MR21, MR22}, but the question of this ``optimality" in systems (such as in our setting) where the optimality would be measured in keeping pest density to minimum levels, remains open. Furthermore, one could consider explicit pest dynamics in the crop field if one were attempting bio-control of certain specific target pests, such as the soybean aphid. In such a setting, depending on the type of host plant, resistant varieties, or susceptible varieties, the dynamics of the pest could change \cite{BOP22}. Such considerations would be important in a bio-control context of the soybean aphid. Dynamics could also be affected by weather events such as drought or flood. This would affect both host plant quality and the herbivore dynamics \cite{phdthesisAB}. Current and future investigations also include the benefit of parasitoid bio-control over predator based bio-control \cite{phdthesisAB}.

All in all, many of these current and future approaches could result in interesting predator-pest dynamics. Field experiments are encouraged, which can show the power of patch structure as a bio-control strategy in agricultural practice today. These and related questions continue to be the subject of our current investigations. 

\section{Appendix}
\subsection{Proof of Lemma \ref{lem:E1_existence_unkidirectional}}
\begin{proof}
From the nullcline  $ \Dot{y_2}=0$ we have,

 \begin{equation} 
     y_2^* = \frac{q_2 y_1^*}{\delta_2 } 
     \label{x2_0_eq_1_uni}
 \end{equation}
 Using the nullcline $ \Dot{x_1}=0$,
  
\begin{equation*}
    x_1^*\ \left \{ \left (1-\frac{x_1^*}{k_p} \right)-\frac{ y_1^*}{1+x_1^*+\alpha \xi} \right \} = 0
\end{equation*}
either $ x_1^* = 0$ or,
\begin{equation}
  y_1^* =  \left(1-\frac{x_1^*}{k_p}\right) (1+x_1^*+\alpha \xi) 
    \label{x2_0_eq_2_uni}
\end{equation}

Now, using the nullcline  $ \Dot{y_1}=0$, 

\begin{equation*}
   y_1^*\ \left \{ \epsilon_1 \left( \frac{x_1+\xi}{1+x_1+\alpha \xi} \right ) - \delta_1  - q_1  \right \} = 0
\end{equation*}
$\therefore$ either $y_1^* = 0$ or,

\begin{equation}
   \epsilon_1 \left ( \frac{x_1^* +\xi}{1+x_1^* +\alpha \xi} \right) - \delta_1  - q_1  = 0
    \label{x2_0_eq_3_uni}
\end{equation}

simplifying \eqref{x2_0_eq_3_uni} gives  the expression of $x_1^*$ in terms of parameters, 
 
 \begin{equation}
     x_1^* = \frac{(\delta_1+q_1)(1+\alpha\xi)-\epsilon_1\xi  }{\epsilon_1 - (\delta_1+q_1)} 
      \label{x2_0_eq_4_uni}
 \end{equation}
 
Thus, the equilibrium point $\hat{E_1}= (x_1^*,y_1^*,0,y_2^*)$ exists if  $\xi < \dfrac{\delta_1+q_1}{\epsilon_1 - \alpha (\delta_1+q_1)}$ and {$ \epsilon_1 > (\delta_1+q_1) $}.
\label{crop_extinctio_uni_existence}
\end{proof}

\subsection{Proof of Lemma \ref{lem:E_stability_uni} }
\begin{proof}
\label{lem proof:E_stability_uni}
Using equations \eqref{x2_0_eq_1_uni} and \eqref{x2_0_eq_2_uni}, the general Jacobian matrix  \eqref{general_jacobian_k1_k3_0_uni} at $(x_1^*,y_1^*,0,y_2^*)$ becomes, 
\vspace{0.25cm}

\begin{equation*} 
\hat{J_1} = \begin{bmatrix}
\dfrac{ x_1^* y_1^*} {(1+x_1^* +\alpha \xi)^2}  - \dfrac{ x_1^*}{k_p} & \dfrac{- x_1^*}{1+x_1^* +\alpha \xi} & 0 & 0 \vspace{0.25cm}
  \\ 
\dfrac{\epsilon_1  (1 + (\alpha - 1) \xi) \ y_1^*}{(1+x_1^* +\alpha \xi)^2} & 0 & 0 & 0 \vspace{0.25cm}
\\
0 & 0 & 1  - y_2^* & 0 \vspace{0.25cm}
  \\
0 & q_2 & \epsilon_2 y_2^* &  - \delta_2 \\
\end{bmatrix}
\end{equation*} 
\vspace{0.25cm}

The characteristic polynomial comes out as:
\begin{equation*}
\left(1- y_2^* - \lambda \right) \left[\left( \dfrac{ x_1^* y_1^*} {(1+x_1^* +\alpha \xi)^2} - \dfrac{ x_1^*}{k_p}  - \lambda\right) \left(\lambda^2 +  \lambda \delta_2 \right) - \dfrac{ x_1^*\left(\delta_2 + \lambda\right)}{1+x_1^* +\alpha \xi} \left(\dfrac{\epsilon_1  (1 + (\alpha - 1) \xi) \ y_1^*}{(1+x_1^* +\alpha \xi)^2}\right) \right] 
\end{equation*}
\begin{equation*}
\text{Let, }  D = 1-y_2^*, \ A = \dfrac{ x_1^* y_1^*} {(1+x_1^* +\alpha \xi)^2}  - \dfrac{ x_1^*}{k_p} , \ B = \dfrac{ - x_1^*}{1+x_1^* +\alpha \xi}, \ C = \dfrac{\epsilon_1  (1 + (\alpha - 1) \xi) \ y_1^*}{(1+x_1^* +\alpha \xi)^2},\text{  where } B  < 0 
\end{equation*}
Now the characteristic equation becomes,
\begin{equation*}
\left(D- \lambda \right)   \{ \lambda^3 + \lambda^2(\delta_2 - A)+ \lambda(-A\delta_2-BC) - BC\delta_2 \}= 0
\end{equation*}
To satisfy the Routh–Hurwitz stability criteria for all negative roots, we should have the following conditions:
\begin{equation}
\label{R-H for x2_0_unidirectional}
1-y_2^*>0, \delta_2>A,\ A\delta_2+BC<0,\ BC\delta_2<0 \ \& \  (\delta_2-A)(A\delta_2+BC)<BC\delta_2
\end{equation}

Thus, the lemma is proved. 
\end{proof}

\subsection{Proof of Lemma \ref{lem:E3_existence_unkidirectional}}
\begin{proof}
    \label{proof_E3_existence_drift}
    From the nullcline $ \dot y_1 = 0$ either we have $y_1^* = 0$ or 
    \begin{equation}
    \label{af_eq_1}
        \left( \frac{\xi}{1+\alpha \xi}\right)= \dfrac{\delta_1 + q_1}{\epsilon_1}
    \end{equation}
  From $\dot{x_2} = 0$,  
  either $x_2^* = 0$ or 
   \begin{equation}
     y_2^* =  \left (1- \frac{x_2^*}{k_c} \right) \left(1+x_2^*\right)
     \label{af_eq_3}
   \end{equation}
and with $\Dot{y_2} = 0 $ we have,
   \begin{equation}
      y_2^* \left \{\epsilon_2 \left ( \frac{x_2^*}{1+x_2^*} \right ) - \delta_2 \right \} + q_2 y_1^*  = 0 
      \label{af_eqn_4}
      \end{equation}
    The value $ y_2^*$ is,     
\begin{equation}
     y_2^* = \dfrac{q_2 y_1^* }{\delta_2 - \epsilon_2 \left ( \dfrac{x_2^*}{1+x_2^*} \right )  }
     \label{af_eq_5}
\end{equation}
For $y_2^*$ to exist, we need the condition on 
$x_2^*$ as, 
$x_2^* < \dfrac{\delta_2}{\epsilon_2-\delta_2}$. To find an expression for $x_2^*$, one can equate \eqref{int_uni_eq_3} and \eqref{int_uni_eq_4} and solve the resulting quadratic equation in $x_2^*$.

Thus, the equilibrium point $\hat{E_2}$ exists if,  $\xi = \dfrac{\delta_1 + q_1 }{\epsilon_1- \alpha(\delta_1 + q_1 )}$, $ \epsilon_1 > (\delta_1+q_1) $ and, $y_1^* > \dfrac{\delta_2}{q_2}$. 
\end{proof}

 \subsection{Proof of Lemma \ref{stability_pest_ext_af_drift}}
\begin{proof}
\label{stability_pest_ext_af_drift_proof}
 Using equations \eqref{af_eq_1}, \eqref{af_eq_3} and \eqref{af_eq_5}, the general Jacobian matrix  \eqref{general_jacobian_k1_k3_0_uni} at $(0,y_1^*,x_2^*,y_2^*)$ becomes, 
\vspace{0.25cm}
\begin{equation*} 
\hat{J_3} = \begin{bmatrix}
1- \dfrac{y_1^*} {(1 +\alpha \xi)}   & 0  & 0 & 0 \vspace{0.25cm}
  \\ 
\dfrac{\epsilon_1  (1 + (\alpha - 1) \xi) \ y_1^*}{(1+x_1^* +\alpha \xi)^2} & 0 & 0 & 0 \vspace{0.25cm}
\\
0 & 0 & \dfrac{x_2^*y_2^*}{(1+x_2^*)^2} - \dfrac{x_2^*}{k_c} & \dfrac{- x_2^*}{1+x_2^*} \vspace{0.25cm}
  \\
0 & q_2 &  \dfrac{\epsilon_2 y_2^*}{(1+x_2^*)^2} &  - \dfrac{q_2y_1^*}{y_2^*}  \\
\end{bmatrix}
\end{equation*} 

Let, $A = \dfrac{x_2^*y_2^*}{(1+x_2^*)^2} - \dfrac{x_2^*}{k_c} ,\ B = \dfrac{- x_2^*}{1+x_2^*} \ C= \dfrac{\epsilon_2 y_2^*}{(1+x_2^*)^2}, \ D = - \dfrac{q_2y_1^*}{y_2^*}$
then the characteristic equation becomes,
\begin{equation}
\lambda  \left((A- \lambda) (D-\lambda) - BC\right) \left(  1- \dfrac{y_1^*} {1 +\alpha \xi} - \lambda \right) = 0 
\end{equation}

Since one of the eigenvalues is zero, this proves the lemma.
\end{proof}

\subsection{Proof of Lemma \ref{lem:E2_existence_unkidirectional}}

\begin{proof}
\label{proof_E2_existence_unkidirectional}
From the nullcline $ \Dot{y_2}=0$ we have,
\begin{equation*}
    - \delta_2 y_2^* + q_2 y_1^*   = 0 
\end{equation*}
\begin{equation}
y_2^* = \dfrac{q_2 y_1^* }{\delta_2}
\label{y2_pest_extinction}
\end{equation}
From the nullcline $\Dot{y_1}=0$ we have, 
\begin{equation*}
y_1^*\ \left \{ \epsilon_1 \left( \frac{\xi}{1+\alpha \xi} \right ) - \delta_1  - q_1  \right \} = 0
\end{equation*}
\begin{equation*}
    \epsilon_1 \left( \frac{\xi}{1+\alpha \xi} \right ) = \delta_1 + q_1 
\end{equation*}
\begin{equation}
    \xi = \dfrac{\delta_1 + q_1 }{\epsilon_1- \alpha(\delta_1 + q_1 )}
\end{equation}

Thus, the equilibrium point $\hat{E_3} = (0,y_1^*,0,y_2^*)$ exists if  $\xi = \dfrac{\delta_1 + q_1 }{\epsilon_1- \alpha(\delta_1 + q_1 )}$ and {$ \epsilon_1 > (\delta_1+q_1) $}.
\end{proof}

\subsection{Proof of Lemma \ref{stability_pest_ext_drift}}
\begin{proof}
\label{proof:stability_pest_ext_drift}
Using \eqref{y2_pest_extinction}, the general Jacobian matrix  \eqref{general_jacobian_k1_k3_0_uni} at $(0,y_1^*,0,y_2^*)$ becomes, 

\begin{equation*} 
\hat{J_2} = \begin{bmatrix}
1- \dfrac{y_1^*} {(1 +\alpha \xi)}   & 0  & 0 & 0 \vspace{0.25cm}
  \\ 
\dfrac{\epsilon_1  (1 + (\alpha - 1) \xi) \ y_1^*}{(1+x_1^* +\alpha \xi)^2} & 0 & 0 & 0 \vspace{0.25cm}
\\
0 & 0 & 1  - y_2^* & 0 \vspace{0.25cm}
  \\
0 & q_2 & \epsilon_2 y_2^* &  - \delta_2 \\
\end{bmatrix}
\end{equation*} 
Now the characteristic equation becomes,
\begin{equation*}
 \lambda \left( 1- y_2^* - \lambda \right) (\delta_2+\lambda) \left(1- \dfrac{y_1^*} {1 +\alpha \xi} - \lambda \right)  = 0
 \end{equation*}
 \begin{equation*}
\lambda_1 = 0, \hspace{0.2cm}  \lambda_2 =  1  - y_2^* ,  \hspace{0.2cm}  \lambda_3 = 1  -  \dfrac{y_1^* } {1 +\alpha \xi}, \hspace{0.2cm} \lambda_4  = - \delta_2  
\end{equation*}
Since one of the eigenvalues is zero, this proves the lemma.
\end{proof}

\subsection{Proof of Lemma \ref{lem:co_existence}}
\begin{proof}
\label{proof_co_existence_uni}
From the nullcline $ \dot x_1 = 0$ either we have $x_1^\star = 0$ or 

\begin{equation}
  \hspace{0.2cm }  y_1^\star = \left (1-\frac{x_1^\star}{k_p}\right) \left(1+x_1^\star+\alpha \xi \right) 
  \label{int_uni_eq_1}
\end{equation}
from $\dot y_1 = 0$
\begin{equation*}
   y_1^\star\ \left \{ \epsilon_1 \left( \frac{x_1^\star+\xi}{1+x_1^\star+\alpha \xi}\right) - \delta_1  - q_1 \right \}  = 0
\end{equation*}
which gives either $ y_1^\star = 0$ or,
\begin{equation}
\label{int_uni_eq_2}
    \left( \frac{x_1^\star+\xi}{1+x_1^\star+\alpha \xi}\right)= \dfrac{\delta_1 + q_1}{\epsilon_1}
\end{equation}
simplifying \eqref{int_uni_eq_2} gives the values of $x_1^\star$ in terms of parameters, 

\begin{equation}
    x_1^\star = \frac{(\delta_1+q_1)(1+\alpha\xi)-\epsilon_1\xi  }{\epsilon_1 - (\delta_1+q_1)} 
\end{equation}
from $\dot x_2 = 0$
either $x_2^\star = 0$ or 
   \begin{equation}
     y_2^\star =  \left (1- \frac{x_2^\star}{k_c} \right) \left(1+x_2^\star\right)
     \label{int_uni_eq_3}
   \end{equation}
and with $\Dot{y_2} = 0 $ we have,
   \begin{equation}
      y_2^\star \left \{\epsilon_2 \left ( \frac{x_2^\star}{1+x_2^\star} \right ) - \delta_2 \right \} + q_2 y_1^\star  = 0 
      \end{equation}
We have the value $ y_2^\star$ as,     
\begin{equation}
     y_2^\star = \dfrac{q_2 y_1^\star }{\delta_2 - \epsilon_2 \left ( \dfrac{x_2^\star}{1+x_2^\star} \right )  }
     \label{int_uni_eq_4}
\end{equation}
For $y_2^\star$ to exist, we need the condition on 
$x_2^\star$ as, 
$x_2^\star < \dfrac{\delta_2}{\epsilon_2-\delta_2}$. To find an expression for $x_2^\star$, one can equate \eqref{int_uni_eq_3} and \eqref{int_uni_eq_4} and solve the resulting quadratic equation in $x_2^\star$. 

The quadratic equation in $x_2^\star$ is,

\begin{equation*}
{x_2^\star}^2 + x_2^\star\left( -k_c - \frac{\delta_2}{\epsilon_2-\delta_2}\right) +\dfrac{ k_c(\delta_2  - q_2 y_1^\star )}{\epsilon_2-\delta_2} = 0
\end{equation*}


For a unique positive $x_2^\star$ we should have, 
\begin{equation*}
\left(k_c + \frac{\delta_2}{\epsilon_2-\delta_2}\right)^2 - 4 \left(\dfrac{ k_c(\delta_2  - q_2 y_1^\star )}{\epsilon_2-\delta_2}\right) \geq 0, \    y_1^\star  > \dfrac{\delta_2}{q_2}
\end{equation*}
Thus, the equilibrium point $\hat{E_4}= (x_1^\star,y_1^\star,x_2^\star,y_2^\star)$ exists if  $\xi < \dfrac{\delta_1+q_1}{\epsilon_1 - \alpha (\delta_1+q_1)}$, {$ \epsilon_1 > (\delta_1+q_1) $} and, $y_1^\star  > \dfrac{\delta_2}{q_2}$.  
\end{proof}

\subsection{Proof of Lemma \ref{lem:loc_uni_coexistence}}
 \begin{proof}
 \label{loc_uni_coexistence_proof}
 Using equations \eqref{int_uni_eq_1}, \eqref{int_uni_eq_2} and \eqref{int_uni_eq_4}, the general Jacobian matrix  \eqref{general_jacobian_k1_k3_0_uni} at $(x_1^\star,y_1^\star,x_2^\star,y_2^\star)$ becomes, 
\vspace{0.25cm}
\begin{equation*} 
\hat{J_4} = \begin{bmatrix}
\dfrac{ x_1^\star y_1^\star} {(1+x_1^\star +\alpha \xi)^2}  - \dfrac{ x_1^\star}{k_p} & \dfrac{- x_1^\star}{1+x_1^\star +\alpha \xi} & 0 & 0 \vspace{0.25cm}
  \\ 
\dfrac{\epsilon_1  (1 + (\alpha - 1) \xi) \ y_1^\star}{(1+x_1^\star +\alpha \xi)^2} & 0 & 0 & 0 \vspace{0.25cm}
\\
0 & 0 & \dfrac{x_2^\star y_2^\star}{(1+x_2^*\star^2} - \dfrac{x_2^\star}{k_c} & \dfrac{- x_2^\star}{1+x_2^\star}\vspace{0.25cm}
  \\
0 & q_2 & \dfrac{\epsilon_2 y_2^\star}{(1+x_2^\star)^2} &  - \dfrac{q_2y_1^\star}{y_2^\star} \\
\end{bmatrix}
\end{equation*} 
\vspace{0.25cm}

\begin{equation*}
\text{Let, } A = \dfrac{ x_1^\star y_1^\star} {(1+x_1^\star +\alpha \xi)^2}  - \dfrac{ x_1^\star}{k_p} , \ B = \dfrac{- x_1^\star}{1+x_1^\star +\alpha \xi}, \ C = \dfrac{\epsilon_1  (1 + (\alpha - 1) \xi) \ y_1^\star}{(1+x_1^\star +\alpha \xi)^2},\ D = \dfrac{x_2^\star y_2^\star}{(1+x_2^\star)^2} - \dfrac{x_2^\star}{k_c},
\end{equation*}

\begin{equation*}
\ E = \dfrac{- x_2^\star}{1+x_2^\star}, \ F = \dfrac{\epsilon_2 y_2^\star}{(1+x_2^\star)^2}, \ G =   - \dfrac{q_2y_1^\star}{y_2^\star} \text{  where } B,E<0, F >0 
\end{equation*}
Now the characteristic equation becomes,
\begin{equation*}
\setlength{\jot}{9pt}
\begin{aligned}
& (A-\lambda) (-\lambda)\left[ (D-\lambda)(G-\lambda)-EF \right] -BC \left[(D-\lambda)(G-\lambda)-EF \right] = 0 \\
& \implies \left[\lambda^2 - \lambda A - BC\right]    \left[(D-\lambda)(G-\lambda)-EF \right] = 0\\
& \implies \lambda^4 - \lambda^3 (D + G + A) + \lambda^2 (DG - EF + A(D+G) - BC)\\
& \hspace{0.6cm} - \lambda\{A(DG-EF) -BC(D+G)\}- BC(DG- EF) = 0 \\
&\implies  \lambda^4 + A_3 \lambda^3 + A_2 \lambda^2 + A_1 \lambda +A_0 = 0 \\
& \text{Where,} \ A_3 = - (D + G + A) , A_2 =  (DG - EF + A(D+G) - BC),\\
& A_1 =  - \{A(DG-EF)-BC(D+G)\}, A_0 = - BC(DG- EF)\\
&\text{To satisfy the Routh–Hurwitz stability criteria for all negative roots, we should have the following conditions:}\\
& A_3 > 0,  A_2 > 0 , A_1 > 0, A_0 > 0, \text{ and }  A_3A_2A_1 > A_1^2 + A_3^2A_0.\nonumber
\end{aligned}
\end{equation*}
\begin{equation} 
\setlength{\jot}{10pt}
\begin{aligned}
&\text{for} \ A_0 > 0,  - BC(DG- EF)>0 \implies BC(DG- EF)<0\\
&\text{for} \ A_1 > 0, - \{A(DG-EF)-BC(D+G)\} >0, \implies A(DG-EF)-BC(D+G) < 0, \\
&\text{for} \ A_2 > 0, (DG - EF + A(D+G) - BC) >0 \\
&\text{for} \ A_3 > 0, - (D + G + A) > 0 \implies  (D + G + A)<0, \\
&\text{for }\   A_3A_2A_1 > A_1^2 + A_3^2A_0 \text{ we have,}\\
& \implies (D + G + A) (DG - EF + A(D+G) - BC)\{A(DG-EF)-BC(D+G)\}\\
& \hspace{1cm}>\{A(DG-EF)-BC(D+G)\}^2 - BC(DG- EF)(D + G + A)^2\\
\end{aligned}
\label{R-H for coexistence_uni}
\end{equation}
Thus, the lemma is proved.
 \end{proof}

\subsection{Proof of lemma \ref{coexistence_existence_super}}

\begin{proof}

From the nullcline  $ \Dot{\tilde{x_1}}=0$ we have,

\begin{equation}
     \tilde{y^*_1} =  \left(1-\frac{\tilde{x^*_1}}{k_p}\right) (1+\tilde{x^*_1}+\alpha \xi) 
    \label{x1_y1_eq_1_uni}
\end{equation}

Now, using the nullcline $ \Dot{\tilde{y_1}}=0$,
\begin{equation*}
   \tilde{y^*_1} \left\{ \epsilon_1 \left ( \frac{\tilde{x^*_1}}{1+\tilde{x^*_1}+\alpha \xi} \right)\ - \tilde\delta \right\} = 0
\end{equation*}

$\therefore \ \text{either} \ \tilde{y^*_1}  = 0$ or, 

\begin{equation}
    \epsilon_1 \left ( \frac{\tilde{x^*_1}}{1+\tilde{x^*_1}+\alpha \xi} \right)\ - \tilde\delta = 0
    \label{exp_for_tildex1*}
\end{equation}

Simplifying \eqref{exp_for_tildex1*} gives $\tilde{x^*_1} $ in terms of parameters,

\begin{equation}
 \tilde{x^*_1}= \dfrac{\tilde{\delta} (1+\alpha\xi)}{\epsilon_1- \tilde{\delta}}
 \label{xy_sup_tildex_value}
\end{equation}

Thus, the equilibrium point $ \tilde{E} = (\tilde {x^*_1}, \tilde {y^*_1})$ exists if, $\epsilon_1 - \tilde{\delta} > 0$
    \label{proof_coexistence_existence_super}
\end{proof}

\subsection{Proof of lemma \ref{coexistence_existence_sub}}

\begin{proof}

From the nullcline  $ \Dot{\bar{x_1}}=0$ we have,

\begin{equation}
     \bar{y^*_1} =  \left(1-\frac{\bar{x^*_1}}{k}\right) (1+\bar{x^*_1}+\alpha \xi) 
    \label{bar_x1_y1_eq_1_uni}
\end{equation}

Now, using the nullcline $ \Dot{\bar{y_1}}=0$,
\begin{equation*}
   \bar{y^*_1}  \left\{ \epsilon_1 \left ( \frac{\bar{x^*_1}}{1+\bar{x^*_1}+\alpha \xi} \right)\ - \bar\delta \right\} = 0
\end{equation*}

$\therefore \ \text{either} \ \bar{y^*_1}  = 0$ or, 

\begin{equation}
    \epsilon_1 \left ( \frac{\bar{x^*_1}}{1+\bar{x^*_1}+\alpha \xi} \right)\ - \bar\delta = 0
    \label{exp_for_barx1*}
\end{equation}

Simplifying \eqref{exp_for_barx1*} gives $\bar{x^*_1} $ in terms of parameters,

\begin{equation*}
 \bar{x^*_1}= \dfrac{\bar{\delta} (1+\alpha\xi)}{\epsilon_1- \bar{\delta}}
\end{equation*}

Thus, the equilibrium point  $\bar{E}= (\bar {x^*_1}, \bar {y^*_1}) \  \text{exists if}  \ \epsilon_1-\bar{\delta} > 0 $ 
\label{proof_coexistence_existence_sub}
\end{proof}

\subsection{Proof of Theorem \ref{thm:t1}}

\begin{proof}
\label{thm proof:t1}
Note the prey populations $x_{1},x_{2}$ are always bounded in comparison to the logistic equation. In the no dispersal case $k_{2}=k_{4}=0$, if $\xi > \xi_{critical} =  \dfrac{\delta_1}{\epsilon_1 -\alpha \delta_1}$, then blow-up in infinite time follows from \cite{PWB23}. Now, in the case of dispersal, we add up the predator equations to obtain,

\begin{eqnarray}
   && \Dot{y_1} + \Dot{y_2} \nonumber \\
   &=& \epsilon_1 \left ( \frac{x_1+\xi}{1+x_1+\alpha \xi} \right)\ y_1 - \delta_1 y_1 + (k_4 + \tilde{k_{2}}\xi) \left (y_2-y_1 \right) \nonumber \\
   &+& \epsilon_2 \left ( \frac{x_2}{1+x_2} \right ) \ y_2 - \delta_2 y_2 + k_4 \left (y_1-y_2  \right) \nonumber \\
\end{eqnarray}

In the case that the predator populations blow-up, their growth rate is exponential at best, and the prey populations will go extinct, this follows via the form of the predator-prey equations. Thus assuming blow-up, we can take $x_{1},x_{2} \rightarrow 0$ in the above, to obtain

\begin{equation}
    \Dot{y_1} + \Dot{y_2} \leq \left(\frac{\epsilon_{1} \xi}{1+ \alpha \xi} -\delta_{1} - \tilde{k_{2}} \xi\right)y_{1} + \left( \tilde{k_{2}} \xi - \delta_{2}\right)y_{2}
\end{equation}

Let us proceed by contradiction. Assume $y_{1},y_{2}$ blow-up in infinite time. Then $(y_{1})^{'},(y_{2})^{'}$ also blow up.






This yields,

\begin{equation}
\Dot{y_1} + \Dot{y_2} + \left(\tilde{k_{2}} \xi -\left(\frac{\epsilon_{1} \xi}{1+ \alpha \xi} -\delta_{1} \right)\right)y_{1} < \left( \tilde{k_{2}} \xi - \delta_{2}\right)y_{2}
\end{equation}

Now, from the parametric assumptions we have, $\delta_{1}+\delta_{2} > \frac{\epsilon_{1} \xi}{1+ \alpha \xi} > \delta_{1}$, Thus $\left(\tilde{k_{2}} \xi -\left(\frac{\epsilon_{1} \xi}{1+ \alpha \xi} -\delta_{1} \right)\right) > \left( \tilde{k_{2}} \xi - \delta_{2}\right).$
Thus, taking limits entails a contradiction.
 Note, If $y_{1}$ blows-up, so must $y_{2}$, and vice versa. This is easily seen from the form of the equations. 
 
 \end{proof}
 
\subsection{Proof of Theorem \ref{thm:t1sd}}

\begin{proof}
\label{thm proof:t1sd}
 Now, in the case of symmetric dispersal, we add up the predator equations to obtain,

\begin{eqnarray}
   && \Dot{y_1} + \Dot{y_2} \nonumber \\
   &=& \epsilon_1 \left ( \frac{x_1+\xi}{1+x_1+\alpha \xi} \right)\ y_1 - \delta_1 y_1 + k_2 \left (y_2-y_1 \right) \nonumber \\
   &+& \epsilon_2 \left ( \frac{x_2}{1+x_2} \right ) \ y_2 - \delta_2 y_2 + k_2 \left (y_1-y_2  \right) \nonumber \\
   &=& \epsilon_1 \left ( \frac{x_1+\xi}{1+x_1+\alpha \xi} \right)\ y_1 + \epsilon_2 \left ( \frac{x_2}{1+x_2} \right ) \ y_2 - \delta_1 y_1 - \delta_2 y_2 \nonumber \\
\end{eqnarray}

Note, blow-up in $y_{1}$, that is $\lim_{t \rightarrow T^{*} < \infty} y_{1}(t) \rightarrow \infty <=> x_{1} \rightarrow 0$.
This follows via standard theory, applied to \eqref{eq:patch_model2pp} - in that if $x_{1}$ went to a positive equilibrium, or cycled, $y_{1}$ could not blow-up, via the prey equation for $x_{1}$. Thus we only need to consider the case when $x_{1} \rightarrow 0$ in the estimate above. So, letting $x_{1} \rightarrow 0$, in the above we obtain,

\begin{equation}
\Dot{y_1} + \Dot{y_2} \leq \epsilon_1 \left ( \frac{\xi}{1+\alpha \xi} \right)\ y_1 + \epsilon_2 \left ( \frac{x_2}{1+x_2} \right ) \ y_2
- \delta_1 y_1 - \delta_2 y_2
\end{equation}

setting $U = y_1 +y_2 $, we obtain,

\begin{equation}
\Dot{U} + \left(\delta_1 + \delta_2 - \frac{\xi}{1+\alpha \xi}\right) U 
\leq G(t)
\end{equation}
Where $G = \epsilon_2 \left ( \dfrac{x_2}{1+x_2} \right ) \ y_2 $, since $x_{2}, y_{2}$ are bounded by standard predator-prey theory, so is $G$. Thus the bound on $U$ follows via Gronwall's lemma,

\begin{equation}
U(t) \leq e^{-\left( \delta_1 + \delta_2 - \frac{\xi}{1+\alpha \xi}\right)t}U(0) + e^{-\left( \delta_1 + \delta_2 - \frac{\xi}{1+\alpha \xi}\right)t}\int^{t}_{0} e^{\left( \delta_1 + \delta_2 - \frac{\xi}{1+\alpha \xi}\right)s}G(s)ds, \ \forall t > 0.
\end{equation}

The above then yields the bound on $y_{1}$ using positivity. This proves the theorem.
\end{proof}

\subsection{Proof of Lemma \ref{lem:E1_existence}}
\begin{proof}
From the nullcline  $ \Dot{y_2}=0$ we have,

 \begin{equation} 
     y_2^* = \frac{k_4}{\delta_2 + k_4} \ y_1^*
     \label{x2_0_eq_1}
 \end{equation}
 Using the nullcline $ \Dot{x_1}=0$,
  
\begin{equation*}
    x_1^*\ \left \{ \left (1-\frac{x_1^*}{k} \right)-\frac{ y_1^*}{1+x_1^*+\alpha \xi} \right \} = 0
\end{equation*}
either $ x_1^* = 0$ or,
\begin{equation}
  y_1^* =  \left(1-\frac{x_1^*}{k}\right) (1+x_1^*+\alpha \xi) 
  \label{x2_0_eq_2}
\end{equation}

Now, using the nullcline  $ \Dot{y_1}=0$, and substituting the value of $y_2^*$ from \eqref{x2_0_eq_1} we have,

\begin{equation*}
   y_1^*\ \left \{ \epsilon_1 \left( \frac{x_1^*+\xi}{1+x_1^*+\alpha \xi} \right ) - \delta_1  - k_2 +  \frac{ k_2 k_4}{\delta_2 + k_4}  \right \} = 0
\end{equation*}
$\therefore$ either $y_1^* = 0$ or,

\begin{equation}
   \epsilon_1 \left ( \frac{x_1^* +\xi}{1+x_1^* +\alpha \xi} \right) - \delta_1  - k_2 +  \frac{ k_2 k_4}{\delta_2 + k_4} = 0
    \label{x2_0_eq_3}
\end{equation}
simplifying \eqref{x2_0_eq_3} gives  the expression of $x_1^*$ in terms of parameters, 
 
 \begin{equation}
     x_1^* = \frac{Q \ \left (1+\alpha \xi \right ) - \epsilon_1 \xi }{\epsilon_1 - Q} 
      \label{x2_0_eq_4}
 \end{equation}
where $Q$ is defined by,

 \begin{equation*}
    Q = \frac{\delta_1 \delta_2  + \delta_1 k_4  + k_2 \delta_2  }{\delta_2 + k_4 } 
 \end{equation*}
Thus, the equilibrium point $E_1 = (x_1^*,y_1^*,0,y_2^*)$ exists if  $\xi < \dfrac{Q}{\epsilon_1 - \alpha Q}$ and {$ \epsilon_1 > Q $}.
     \label{proof_E1_existence}
\end{proof}

\subsection{Proof of Lemma \ref{lem:E1_stability} }
\begin{proof}
\label{lem proof:E1_stability}
Using equations \eqref{x2_0_eq_1} and \eqref{x2_0_eq_2}, the general Jacobian matrix  \eqref{general_jacobian_k1_k3_0} at $(x_1^*,y_1^*,0,y_2^*)$ becomes,
\begin{equation*} 
J_1 = \begin{bmatrix}
\dfrac{ x_1^* y_1^*} {(1+x_1^* +\alpha \xi)^2}  - \dfrac{ x_1^*}{k} & \dfrac{- x_1^*}{1+x_1^* +\alpha \xi} & 0 & 0 \vspace{0.25cm}
  \\ 
\dfrac{\epsilon_1  (1 + (\alpha - 1) \xi) \ y_1^*}{(1+x_1^* +\alpha \xi)^2} &  -\dfrac{ k_2 k_4}{\delta_2 + k_4} & 0 & k_2 \vspace{0.25cm}
\\
0 & 0 & 1  - y_2^* & 0 \vspace{0.25cm}
  \\
0 & k_4 & \epsilon_2 y_2^* &  - \delta_2 - k_4
\end{bmatrix}
\end{equation*} 
\begin{equation*}
\setlength{\jot}{10pt}
\begin{aligned}
&\text{The characteristic polynomial comes out as:}\\
&=\left(1- y_2^* - \lambda \right) \left[\left( \dfrac{ x_1^* y_1^*} {(1+x_1^* +\alpha \xi)^2}  - \dfrac{ x_1^*}{k} - \lambda\right) \left(\lambda^2 +  \lambda \left (\dfrac{k_2 y_2^*}{y_1^*} + \dfrac{k_4 y_1^*}{y_2^*} \right )\right) \right] \\ 
&+ \left(1- y_2^* - \lambda \right) \left[\left (\dfrac{ x_1^*}{1+x_1^* +\alpha \xi}       \right)\left( \dfrac{-k_4 y_1^*}{y_2^*} - \lambda \right) \left(  \dfrac{\epsilon_1  (1 + (\alpha - 1) \xi) \ y_1^*}{(1+x_1^* +\alpha \xi)^2} \right)  \right]\\ 
& \text{Let, }     D= 1-y_2^*, \ A = \dfrac{ x_1^* y_1^*} {(1+x_1^* +\alpha \xi)^2}  - \dfrac{ x_1^*}{k} , \ B = \dfrac{ x_1^*}{1+x_1^* +\alpha \xi}, \ C = \dfrac{\epsilon_1  (1 + (\alpha - 1) \xi) \ y_1^*}{(1+x_1^* +\alpha \xi)^2},\\
&\ P = \dfrac{k_2 y_2^*}{y_1^*} + \dfrac{k_4 y_1^*}{y_2^*},\ E = \dfrac{k_4 y_1^*}{y_2^*} \text{  where } B, P, E > 0 \\
&\text{Now the characteristic equation becomes,} \\
&\implies\left(D- \lambda \right)   \{ \left(A- \lambda \right)   \left( \lambda^2 + P \lambda \right) + BC   \left(-E - \lambda \right) \}= 0\\
&\implies \lambda^4 - \lambda^3\left( A - P + D \right) +  \lambda^2 \left( D(A - P) - (AP-BC)  \right) + \lambda \left( D(AP - BC) + BCE  \right) - BCDE = 0 \nonumber \\
&\implies \lambda^4 +A_3 \lambda^3 + A_2 \lambda^2 + A_1 \lambda + A_0  = 0  \\
&\text{where, }A_3 = - \left( A - P + D \right), \ A_2 =  D(A - P) - (AP-BC),\  A_1 =  D(AP - BC) +BCE ,  \ A_0 = -BCDE\\
&\text{To satisfy the Routh–Hurwitz stability criteria for all negative roots, we should have the following conditions:}\\
& A_3 > 0,  A_2 > 0 , A_1 > 0, A_0 > 0, \text{ and }  A_3A_2A_1 > A_1^2 + A_3^2A_0.\nonumber
\end{aligned}
\end{equation*}
\begin{equation} 
\setlength{\jot}{10pt}
\begin{aligned}
&\text{for} A_0 > 0 , -BCDE > 0 \implies BCDE<0 
\text{ as }  B,E > 0 , C>0 \text{ (due to feasibility condition)} \implies D<0\\
&\text{for }  A_1 > 0, D(AP - BC) +BCE >0\\
&\text{for }  A_2 >0, D(A - P) - (AP-BC) >0 \implies D(A-P) > AP-BC\\
&\text{for }  A_3 > 0,  - \left( A - P + D \right) >0 \implies  \left( A - P + D \right) <0\\
&\text{for }  A_3A_2A_1 > A_1^2 + A_3^2A_0 \text{ we have,}\\
&\implies- \left( A - P + D \right) \left( D(A - P) - (AP-BC) \right) \left( D(AP - BC) +BCE \right) \\
 & \hspace{1cm} >\left( D(AP - BC) +BCE  \right)^2 - \left( A - P + D \right)^2 BCDE\\
&\implies(A-P)(AP-BC)(D^3 + D^2(A-P)-D(AP-BC)-BCE)  \\
 & \hspace{1cm} < BCE (D^3 + D^2(A-P)-D(AP-BC)- BCE)\\
&\because \ (D^3 + D^2(A-P)-D(AP-BC)-BCE) < 0\implies (A-P)(AP-BC) > BCE
\end{aligned}
\label{R-H for x2_0}
\end{equation}

Thus, the lemma is proved. 
\end{proof}
\subsection{Proof of Lemma \ref{lem:E2_existence}}
\begin{proof}
Using the equation $\Dot{y_1} = 0 $ then,
\begin{equation*}
 y_1^* \left \{\epsilon_1 \left ( \frac{\xi}{1+\alpha \xi} \right)\  - \delta_1  - k_2 \right \}  + k_2 \ y_2^* = 0 ,
\end{equation*}

\begin{equation}
 y_2^* = \frac{-y_1^*}{k_2} \left \{\epsilon_1 \left ( \frac{\xi}{1+\alpha \xi} \right)\  - \delta_1 - k_2 \right \} 
 \label{x1_0_eq_1}
\end{equation}
    from $\Dot{x_2} = 0 $,
   \begin{equation*}
        x_2^* \left (1- \frac{x_2^*}{k} \right)-\frac{x_2^*  y_2^*}{1+x_2^*} = 0      
   \end{equation*}
   either $x_2^* = 0$ or 
   \begin{equation}
     y_2^*  =  \left (1- \frac{x_2^*}{k} \right) \left(1+x_2^*\right) 
     \label{x1_0_eq_2}
   \end{equation}
   and from $\Dot{y_2} = 0 $ then,
   \begin{equation}
    y_1^* = \frac{- y_2^*}{k_4}\left \{\epsilon_2 \left ( \frac{x_2^*}{1+x_2^*} \right ) - \delta_2 - k_4    \right \}
    \label{x1_0_eq_3}
\end{equation}
substituting the value of $y_1^*$ from \eqref{x1_0_eq_1},
\begin{equation*}
    y_2^* \left \{ k_2 - \frac{1}{k_4} \left(  \frac{ \epsilon_2 x_2^*}{1+x_2^*}   - \delta_2 - k_4    \right) \left(   \frac{\epsilon_1 \xi}{1+\alpha \xi} \  - \delta_1 - k_2   \right)\right \} = 0 
\end{equation*}
either $y_2^* = 0 $ or \begin{equation*}
    \left\{ k_2 - \dfrac{1}{k_4} \left(  \dfrac{ \epsilon_2 x_2^*}{1+x_2^*}   - \delta_2 - k_4    \right) \left(   \dfrac{\epsilon_1 \xi}{1+\alpha \xi} \  - \delta_1 - k_2   \right)\right \}= 0 
\end{equation*}
which gives the value of $x_2^*$ in terms of parameters only,
\begin{equation}
    x_2^* = \dfrac{\tilde{Q}}{\epsilon_2 - \tilde{Q}} 
    \label{x1_0_eq_4}
\end{equation}
where $ \tilde{Q}$ is define by,
\begin{equation*}
 \tilde{Q} =    \delta_2 + k_4 + \dfrac{k_2 \ k_4}{\left(   \dfrac{\epsilon_1 \xi}{1+\alpha \xi} \  - \delta_1 - k_2   \right) }
\end{equation*}

Thus, the equilibrium point $E_2 = (0,y_1^*,x_2^*,y_2^*)$ exists if $\epsilon_2 - \tilde{Q} > 0 \ \& \ \tilde{Q} > 0 $.
\label{proof_E2_existence}
\end{proof}
\subsection{Proof of Lemma \ref{lem:E2_stability}}
 \begin{proof}
  \label{lem proof:E2_stability}   
 
Using equations \eqref{x1_0_eq_1}, \eqref{x1_0_eq_2} and \eqref{x1_0_eq_3},  the general Jacobian matrix  \eqref{general_jacobian_k1_k3_0} at  $(0,y_1^*,x_2^*,y_2^*)$ becomes, 
 
\begin{equation*} 
J_2 = \begin{bmatrix}
1 -   \dfrac{y_1^*}   {\left(1+\alpha \xi\right)} & 0 & 0 & 0 \vspace{0.25cm}
  \\ 
\dfrac{\epsilon_1  \left(1 + \left(\alpha - 1\right) \xi\right) \ y_1^*}{(1+\alpha \xi)^2} &   \dfrac{- k_2 y_2^*}{y_1^*}& 0 & k_2 \vspace{0.25cm}\\
0 & 0 & \dfrac{x_2^* y_2^*}{\left(1+x_2^*\right)^2}    - \dfrac{x_2^*}{k} & 
  \dfrac{- x_2^*}{1+x_2^*} \vspace{0.25cm}\\
  0 & k_4 & \dfrac{\epsilon_2 y_2^*}{(1+x_2^*)^2} & \dfrac{-k_4 y_1^*}{y_2^*}
\end{bmatrix}
 \end{equation*}
 \begin{equation*}
\setlength{\jot}{10pt}
\begin{aligned}
&\text{The characteristic equation is given by:}\\
& \left( 1 - \dfrac{y_1^*}{1+\alpha \xi} - \lambda \right)\left[\left( \dfrac{x_2^* y_2^*}{\left(1+x_2^*\right)^2}    - \dfrac{x_2^*}{k} - \lambda\right) \left(\lambda^2 +  \lambda \left (\dfrac{k_2 y_2^*}{y_1^*} + \dfrac{k_4 y_1^*}{y_2^*} \right )\right) + \left(\dfrac{-k_2 y_2^*}{y_1^*} - \lambda\right) \left(\dfrac{\epsilon_2 x_2^* y_2^*}{(1+x_2^*)^3}   \right) \right] = 0\\
& \text{Let, }D =  \left( 1 - \dfrac{y_1^*}{1+\alpha \xi}\right),\hspace{0.2cm} A =\dfrac{x_2^* y_2^*}{\left(1+x_2^*\right)^2}    - \dfrac{x_2^*}{k}, \hspace{0.25cm}
C = \dfrac{\epsilon_2 x_2^* y_2^*}{(1+x_2^*)^3} , \hspace{0.25cm} P = \dfrac{k_2 y_2^*}{y_1^*} + \dfrac{k_4 y_1^*}{y_2^*},  \hspace{0.25cm} E = \dfrac{k_2 y_2^*}{y_1^*}\\
& \text{  where } C, P, E > 0 \\
&\text{Now the characteristic equation becomes,} \ (D- \lambda) \left[\left( A - \lambda\right) \left( \lambda^2 +  P \lambda \right) + \left(- E - \lambda\right) C \right] = 0 \\
& \implies \lambda^{4}    - \lambda^{3}\left(A-P+D \right)  + \lambda^{2}\left( D(A - P) - (AP-C)  \right) + \lambda \left( D(AP - C) + CE  \right) - CDE = 0 \\
&\implies \lambda^4 +A_3 \lambda^3 + A_2 \lambda^2 + A_1 \lambda + A_0  = 0 \\
& \text{where,} \ A_3 = - \left( A - P + D \right), A_2 = \left(D(A - P) - (AP-C)\right) , A_1 = \left( D(AP - C) +CE  \right), A_0 = -CDE\\
&\text{To satisfy the Routh–Hurwitz stability criteria for all negative roots, we should have the following conditions:}\\
& A_3 > 0,  A_2 > 0 , A_1 > 0, A_0 > 0, \text{ and }  A_3A_2A_1 > A_1^2 + A_3^2A_0.\nonumber
\end{aligned}
\end{equation*}
\begin{equation} \label{RH for x1_0}
\setlength{\jot}{10pt}
     \begin{aligned}
&\text{  for } A_0 > 0,  -CDE > 0 \implies CDE<0 \text{ as } E,C > 0 \implies D<0  \\
&\text{  for } A_1 > 0, D(AP - C) + CE >0 \\ &\text{  for }  A_2 >0 , D(A - P) - (AP-C) >0 \implies D(A-P) > AP-C\\
&\text{ for } A_3 > 0,
 - \left( A - P + D \right) >0 \implies  \left( A - P + D \right) <0 \\
&\text{ for } A_3A_2A_1 > A_1^2 + A_3^2A_0 \text{ we have},\\
&- \left( A - P + D \right) \left( D(A - P) - (AP-C) \right) \left( D(AP - C) +CE \right) > \left( D(AP - C) +CE  \right)^2 - \left( A - P + D \right)^2 CDE\\
&\implies (A-P)(AP-C)(D^3 + D^2(A-P)-D(AP-C)-CE) < CE (D^3 + D^2(A-P)-D(AP-C)-CE)\\
&\because  (D^3 + D^2(A-P)-D(AP-C)-CE) < 0\implies(A-P)(AP-C) > CE
    \end{aligned}
 \end{equation}
 Thus, the lemma is proved. 
\end{proof}
\subsection{Proof of Lemma \ref{lem:E3_existence}}
\begin{proof}
From the nullcline $ \Dot{y_2}=0$ we have,

 \begin{equation}
     y_2^* = \dfrac{k_4}{\delta_2 + k_4} \ y_1^*
     \label{x1_x2_0_eq1}
 \end{equation}
and from $\Dot{y_1} = 0 $,
\begin{equation*}
 y_1^* \left\{\epsilon_1 \left ( \frac{\xi}{1+\alpha \xi} \right)\  - \delta_1  - k_2 \right \} + k_2 \ y_2^* = 0 ,
\end{equation*}
Using the value of $y_2^*$ from \eqref{x1_x2_0_eq1},

\begin{equation*}
 y_1^* \left\{\epsilon_1 \left ( \frac{\xi}{1+\alpha \xi} \right)\  - \delta_1  - k_2 + \dfrac{k_2 k_4}{\delta_2 + k_4} \right\}   = 0 ,
\end{equation*}
either $y_1^* = 0 $ or,
\begin{equation}
\left\{\epsilon_1 \left ( \frac{\xi}{1+\alpha \xi} \right)\  - \delta_1  - k_2 + \dfrac{k_2 k_4}{\delta_2 + k_4} \right\}   = 0
\label{x1_x2_0_eq2}
\end{equation}
Simplifying \eqref{x1_x2_0_eq2} for $\xi$ gives,

\begin{equation*}
    \xi = \dfrac{{Q}}{\epsilon_1 - \alpha {Q}} \text{ given,} \ \epsilon_1 - \alpha {Q} > 0
\end{equation*}
where $ {Q}$ is defined by:

\begin{equation*}
    Q = \frac{\delta_1 \delta_2  + \delta_1 k_4  + k_2 \delta_2  }{\delta_2 + k_4 } 
\end{equation*}
Therefore, the equilibrium point $ E_3 =(0,y_1^*,0,y_2^*)$ exists if $\xi = \dfrac{{Q}}{\epsilon_1 - \alpha {Q}}$ and without additional food  $ E_3 $ doesn't exists. 
\label{proof_E3_existence}
\end{proof}
\subsection{Proof of Lemma \ref{lem:E3_stability}}
\begin{proof}
Using equation \eqref{x1_x2_0_eq2}, the general Jacobian matrix  \eqref{general_jacobian_k1_k3_0} at $(0,y_1^*,0,y_2^*)$  we have,
\begin{equation*} 
J_3 = \begin{bmatrix}
1  -  \dfrac{y_1^* } {(1 +\alpha \xi)} & 0 & 0 & 0 \vspace{0.25cm}
  \\ 
\dfrac{\epsilon_1  (1 + (\alpha - 1) \xi) \ y_1^*}{(1 +\alpha \xi)^2} &   \dfrac{- k_2 k_4}{\delta_2 + k_4} & 0 & k_2 \vspace{0.25cm}
\\
0 & 0 & 1  - y_2^* & 0 \vspace{0.25cm}
  \\
  0 & k_4 & \epsilon_2 y_2^* &  - \delta_2 - k_4
   \end{bmatrix}
 \end{equation*}
The characteristic equation is given by:
\begin{equation*}
\setlength{\jot}{10pt}
\begin{aligned}
& \lambda  \left(1  -  \dfrac{y_1^* } {1 +\alpha \xi} - \lambda  \right) \left(1  - y_2^* - \lambda \right) 
 \left( \lambda  + \left(\dfrac{ k_2 k_4}{\delta_2 + k_4} + \delta_2 + k_4 \right)\right) = 0\\
 & \lambda_1 = 0, \hspace{0.2cm}  \lambda_2 =  1  - y_2^* ,  \hspace{0.2cm}  \lambda_3 = 1  -  \dfrac{y_1^* } {1 +\alpha \xi}, \hspace{0.2cm} \lambda_4  = - \left(\dfrac{ k_2 k_4}{\delta_2 + k_4} + \delta_2 + k_4\right)\\
  \end{aligned}
\end{equation*}
Since one of the eigenvalues is zero, this proves the lemma. 
\label{lem proof:E3_stability}
\end{proof}

\subsection{Proof of Lemma \ref{lem:E4_existence}}
\begin{proof}
From the nullcline $ \dot x_1 = 0$ either we have $x_1^\star = 0$ or 

\begin{equation}
  \hspace{0.2cm }  y_1^\star = \left (1-\frac{x_1^\star}{k}\right) \left(1+x_1^\star+\alpha \xi \right) 
  \label{int_eq_1}
\end{equation}
from $\dot y_1 = 0$
\begin{equation}
   y_1^\star\ \left \{ \epsilon_1 \left( \frac{x_1^\star+\xi}{1+x_1^\star+\alpha \xi}\right) - \delta_1  - k_2 \right \} + k_2 y_2^\star = 0
   \label{int_eq_2}
\end{equation}
from $\dot x_2 = 0$
either $x_2^\star = 0$ or 
   \begin{equation}
     y_2^\star =  \left (1- \frac{x_2^\star}{k} \right) \left(1+x_2^\star\right)
     \label{int_eq_3}
   \end{equation}
and with $\Dot{y_2} = 0 $ we have,
   \begin{equation}
      y_2^\star \left\{\epsilon_2 \left ( \frac{x_2^\star}{1+x_2^\star} \right ) - \delta_2 - k_4    \right\} + k_4 y_1^\star  = 0
      \label{int_eq_4}
\end{equation}
From $\eqref{int_eq_1}$ and $\eqref{int_eq_2}$ we have,
\begin{equation}
    y_2^\star = \frac{-\left(1-\dfrac{x_1^\star}{k}\right)}{k_2} \left \{ \epsilon_1 ( x_1^\star+\xi)- (\delta_1  + k_2 ) (1+x_1^\star+\alpha \xi)\right \} 
    \label{int_eq_5}
\end{equation}
For $y_2^\star$ to exist, we have the condition on $x_1^\star$, 
\begin{equation*}
     \epsilon_1 ( x_1^\star+\xi) < (\delta_1  + k_2 ) (1+x_1^\star+\alpha \xi)
\end{equation*}
if $\epsilon_1 > \delta_1 + k_2$ then,
\begin{equation*}
    x_1^\star < \frac{\delta_1 + k_2}{\epsilon_1 - (\delta_1 + k_2)} - \alpha \xi
\end{equation*}
From $\eqref{int_eq_3}$ and $\eqref{int_eq_4}$ we have,
\begin{equation}
     y_1^\star = \frac{-\left(1-\dfrac{x_2^\star}{k}\right)}{k_4} \left \{ \epsilon_2 x_2^\star - (\delta_2  + k_4 ) (1+x_2^\star)\right \}  
     \label{int_eq_6}
\end{equation}
For $y_1^\star$ to exist, we have the condition on $x_2^\star$,

\begin{equation*}
     \epsilon_2 x_2^\star < (\delta_2  + k_4 ) (1+x_2^\star)
     \label{int_eq_7}
\end{equation*}
if $\epsilon_2 > \delta_2 + k_4$ then,
\begin{equation*}
    x_2^\star < \frac{\delta_2 + k_4}{\epsilon_2 - (\delta_2 + k_4)} 
    \label{int_eq_8}
\end{equation*}
Using $\eqref{int_eq_1}$ and $\eqref{int_eq_6}$ after eliminating $y_1^\star$ we have, 

\begin{equation}
   k_4  \left (1-\frac{x_1^\star}{k}  \right) \left(1+x_1^\star+\alpha \xi \right)  = -(1-\dfrac{x_2^\star}{k}) \left \{ \epsilon_2 x_2^\star - (\delta_2  + k_4 ) (1+x_2^\star)\right \}  
   \label{x1_x2_implicit_eq_1}
\end{equation}
Using $\eqref{int_eq_3}$ and $\eqref{int_eq_5}$ after eliminating $y_2^\star$ we have,

\begin{equation}
    k_2 \left (1- \frac{x_2^\star}{k} \right) \left(1+x_2^\star\right) = -\left(1-\dfrac{x_1^\star}{k}\right) \left \{ \epsilon_1 ( x_1^\star+\xi)- (\delta_1  + k_2 ) (1+x_1^\star+\alpha \xi)\right \} 
    \label{x1_x2_implicit_eq_2}
\end{equation}
To find expression for $x_2^\star$, we equate $\eqref{x1_x2_implicit_eq_1}$ and $\eqref{x1_x2_implicit_eq_2}$ then, 

\begin{equation}
   k_2 k_4 \left(1+x_2^\star\right) \left(1+x_1^\star+\alpha \xi \right) =  \left \{ \epsilon_2 x_2^\star - (\delta_2  + k_4 ) (1+x_2^\star)\right \} \left \{ \epsilon_1 ( x_1^\star+\xi)- (\delta_1  + k_2 ) (1+x_1^\star+\alpha \xi)\right \} 
   \label{equating_x1_x2_conditions}
\end{equation}

\begin{equation*}
\implies
 \dfrac{k_2 k_4  \left(1+x_1^\star+\alpha \xi \right)} {\left \{ \epsilon_1 ( x_1^\star+\xi)- (\delta_1  + k_2 ) (1+x_1^\star+\alpha \xi)\right \} } = \dfrac{\left \{ \epsilon_2 x_2^\star - (\delta_2  + k_4 ) (1+x_2^\star)\right \} }{ \left(1+x_2^\star\right)}
\end{equation*}
\

\begin{equation*}
\implies
 \dfrac{(\epsilon_2 - (\delta_2 +k_4)) \left \{ \epsilon_1 ( x_1^\star+\xi)- (\delta_1  + k_2 ) (1+x_1^\star+\alpha \xi)\right \} - k_2 k_4  \left(1+x_1^\star+\alpha \xi \right) }{ \epsilon_2 \left \{ \epsilon_1 ( x_1^\star+\xi)- (\delta_1  + k_2 ) (1+x_1^\star+\alpha \xi)\right \}}   = \dfrac{1}{1+x_2^\star}
\end{equation*}
\

\begin{equation*}
    x_2^\star = f(x_1^\star)=\dfrac{\epsilon_1(\delta_2 +k_4)(x_1^\star + \xi) - (1+x_1^\star+\alpha \xi)(\delta_1 \delta_2 + k_2 \delta_2 + k_4 \delta_1)}{(\epsilon_2 - (\delta_2 +k_4)) \left \{ \epsilon_1 ( x_1^\star+\xi)- (\delta_1  + k_2 ) (1+x_1^\star+\alpha \xi)\right \} - k_2 k_4  \left(1+x_1^\star+\alpha \xi \right)}
\end{equation*}

\

The expression for $x_1^\star$ can be found by solving the four nullclines, \eqref{int_eq_1}, \eqref{int_eq_2}, \eqref{int_eq_3} and, \eqref{int_eq_4}.
\label{proof_E4_existence}
\end{proof}

\subsection{Proof of Lemma \ref{lem:E4_stability}}
\begin{proof}
\label{lem proof:E4_stability}
 The general Jacobian matrix \eqref{general_jacobian_k1_k3_0} at $(x_1^\star,y_1^\star,x_2^\star, y_2^\star)$ becomes,
 \begin{equation*} 
J_4 = \begin{bmatrix}
J_{11} & J_{12} & 0 & 0 
  \\ 
J_{21}& J_{22} & 0 & J_{24}
\\
0 & 0 & J_{33}& 
 J_{34}
  \\
  0 & J_{42}  & J_{43}  & J_{44}
 \end{bmatrix}
 \end{equation*}
\begin{equation*}
\setlength{\jot}{10pt}
\begin{aligned}
 &  \text{where, } J_{11} = 1 - \dfrac{2 x_1^\star}{k} + \dfrac{  \left(1-\dfrac{x_2^\star}{k}\right) \left(1+ \alpha \xi \right)\left \{ \epsilon_2 x_2^\star - (\delta_2  + k_4 ) (1+x_2^\star)\right \} } { k_4 \left(1+x_1^\star+\alpha \xi\right)^2},   \ J_{12} =  \dfrac{- x_1^\star}{1+x_1^\star +\alpha \xi}\\
 & J_{21} = \dfrac{ - \epsilon_1 \left(1-\dfrac{x_2^\star}{k}\right)  \left(1 + \left(\alpha - 1\right) \xi\right) \ \left \{ \epsilon_2 x_2^\star - (\delta_2  + k_4 ) (1+x_2^\star)\right \} }{k_4 (1+x_1^\star+\alpha \xi)^2}, \ J_{22}= \epsilon_1  \dfrac{\left(x_1^\star + \xi\right)}{1+x_1^\star +\alpha \xi} - \delta_1 - k_2\\
 &J_{24} = k_2, \ J_{33} = 1 - \dfrac{2 x_2^\star}{k} + \dfrac{ \left(1-\dfrac{x_1^\star}{k}\right) 
  \left \{ \epsilon_1 ( x_1^\star +\xi)- (\delta_1  + k_2 ) (1+x_1^\star +\alpha \xi)\right \}  }{k_2 (1+x_2^\star)^2}, \ J_{34} = \dfrac{- x_2^\star}{1+x_2^\star},\\
  &J_{42} = k_4, \ J_{43} =  \dfrac{- \epsilon_2 \left(1-\dfrac{x_1^\star}{k}\right) 
  \left \{ \epsilon_1 ( x_1^\star +\xi)- (\delta_1  + k_2 ) (1+x_1^\star +\alpha \xi)\right \}}{k_2 (1+x_2^\star)^2}  ,\  J_{44} = \dfrac{\epsilon_2 x_2^\star}{1+x_2^\star} - \delta_2 - k_4\\
\end{aligned}
\end{equation*}

\vspace{0.3cm}

The characteristic equation is given by,
$  \lambda^4 + A \lambda^3 + B\lambda^2+ C\lambda +D = 0$ where,

\begin{equation*}
\setlength{\jot}{10pt}
\begin{aligned}
& A= -(J_{11}+J_{22}+J_{33}+J_{44}) ,\hspace{0.05cm} B= J_{11} J_{22} + J_{33}J_{44} + J_{12} J_{21} - (J_{34} J_{43} + J_{24} J_{42}) + (J_{11}+J_{22})(J_{33}+J_{44})  \\
& C = (J_{11} + J_{33})J_{24} J_{42}  -(J_{33}+J_{44})(J_{11} J_{22} +J_{12} J_{21} )- (J_{11}+J_{22})( J_{33} J_{44}-J_{34} J_{43})\\
& D = (J_{11}J_{22} + J_{12} J_{21})( J_{33} J_{44}-J_{34} J_{43}) - J_{11} J_{24} J_{33} J_{42} 
\end{aligned}
\end{equation*}
From the Routh–Hurwitz stability criteria, we should have the following conditions:
\begin{equation}
  A >0, B>0, C>0, D>0, \And  ABC> C^2 + A^2 D
\label{coexis_stability_conditions}
\end{equation}
Thus, the lemma is proved. 
\end{proof}


\subsection{Chaos Simulations} \label{chaos}
The change in system dynamics is studied with a change in parameter $\xi$. We used the parameter set $k=70, \alpha=0.2, \epsilon_1=\epsilon_2=0.4, \delta_1=0.2, \delta_2=0.16,  k_2=0.3, k_4=0.2 \ \text{with I.C.} \ = [10,10,10,10].$  Using this parameter set and taking $\xi=0.7$, we see that $x_2$ goes extinct while all other populations show stable limit cycle oscillations as seen in Figure (\ref{fig:no_chaos_for_high_zi}). When $\xi$ is decreased to $\xi=0.41212$, we see chaotic dynamics in $x_2,y_2,y_1$ populations with aperiodic cycles as seen in time series in Figure (\ref{fig:chaos_time_series}). 
\vspace{0.20cm}

    
The chaotic dynamics are further investigated by studying the maximum Lyapunov exponent with respect to the parameter $\xi$ as seen in Figure (\ref{fig:lyapunov_ex_with_parameter}), and it can be analyzed that we get the maximum Lyapunov exponent as positive for different range of $\xi$ values. Thus, for the described parameter set, we can have chaos in windows of $\xi$ parameter values, one of which is $\xi=0.41212$. Hence, we conclude the presence of chaotic dynamics in the system \eqref{eq:patch_model2}.

\begin{figure}
\begin{subfigure}{.48\textwidth}
   \centering
\includegraphics[width = 9cm]{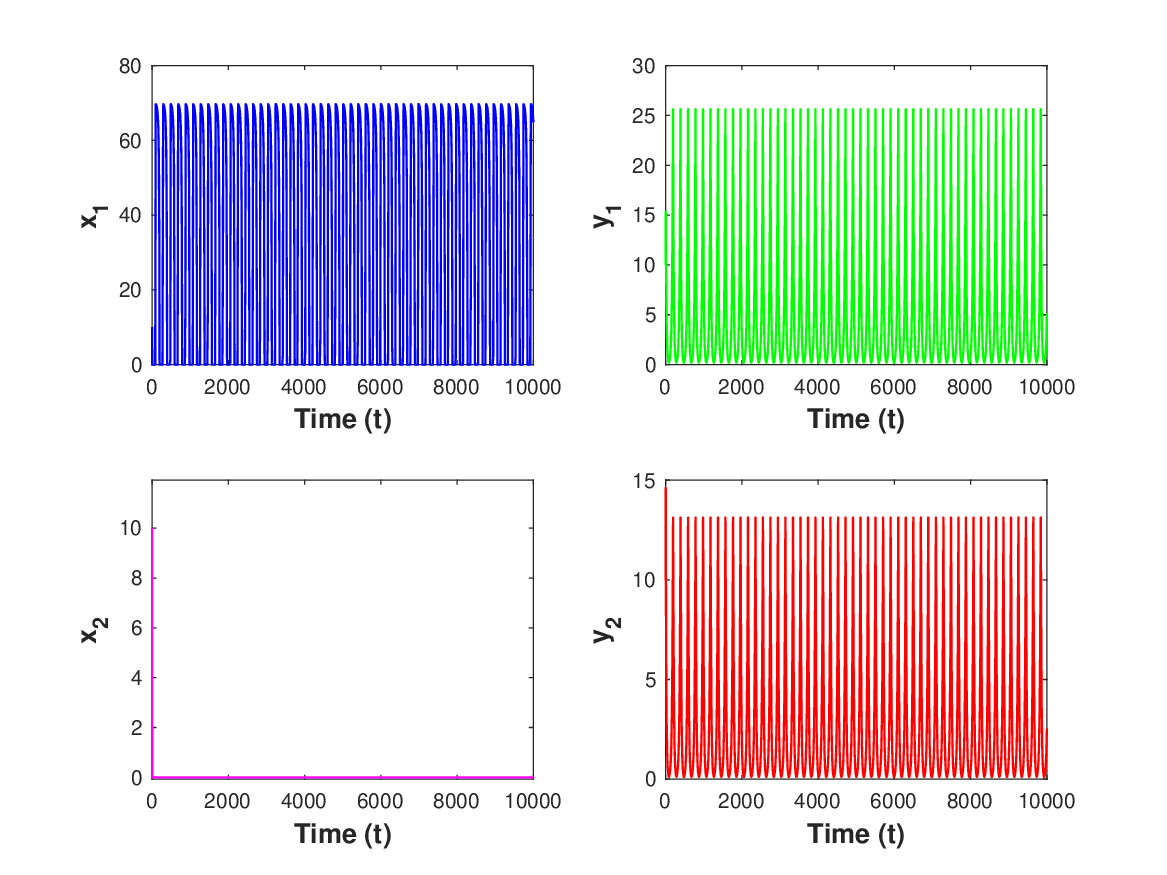}
\subcaption{ $\xi=0.7 $}
\label{fig:no_chaos_for_high_zi}
\end{subfigure}
\begin{subfigure}{.48\textwidth}
\centering
\includegraphics[width = 9cm]{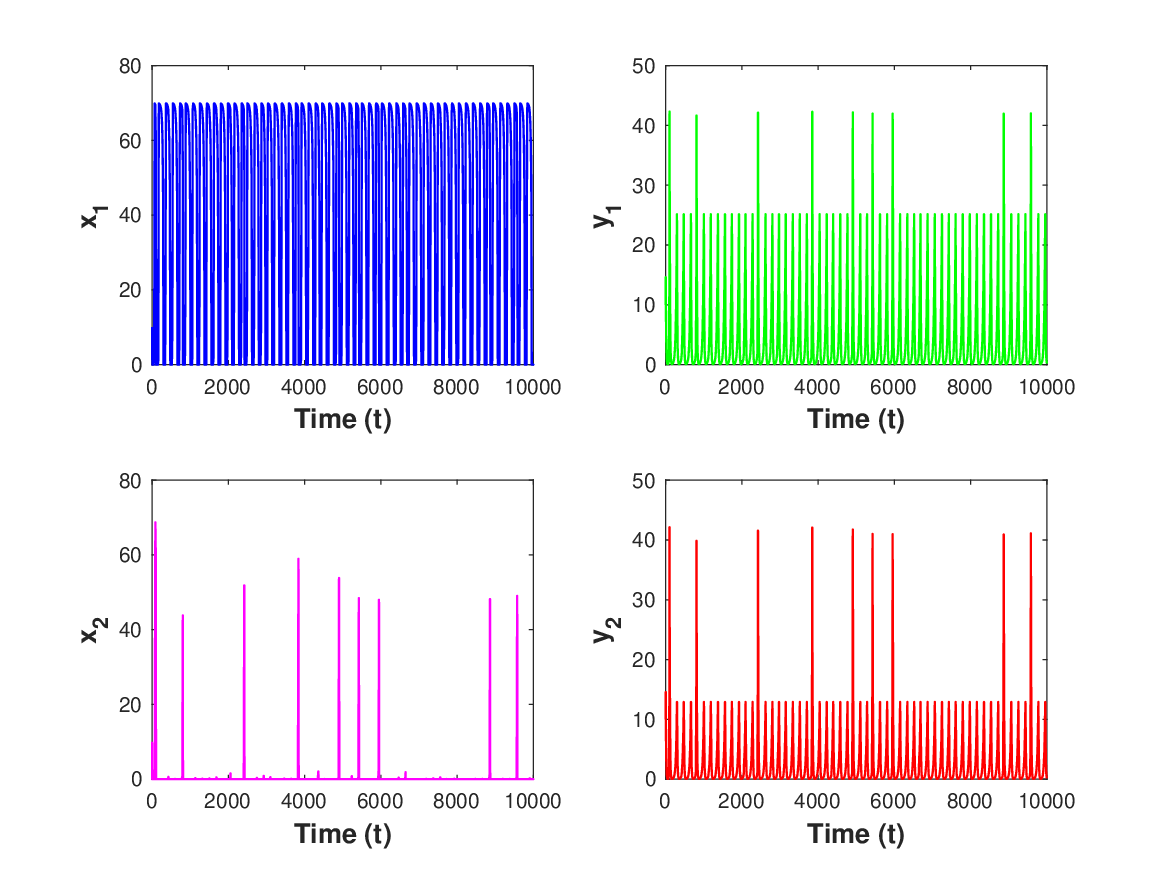}
\subcaption{$\xi=0.41212$}
\label{fig:chaos_time_series}
\end{subfigure}
\caption{ In Figure \eqref{fig:no_chaos_for_high_zi}, we observe the extinction of the pest population in $\Omega_2$, while all other populations display stable limit cycles. Conversely, in Figure \eqref{fig:chaos_time_series}, chaotic dynamics can be seen in the $x_2$ population. 
}
\end{figure}


\begin{figure}[H]
{\includegraphics[width = 9cm,height = 6.65cm]{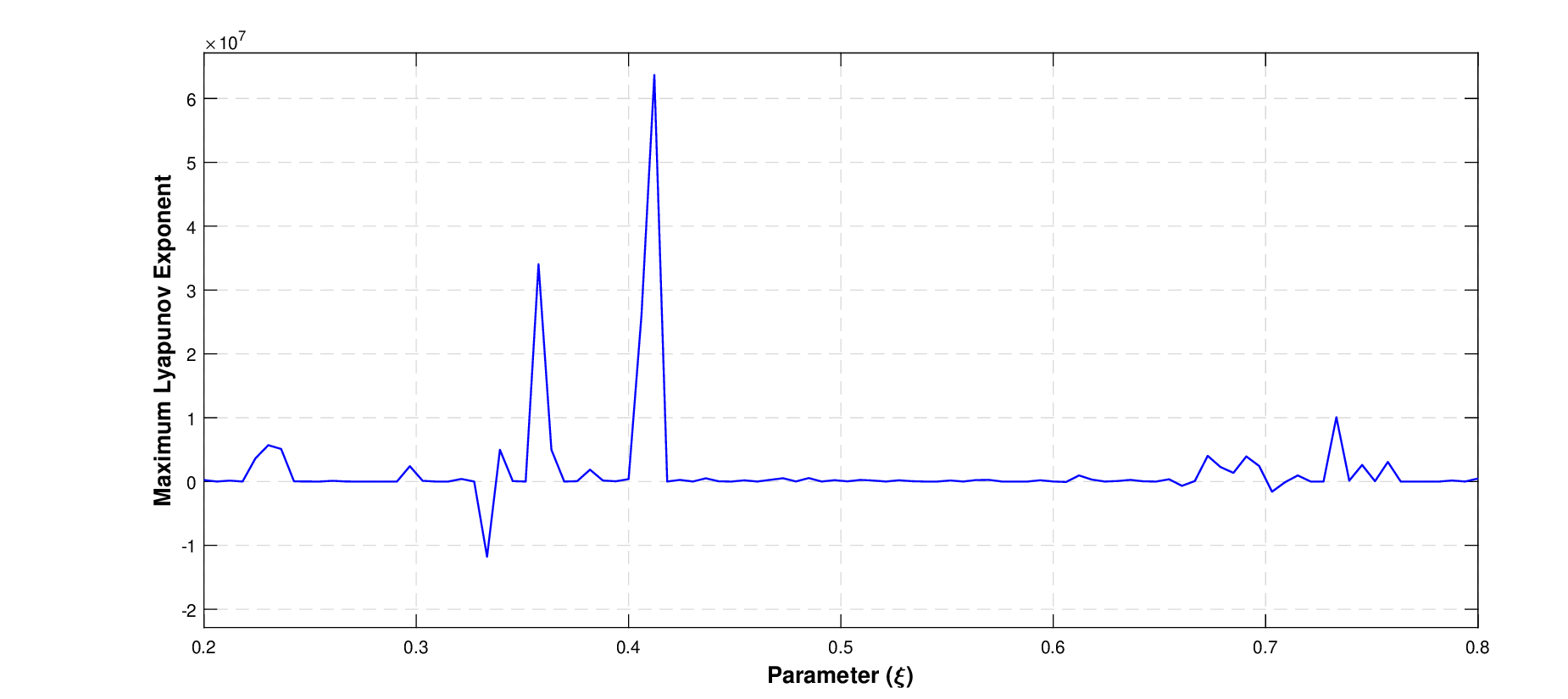}}
\caption{
The graph shows the maximum Lyapunov exponent with respect to the parameter $\xi$. The graph shows chaotic dynamics at some $\xi$ values where the maximum Lyapunov is positive.
}
\label{fig:lyapunov_ex_with_parameter}
\end{figure}







\begin{figure}
\centering
\includegraphics[width = 10cm,height=7cm]{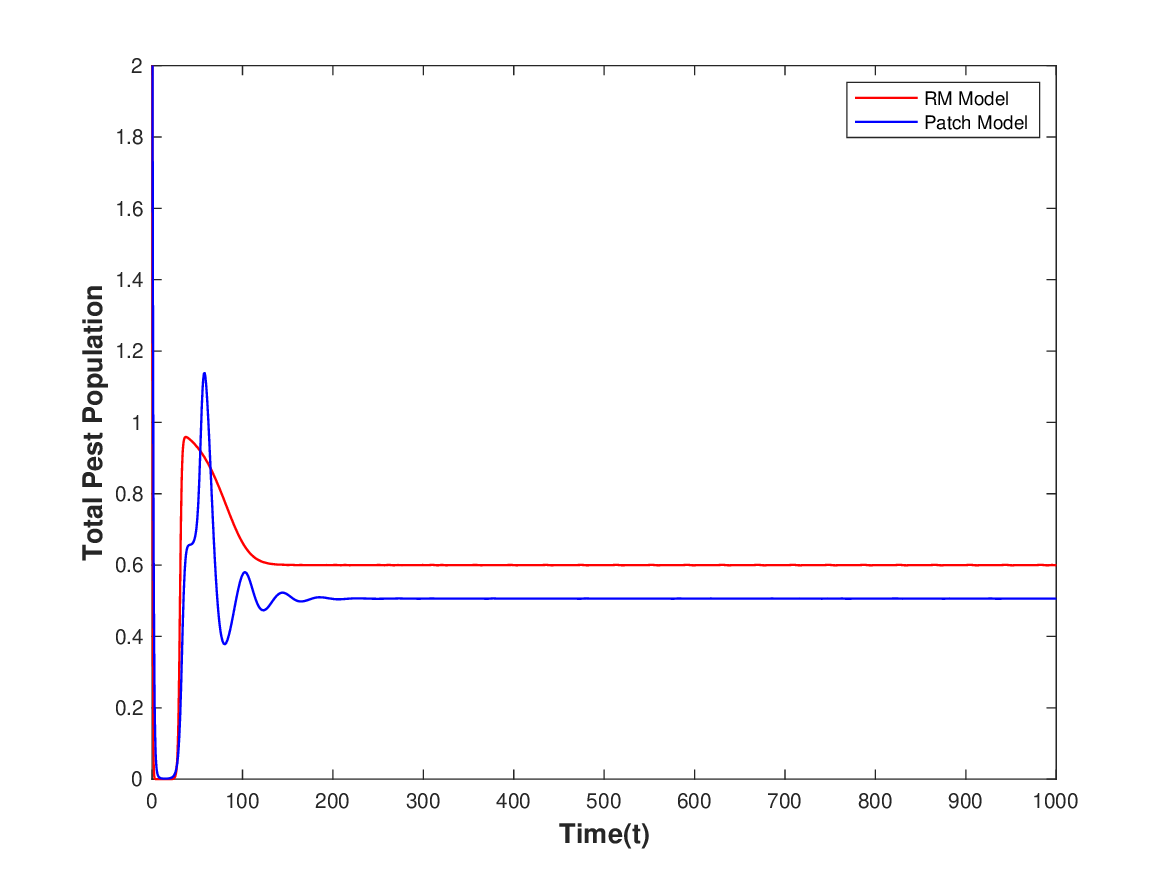}
\caption{Comparison of total pest population: RM model vs. Patch model for interior equilibrium. Here, the parameters are  $k=1, \alpha=0.2, \xi=0.37, \epsilon_1 = \epsilon_2 =0.4, \delta_1=\delta_2=0.15,k_2=0.1,k_4=0.2$. The initial condition of the population for the patch model is $[2,2,2,2]$, and that for the RM model is $[4,4]$. }
\label{fig:int_rm_vs_patch_fig}
\end{figure}

\begin{figure}
\includegraphics[width = 10cm,height=7cm]{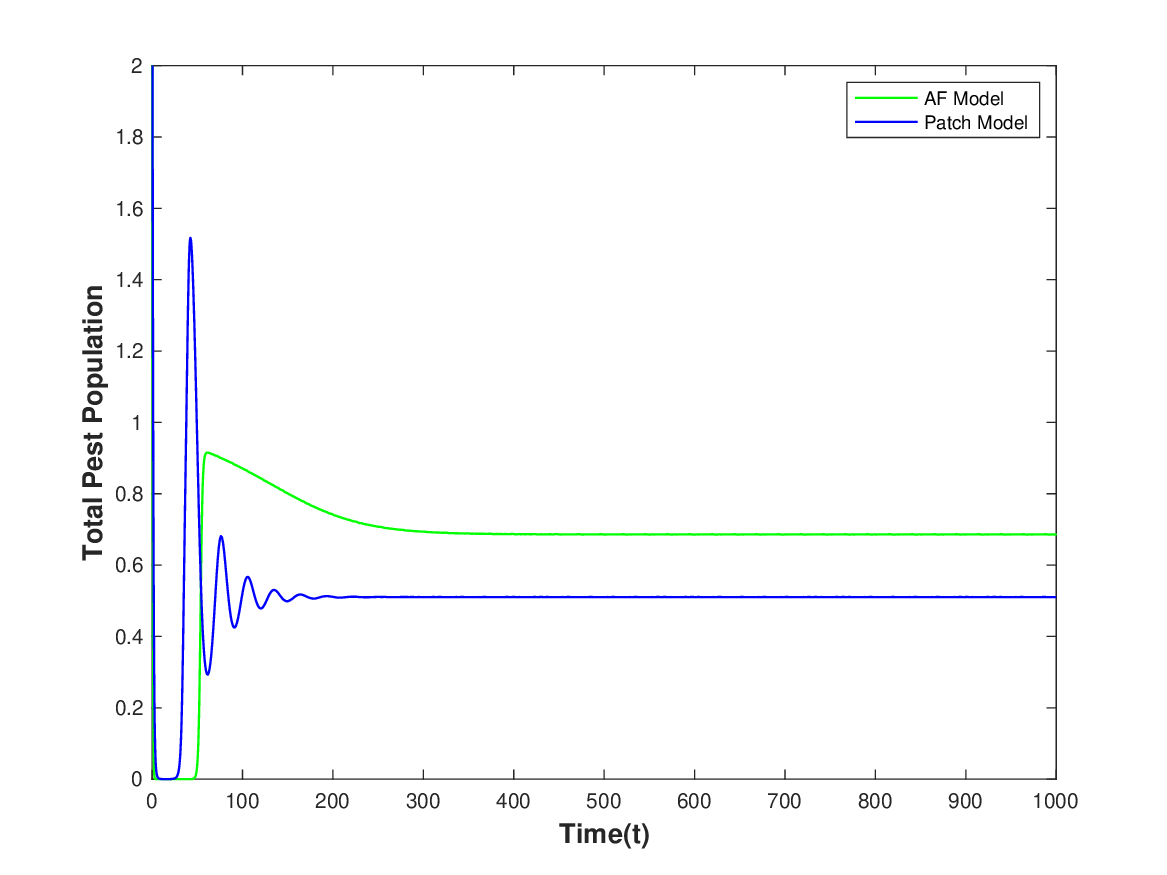}
 \caption{Comparison of total pest population: Classical two-species additional food model vs. Patch model for interior equilibrium. Here, the parameters are $k=1, \alpha=0.2, \xi=0.37, \epsilon_1 =0.25, \epsilon_2 =0.6, \delta_1 = 0.15, \delta_2=0.07,k_2=0.4,k_4=0.2$ with I.C. for patch model $= [2,2,2,2]$ and for classical additional food model is $[4, 4]$. }
\label{fig:int_saf_vs_patch}
\end{figure}

\begin{rem}
\label{int_bio_control}
Figure (\ref{fig:int_rm_vs_patch_fig}), time series shows the comparison of interior equilibrium for \eqref{eq:patch_model2} and \eqref{rmmodel}. We have seen numerically, under some parametric restrictions, \eqref{eq:patch_model2} can exhibit a lower pest population than \eqref{rmmodel}. In this comparison, the patch model's total pest population is $x_1^\star+x_2^\star$, whereas the pest population in the Rosenzweig-MacArthur model is $x_r^*$.
\par  Figure (\ref{fig:int_saf_vs_patch}) demonstrates that, within certain parameter constraints, \eqref{eq:patch_model2} can have a lower pest population than  \eqref{org_model}. In this comparison, the patch model's total pest population is $x_1^\star+x_2^\star$, whereas the pest population in the classical two-species additional food model is $x_h^*$.

\end{rem}

\section{Further prior results on AF}
\label{prior results}
\subsection{Type III response}

The type III response has been considered in a number of AF models, \cite{SPV18, SPD17}. However an adaptation of the results in \cite{PWB20}, enables the following corollary,

\begin{cor}
\label{cor:c1ug}
Consider a predator-pest system described via \eqref{Eqn:1g}, where 
$f(x,\xi,\alpha) = \frac{x^{2}}{1+\alpha \xi^{2} + x^{2}},
  g(x,\xi,\alpha) = \frac{\beta (x^{2}+\xi^{2})}{1+\alpha \xi^{2} + x^{2}}$
and assume the parametric restrictions, $\beta - \delta \alpha >0$. If the quantity of additional food $\xi$ is chosen s.t., 
$\xi \in \left( \sqrt{\frac{\delta}{\beta - \delta \alpha}}, \infty \right)$, then solutions initiating from any positive initial condition, blow up in infinite time. That is,

\begin{equation*}
\lim_{t \rightarrow \infty}{(x(t),y(t))} \rightarrow (0,\infty) 
\end{equation*}
\end{cor}

\subsection{Predator dependent responses and the inhibitory effect}

In order to prevent unbounded predator growth, several alternative mechanisms have been proposed.

\begin{itemize}
\item[(i)] The inhibitory effect produced by prey defense, such as via the type IV functional response \cite{VA22}.
\item[(ii)] The inhibitory effects produced by predator interference, such as via the Beddington-DeAngelis functional response \cite{S02, W23}.
\item[(iii)] The inhibitory effect produced by a purely ratio dependent response \cite{S45, S46}.
\item[(iv)] The damping effect produced by intraspecific predator competition \cite{PWB23}.
\end{itemize}

All of the above prevent the unbounded growth of the predator. In Prasad et al. \cite{S02}, the following model was proposed: \eqref{model1}. 

\begin{equation}\label{model1}
 \frac{dx}{dt} = x\left(1-\frac{x}{k}\right) - \frac{xy}{1+\alpha \xi+x+\epsilon y}, \
       \frac{dy}{dt} =\frac{\beta (x+\xi)y}{1+\alpha \xi+x+\epsilon y}-\delta y
 \end{equation}

Herein, the Beddington–DeAngelis functional response \cite{S33} incorporates mutual interference amongst predators via the term $\epsilon y$ where the parameter $\epsilon $ is a measure of this mutual interference effect.
The analysis in \cite{S02} focused on two separate cases, (i) high interference $\epsilon > 1$ and (ii) low interference $0<\epsilon <1$. In (i), there is always one unique interior equilibrium - if a feasible pest free equilibrium exists, it is a saddle. Thus, from the point of view of pest eradication, high predator interference is not desirable. In that no matter what quality and/or quantity of additional food is chosen, the pest eradication state cannot be (even locally) attracting. In case (ii), it was shown therein that,

\begin{itemize}

\item [(i)] The predator density is always bounded. It cannot grow past a uniform (in parameter) bound, no matter how the parameters are changed, so as to engineer control of the pest.

\item [(ii)] There can exist one unique interior equilibrium, as long as $\xi<\frac{\delta}{\beta-\delta\alpha-\beta\epsilon}$. 

\item [(iii)] A pest free equilibrium can coexist with the interior equilibrium - in that case, the pest free equilibrium is always a saddle.

\item [(iv)] The pest free state is globally stable if $\xi>\frac{\delta}{\beta-\delta\alpha-\beta\epsilon}$. This change of stability occurs through a transcritical bifurcation.

\item[(v)] A predator-free equilibrium can also exist.
\end{itemize}

The analysis in \cite{S02} is based on the geometric assumption that the slope of the predator nullcline (say $m_{1}$) in \eqref{model1} is greater than the slope of the tangent line to the prey nullcline (say $m$) at 
$(0,\frac{1+\alpha \xi}{1-\epsilon})$,  that is where it cuts the predator (y) axis. This is the reason that the authors in \cite{S02} derive global stability of the pest free state when
the condition in (iv) above holds. This condition follows by choosing the $y$-intercept of the predator nullcline to be larger than where the prey nullcline intersects the $y$-axis. This, in conjunction with $m_{1} > m$, yields global stability of the pest free state.
 Thus, the fundamental premise of \cite{S02}, is to assume AF introductions $\xi$, such that $\xi > h(\epsilon)$. Note, recently, we consider $\xi < h(\epsilon)$ and provide a much tighter window on the quantities of AF needed for pest eradication. In particular, we are able to derive pest extinction for $h(\epsilon) > f(\epsilon) > \xi > g(\epsilon)$, \cite{W23}.

\subsection{Constant Predator harvesting}

The basic predator-pest model, with type II response and constant predator harvesting, is as follows \cite{sen2015global}

\begin{equation}
\label{Eqn:1ph}
\frac{dx}{dt} = x\left(1-\frac{x}{\gamma}\right) - \frac{xy}{1+\alpha \xi + x}, \    \frac{dy}{dt} =\frac{\beta xy}{1+\alpha \xi + x} + \frac{\beta \xi y}{1+\alpha \xi + x} - \delta y - \rho.
\end{equation}

Here, the underlying assumption is that the predator is harvested at a constant amount $\rho$. In the case that $\rho = 0$, the model is reduced to the classic model of \cite{SP07}. Dynamically, this model only possesses the Hopf bifurcation. In the case, $\rho > 0$, the model is very rich dynamically, and several higher codimensional bifurcations have been proved.

\begin{itemize}

\item [(i)] Up to 2 interior equilibrium can exist. An axial pest free state can also exist - which is always unstable as a saddle or source. Furthermore finite time extinction of the predator is achievable for certain initial data.

\item [(ii)] Dynamically, the model is very rich. apart from the standard co-dimensional one Hopf and saddle-node bifurcation, it is shown that a global homoclinic bifurcation can occur, as well as a codimension two Bogdanov Takens bifurcation is shown.

\item [(iii)] Despite the rich dynamics, the question of large data solutions is left open. In this case ``large" is data that lies above the stable manifold of the interior saddle equilibrium. It is claimed that the predator can continue to exist on the AF.

\end{itemize}

\printbibliography

\end{document}